\documentclass[12pt]{article}
\usepackage{my-macros}
\usepackage[pagewise]{lineno}

\begin{document}

\title{
Gaussianized Design Optimization for Covariate Balance in Randomized Experiments}

\author{Wenxuan Guo, Tengyuan Liang, and Panos Toulis}
\affil{University of Chicago, Booth School of Business}

\date{}

\maketitle

\onehalfspacing

\begin{abstract}
Achieving covariate balance in randomized experiments enhances the precision of treatment effect estimation. However, existing methods often require heuristic adjustments based on domain knowledge and are primarily developed for binary treatments. This paper presents Gaussianized Design Optimization, a novel framework for optimally balancing covariates in experimental design. The core idea is to Gaussianize the treatment assignments: we model treatments as transformations of random variables drawn from a multivariate Gaussian distribution, converting the design problem into a nonlinear continuous optimization over Gaussian covariance matrices. Compared to existing methods, our approach offers significant flexibility in optimizing covariate balance across a diverse range of designs and covariate types. Adapting the Burer-Monteiro approach for solving semidefinite programs, we introduce first-order local algorithms for optimizing covariate balance, improving upon several widely used designs. Furthermore, we develop inferential procedures for constructing design-based confidence intervals under Gaussianization and extend the framework to accommodate continuous treatments. Simulations demonstrate the effectiveness of Gaussianization in multiple practical scenarios.
\end{abstract}



\section{Introduction}\label{sec:intro}
Randomized experiments are considered the gold standard for causal inference in the study of treatment effects~\citep{imbens2015causal}. Common design choices include completely randomized experiments and independent Bernoulli randomization, which treat experimental units equally and allow for valid estimation of a wide variety of causal quantities. In many real-world experiments, incorporating covariates can enhance balance and improve the precision of treatment effect estimation. Examples of this include matched-pair designs~\citep{fisher1935design, bai2022optimality} and rerandomization methods~\citep{morgan2012rerandomization, li2020rerandomization}. Despite their widespread use, several important design optimization questions related to general covariates and treatments remain underexplored. 

By leveraging concepts from continuous optimization and Gaussian processes, we make progress toward addressing two design optimization questions in this paper: (i) How to numerically optimize for covariate balance in experimental designs with general covariates? (ii) How to balance covariates with multiple treatment arms, or in general settings involving a continuum of treatment arms?

We propose {\em Gaussianized Design Optimization}, a framework for experimental design that transforms the design problem into an optimization problem in an embedded Gaussian space. To illustrate the basic idea, suppose we have $n$ experimental units, each receiving a treatment taking values in a discrete space, say, $D_i\in\mathbb{D}$ for $i=1, \ldots, n$.  Then, Gaussianization refers to the action of modeling treatments $\{D_i\}_{i=1}^n$ as random variables derived from Gaussian vectors,
\begin{equation*}
    D_i = g(T_i)\;,\quad T\coloneqq (T_1, \dots, T_n) \sim \cN(0, \Sigma) \;.
\end{equation*}
Here $g: \R\to \mathbb{D}$ is a pre-specified function that maps the Gaussian variables $T_i$ to the treatment space, and $\Sigma$ is a design matrix from the correlation elliptope, 
\begin{equation}\label{eq:elliptope}
    \cE = \{~\Sigma\in\R^{n\times n} \mid \Sigma\succeq 0, \Sigma_{ii} = 1~\}\;.
\end{equation}

Gaussianization thus transforms the design problem on $\{D_i\}_{i=1}^n$ to a design problem on $\{T_i\}_{i=1}^n$, motivating design optimization in the embedded Gaussian space. Given pre-treatment covariates $X\in\R^{n\times d}$, we propose to solve
\begin{equation}\label{eq:cov_metric}
    \min_{\Sigma\in\cE}\|X^\top f(\Sigma) X\|_{\mathrm{norm}}\;,\quad \mathrm{norm} \in\{ \mathrm{nuc}, \mathrm{op}\}\;.
\end{equation}
Here, $f$ is a function with an analytical expression applied elementwise, defined later in Sections \ref{sec:example} and \ref{sec:discrete}, that controls the aspects of the design that are important for covariate balance. We use $\nuc$ and $\op$ to abbreviate the nuclear norm and operator norm, respectively. The objective $\|X^\top f(\Sigma) X\|_{\mathrm{norm}}$ serves as a surrogate metric to optimize the covariate balance, which we explain in Section \ref{sec:example}. 

From an optimization perspective, Equation \eqref{eq:cov_metric} is a nonlinear optimization problem over the correlation elliptope, and we propose a first-order local algorithm to iteratively update the design $\Sigma$. Due to the complex, non-convex nature of the covariate balance objective, our algorithm is only guaranteed to find local optimizers near the initial point. This local optimality and the computational barrier of design optimization are further discussed in Section \ref{sec:example}.

Gaussianization transforms the design problem into an optimization task, providing a flexible framework for optimizing covariate balance. Importantly, this approach applies directly to any number of treatment arms and any type of covariates. In contrast, most existing research on optimal design focuses on binary treatments \citep{li2020jrssb, Harshaw2019, bai2023randomize}, and certain optimality criteria require additional knowledge about the outcome-generating model \citep{bai2022optimality}.

In certain experiments, the treatment variable is inherently continuous (e.g., a medication dosage), making it insufficient to confine the design to a small number of discrete arms. To address this limitation, 
we further extend the discrete treatment setting by allowing $\mathbb{D} = \R$ and propose Gaussian designs, which directly assign $T = (T_1, \dots, T_n)\sim\cN(0, \Sigma)$ as actual treatments. As shown in Section \ref{sec:continuous}, this approach offers two main advantages. First, it allows the exploration of structural properties of potential outcome functions, including monotonicity and convexity. Second, it enables covariate balance optimization akin to the discrete setup. Thus, Gaussian designs harness the flexibility of Gaussianization and contribute to the growing literature on continuous treatment effects.

In Section \ref{sec:asymp}, we investigate design-based inference under Gaussianization, where the outcome-generating model is fixed, and all randomness arises from the Gaussian treatments $\{T_i\}_{i=1}^n$. Under a local perturbation condition, we establish asymptotic normality for the proposed estimator and present valid inferential procedures. Collectively, we establish a comprehensive framework that integrates design optimization, estimation, and inference under Gaussianization.

\subsection{An Example: Gaussianization with Three Treatment Arms}
\label{sec:example}
To contextualize the idea, we first walk through Gaussianized design optimization with a simple three-treatment example, i.e., $\mathbb{D} = \{1,2,3\}$, supplemented with a numerical simulation. Following the standard potential outcome framework \citep{Neyman1923}, we define $Y_i(k)$ as the potential outcome for unit $i$ under treatment $k$ for $k = 1,2,3$. The observed outcome for unit $i$ is then defined as  $Y_i = \sum_{k=1}^3 \mathbb{I}\{D_i=k\} Y_i(k)$ and $D = (D_1, \dots, D_n)$ is the treatment vector. Let $X\in\R^{n\times d}$ be the covariate matrix, and $X_i\in\R^{d}$ be unit $i$'s covariates.

In this three treatment arms setup, the Gaussianized design optimization framework breaks down to the procedure below. Technical details will be provided in Sections \ref{sec:discrete} and \ref{sec:opt}.
\begin{testprocedure}[ (High-Level Procedure of Gaussianized Design Optimization)]\label{proc:GDO}
\rm
\begin{enumerate}
    \item Specify the estimand: Here, we focus on the average treatment effect of all treatment arm 
\begin{equation*}
    \tau = \frac{1}{3} \sum_{k=1}^3 \tau_k\;,\quad \tau_k = \frac{1}{n} \sum_{i=1}^n Y_i(k)\;.
\end{equation*}
    We use a Hortivz-Thompson estimator $\widehat{\tau}$ to unbiasedly estimate this quantity.
    \item Derive measures of covariate balance: We propose the following covariate balance measures
    \begin{equation*}
        \sum_{k=1}^3 \|X^\top \Cov_k(D) X\|_{\mathrm{norm}}\;,\quad \norm  \in \{\op, \nuc\}\;,
    \end{equation*}
    where $\Cov_k(D)$ is the covariance matrix of $(\mathbb{I}\{D_1=k\}, \dots, \mathbb{I}\{D_n=k\})$. The measures in the operator and nuclear norm capture the worst-case and average-case mean squared error (MSE) of $\widehat{\tau}$, respectively, as explained in Section \ref{sec:discrete}.
   \item Gaussianizaton: We model treatments by $D_i = g(T_i)$ and derive that $\Cov_k(D) = f_k(\Sigma)$, $k = 1,2,3$ with analytical expressions of $f_k$. These known functions $f_k$ are explicitly given in Proposition \ref{prop:f_formula_general}. The act of Gaussianization translates covariate balance measures to an analytical function on the Gaussian covariance $\Sigma$, as follows:
\begin{equation*}
    \sum_{k=1}^3 \|X^\top \Cov_k(D) X\|_{\mathrm{norm}} = \sum_{k=1}^3 \|X^\top f_k(\Sigma) X\|_{\mathrm{norm}}\;, \quad {\mathrm{norm}}\in \{\mathrm{op}, \mathrm{nuc}\}\;.
\end{equation*}
    \item Solve the design optimization: We apply a first-order algorithm (projected gradient descent on the Gaussianized space) in Section \ref{sec:opt} to solve
    \begin{equation}\label{eq:obj3}
        \underset{\Sigma \in \cE}{\min}\sum_{k=1}^3 \|X^\top f_k(\Sigma) X\|_{\mathrm{norm}}\;.
    \end{equation}
    This returns a locally optimal Gaussian covariance matrix $\Sigma^*$.
    \item Assign treatments. Generate treatments through $D_i = g(T_i)$, $T\sim\cN(0, \Sigma^*)$.
\end{enumerate}
\end{testprocedure}

\paragraph{Optimization and sampling benefits.}
Procedure \ref{proc:GDO} applies to general experimental setups with $\mathbb{D} = \{1, \dots, K\}$, where $K$ is the total number of treatment arms (Section \ref{sec:discrete}). Given covariate balance measures of the form $\sum_{k=1}^K \|X^\top \Cov_k(D) X\|_{\mathrm{norm}}$ in Step 2, one would naturally search for a valid design with the optimal covariate balance measure. However, it is unclear how to directly optimize the design for $D$. First, optimization over the treatment covariance $\Cov_k(D)$ is computationally challenging: for binary treatment assignments ($K = 2$), we need to solve
\begin{equation*}
\min_{\Cov_1(D)}\|X^\top \Cov_1(D) X\|_{\mathrm{nuc}} = \min_{\Cov_1(D)}\tr( XX^\top \Cov_1(D)) \Leftrightarrow \min_{C\in\mathcal{C}}\tr( XX^\top C)\;.
\end{equation*}
The feasible set of $\Cov_1(D)$ is affinely isomorphic to the cut polytope $\mathcal{C}$ \citep{huber2017bernoulli}, and the optimization problem is thus equivalent to the Max-Cut problem \citep{barahona1986cut}, which is NP-hard. Second, even if one somehow obtains an approximate solution for $\Cov_k(D)$, it is still unclear how to sample discrete treatment assignments $\{D_i\}_{i=1}^n$ to achieve the desired covariance matrices $\Cov_k(D)$, $k = 1, \dots, K$.

Gaussianization mitigates these computational and sampling difficulties. Based on the discussion above, Gaussianization transforms the design problem into a nonlinear optimization of the form \eqref{eq:obj3}.
Compared to the direct optimization on design $D$, Gaussianization provides two key advantages. First, the optimization becomes generic nonlinear programming on the correlation elliptope, which can be efficiently solved using first-order local algorithms. Second, once an optimizer $\Sigma^*$ is obtained, treatment assignments can be sampled directly via $D_i = g(T_i)$, $T\sim\cN(0, \Sigma^*)$. Notably, this Gaussianization idea has been applied in optimization and theoretical computer science literature, where it is known as randomized hyperplane rounding \citep{williamson2011design}. More specifically, \cite{goemans1995improved} propose an approximation algorithm for the Max-Cut problem, where the idea is to generate a cut vector by thresholding a correlated Gaussian vector, with the correlation matrix obtained as the solution to a semidefinite program. Our approach shares a similar procedure when $K=2$: the action of Gaussianization in our approach is precisely the Goemans-Williamson rounding, with the analytic function $f(x) = \arccos(x)$ (derived from $f_k$ in Proposition \ref{prop:f_formula_general}) shared in the analysis.

\paragraph{Constraints on Gaussianized design optimization.}
From the discussion above, the Gaussianization framework focuses on a specific class of designs whose variance-covariance matrices satisfy
\begin{equation*}
    \Cov_k(D) \in \{f_k(\Sigma) \mid \Sigma\in\cE\}\;.
\end{equation*}
This design class, induced by Gaussianization, is generally a subset of all possible designs, which may be limited under certain scenarios. However, this limitation may still be preferable to the global design optimization that involves NP-hard instances and significant sampling challenges. Probably due to this reason, existing works on design optimality usually focus on specific design classes, such as stratified designs \citep{bai2022optimality} and rerandomization designs \citep{wang2023bestchoice,li2020jrssb}.

Since the design objective is generally non-convex in $\Sigma$, only a local optimizer can be guaranteed. While the local optimality may appear restrictive, it is worth noting that one can flexibly initialize the design optimization from any $\Sigma \in \cE$. For example, in the simulation below, we initialize the design optimization using a stratified design. In this sense, our approach serves as a general optimization tool to refine an input design, where the design may already adapt to covariates through other existing methods.

\paragraph{Simulation.}
To demonstrate the concrete benefits of design optimization, we conduct a simple simulation that evaluates the MSE of $\widehat{\tau}$ under various designs. Two covariate structures are considered: (a) the first covariate has the largest scale and serves as the sole informative feature, and (b) all covariates are uniformly generated and are equally informative. We initialize our iterative algorithm using either an i.i.d. design ($\Sigma=I_n$) or a block design, where $\Sigma$ is a block correlation matrix representing a classical block design constructed by sorting the first covariate. Further details of the simulation are provided in Section \ref{sec:appendix_simu}.

Figure \ref{fig:mse_3treatments} shows the MSE trajectories over iterations in Gaussianized design optimization. In setup (a), initializing with a block design yields a smaller MSE by leveraging the highly informative covariate. Furthermore, starting from the i.i.d. design and minimizing the covariate balance measure with the operator norm results in a lower final MSE. Different initial designs in setup (b) produce similar early-stage MSEs but diverge in their final performance. Notably, with suitable choices of the initial design and the norm, Gaussianized design optimization reduces the MSE by more than 60\% through iterations. Figure \ref{fig:corr_3treatments} provides heatmaps of the correlation matrices for different initializations and norms. Observe that Panels (b), (c), (g), and (h) preserve the block structure, which highlights how design optimization makes local improvements.

\begin{figure}[t!]
    \vspace{-7mm}
    \centering
    \subfloat[]{\includegraphics[width=.45\linewidth]{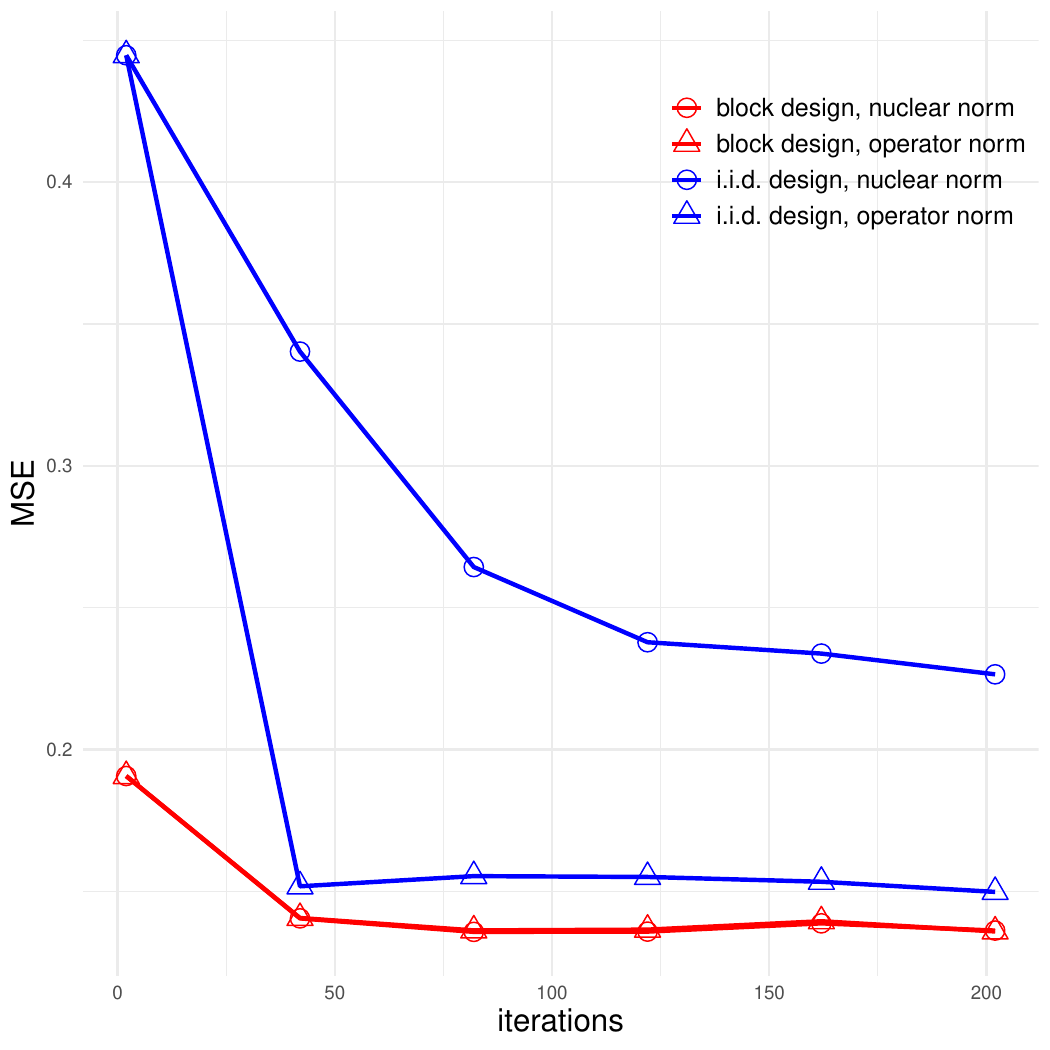}}
    \subfloat[]{\includegraphics[width=.45\linewidth]{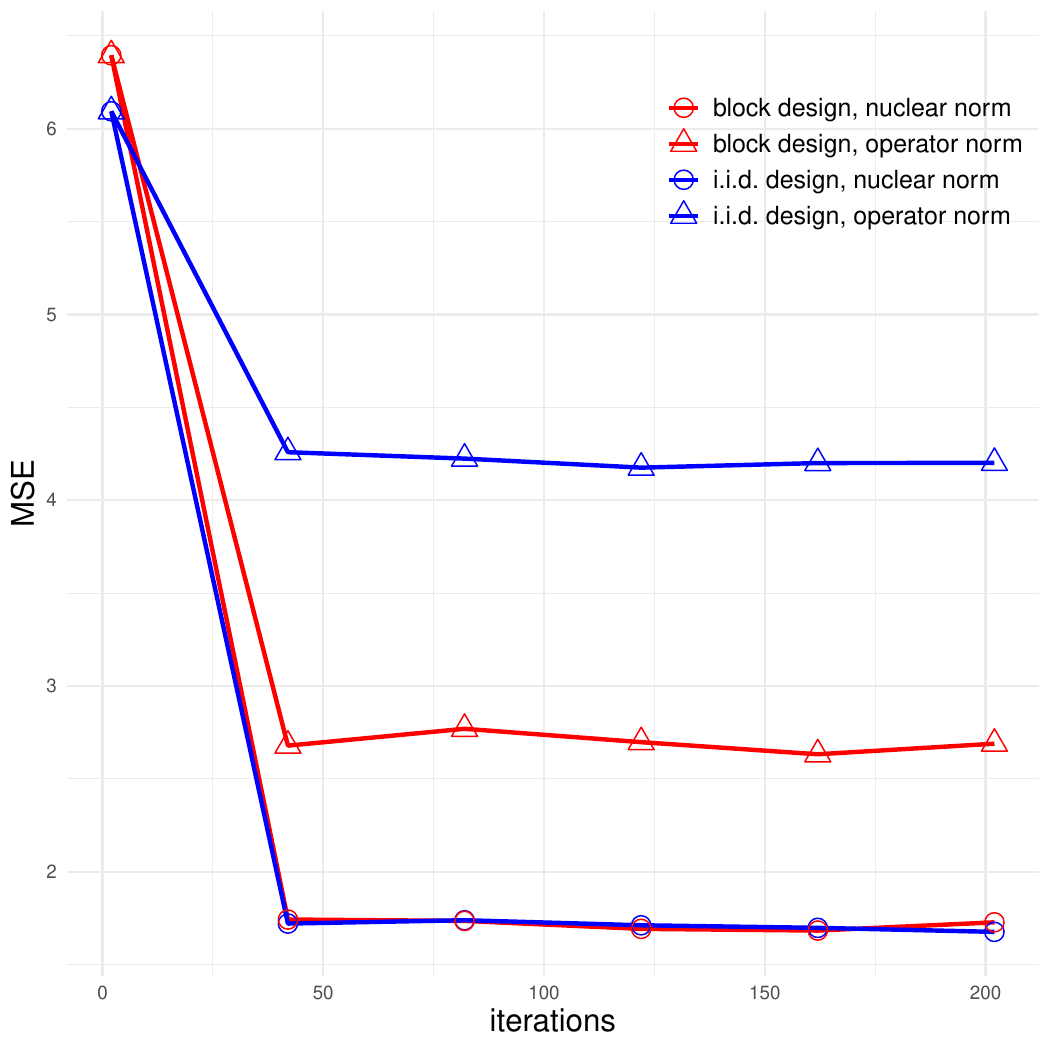}}
    \caption{MSE of the Horvitz-Thompson estimators over iterations of design optimization.}
    \label{fig:mse_3treatments}
\end{figure}

\begin{figure}[t!]
    \vspace{-7mm}
    \centering
    \subfloat[]{\includegraphics[width=.2\linewidth]{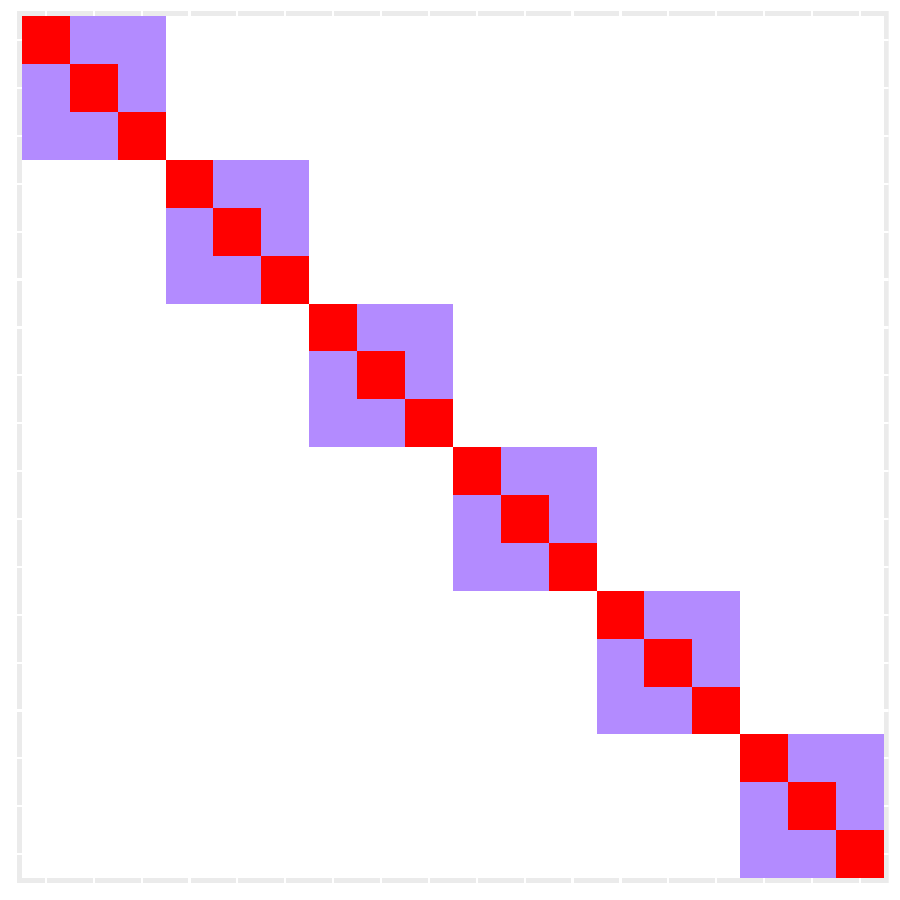}}
    \subfloat[]{\includegraphics[width=.2\linewidth]{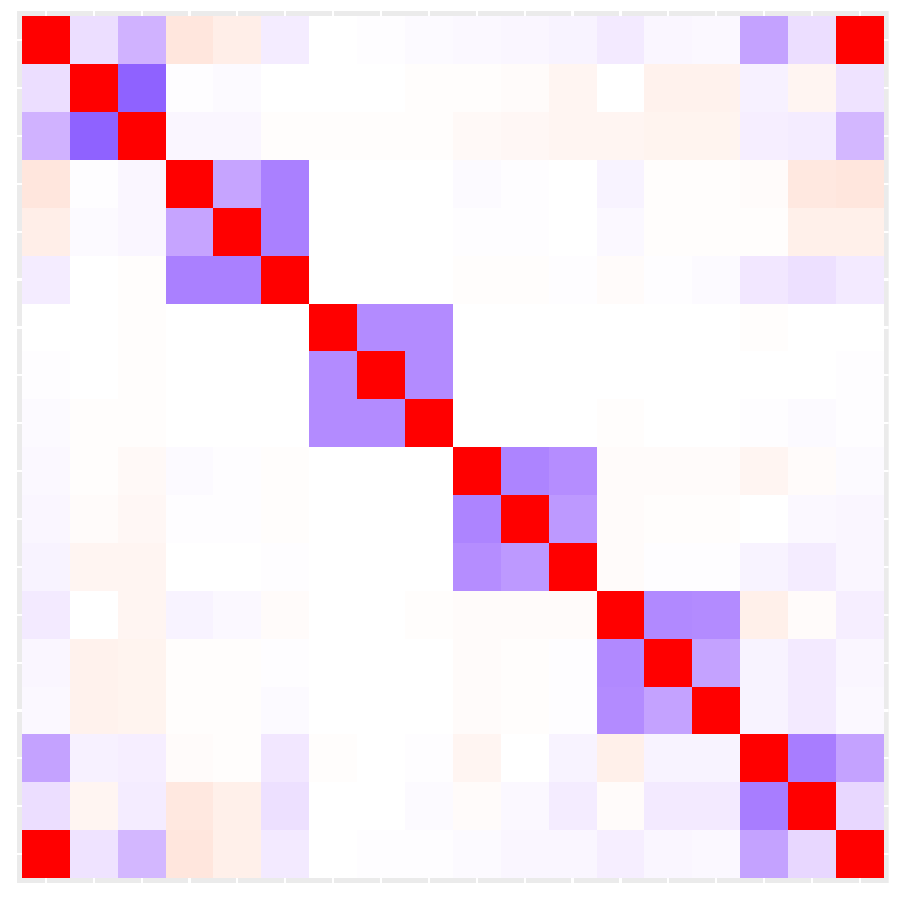}}
    \subfloat[]{\includegraphics[width=.2\linewidth]{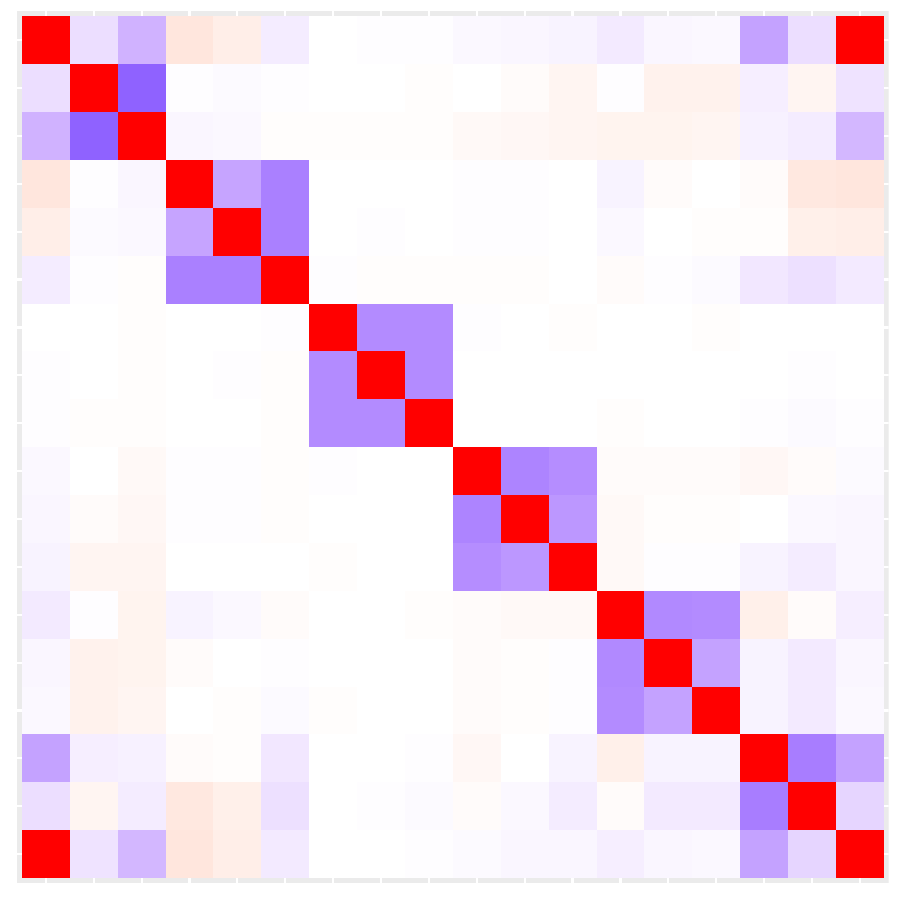}}
    \subfloat[]{\includegraphics[width=.2\linewidth]{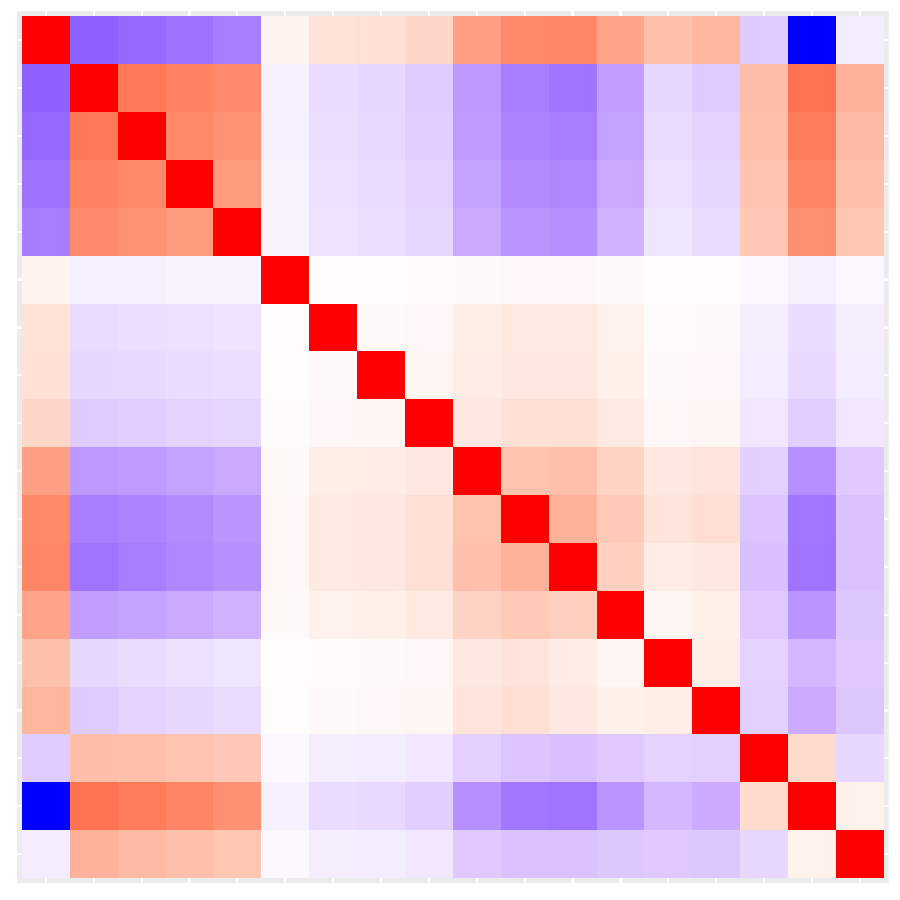}}
    \subfloat[]{\includegraphics[width=.23\linewidth]{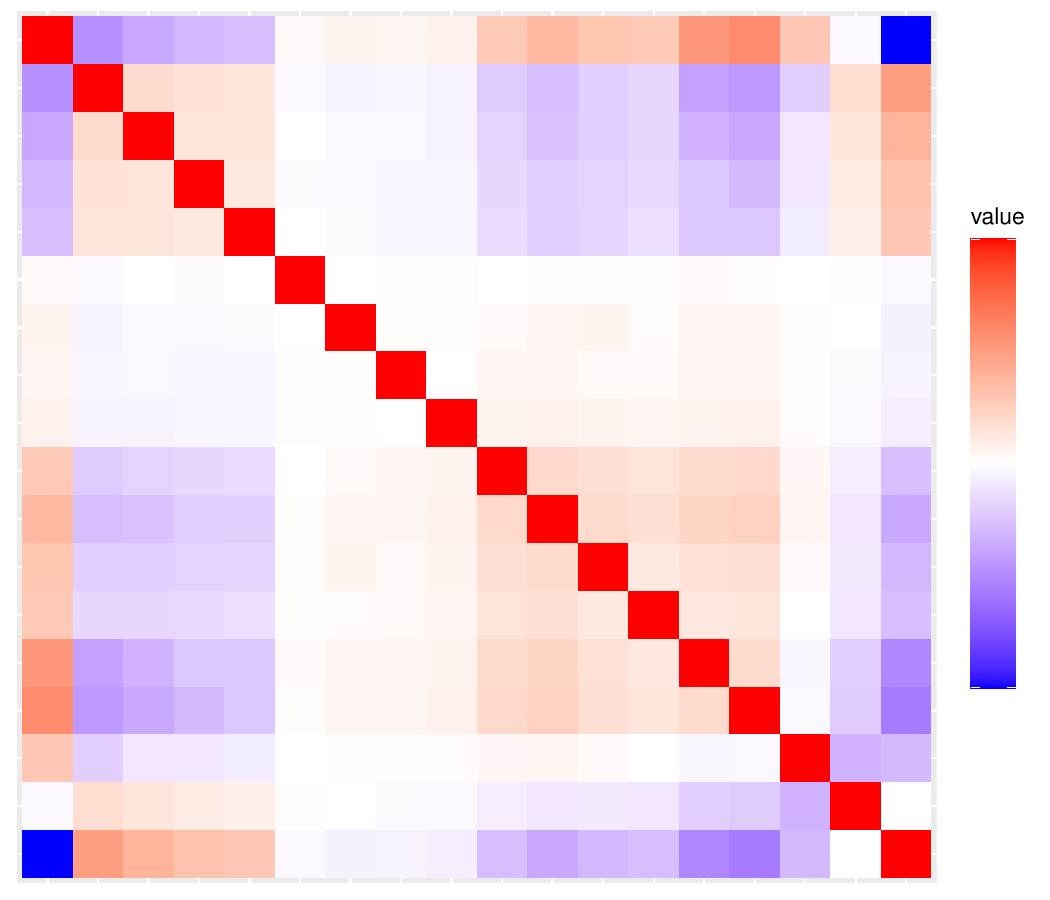}}\\
    \subfloat[]{\includegraphics[width=.2\linewidth]{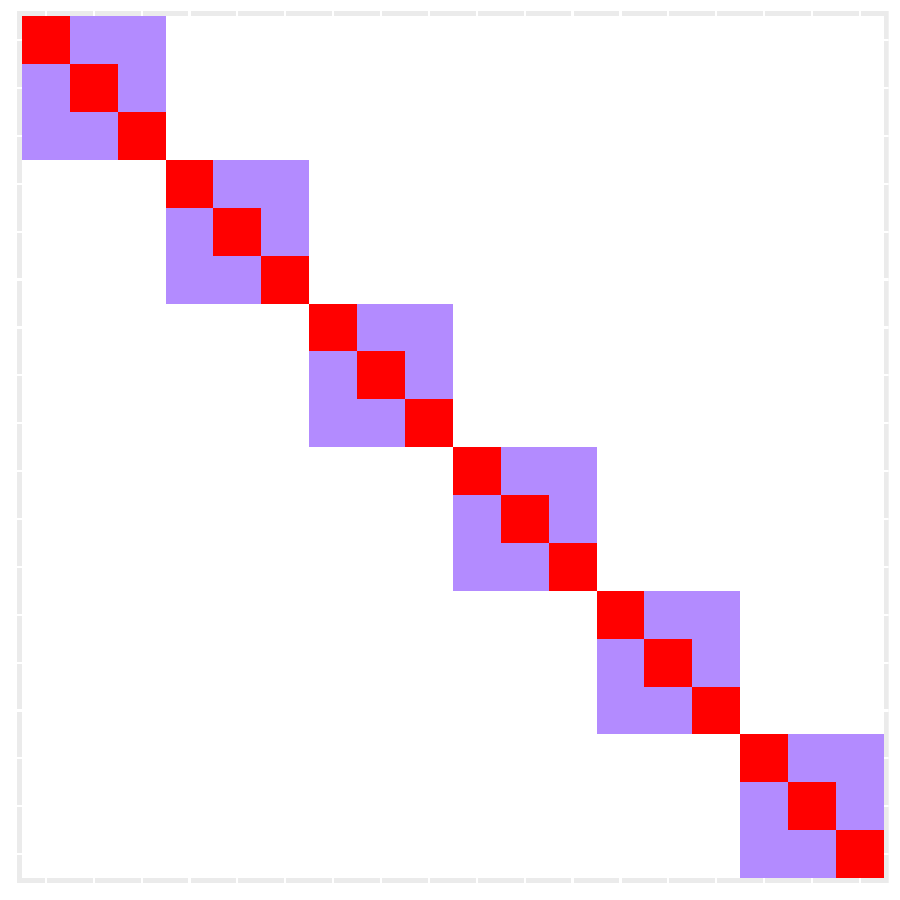}}
    \subfloat[]{\includegraphics[width=.2\linewidth]{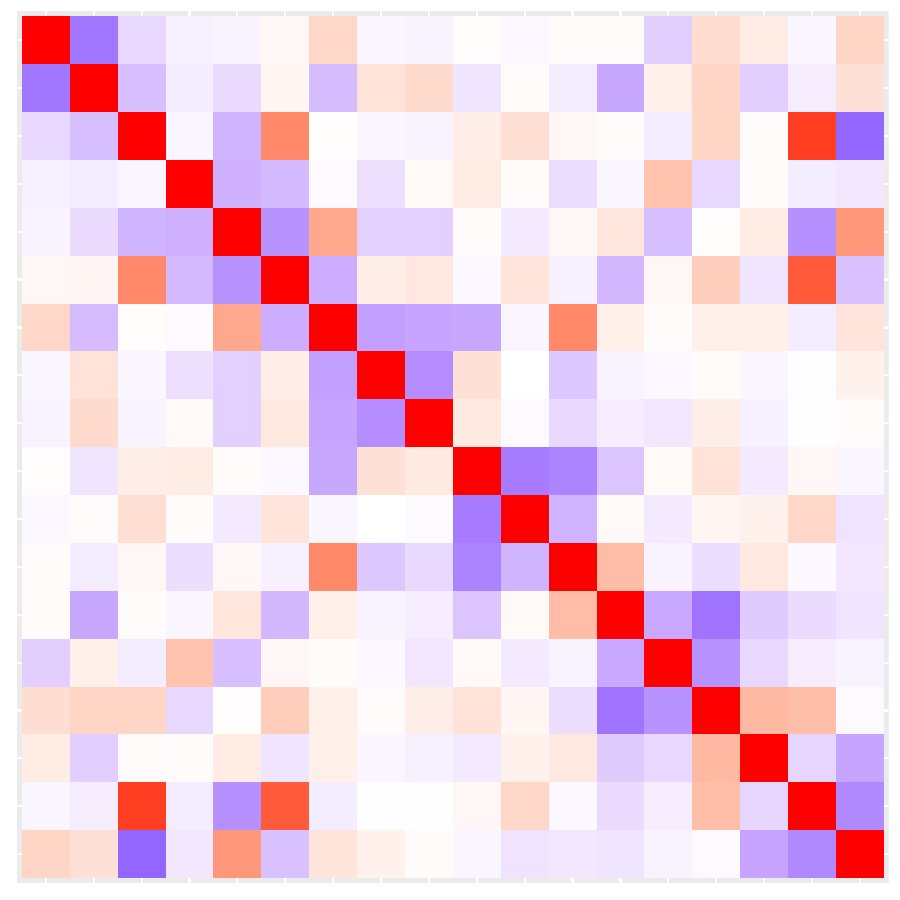}}
    \subfloat[]{\includegraphics[width=.2\linewidth]{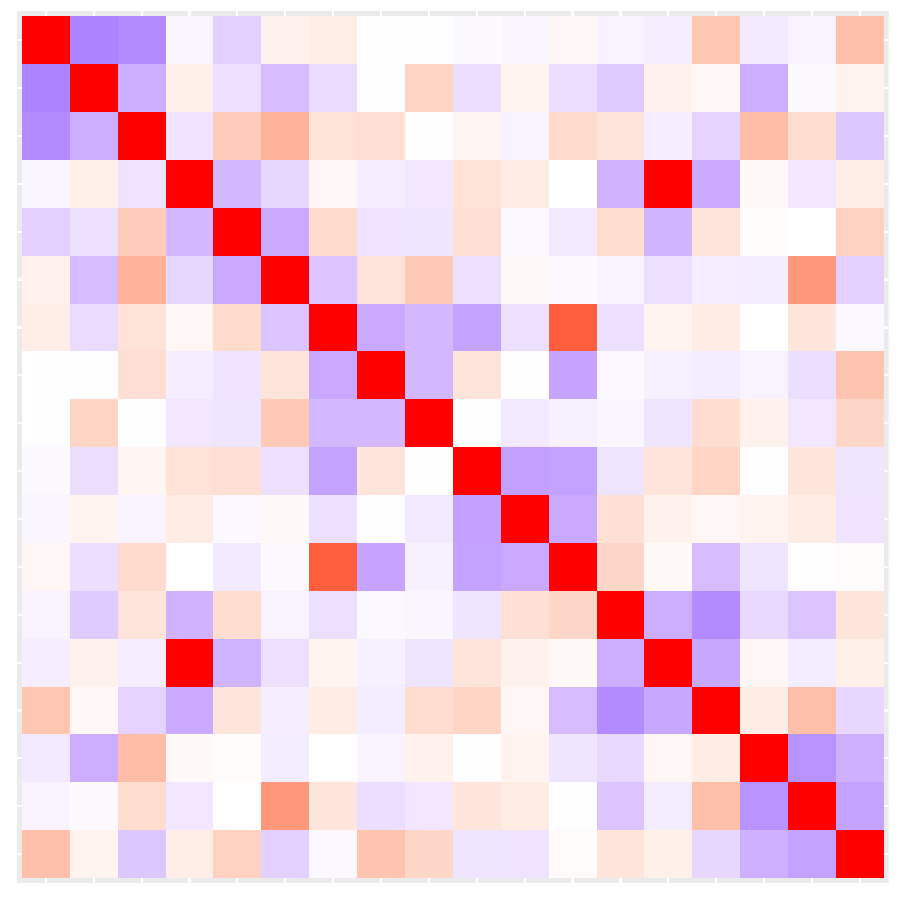}}
    \subfloat[]{\includegraphics[width=.2\linewidth]{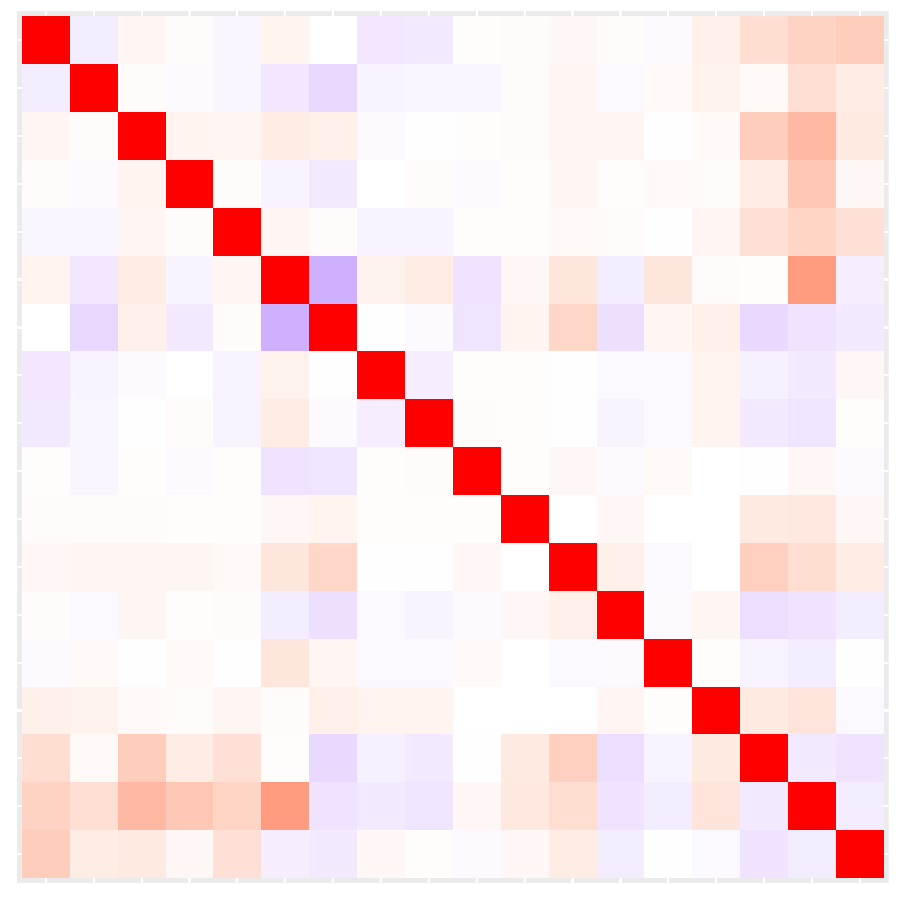}}
    \subfloat[]{\includegraphics[width=.23\linewidth]{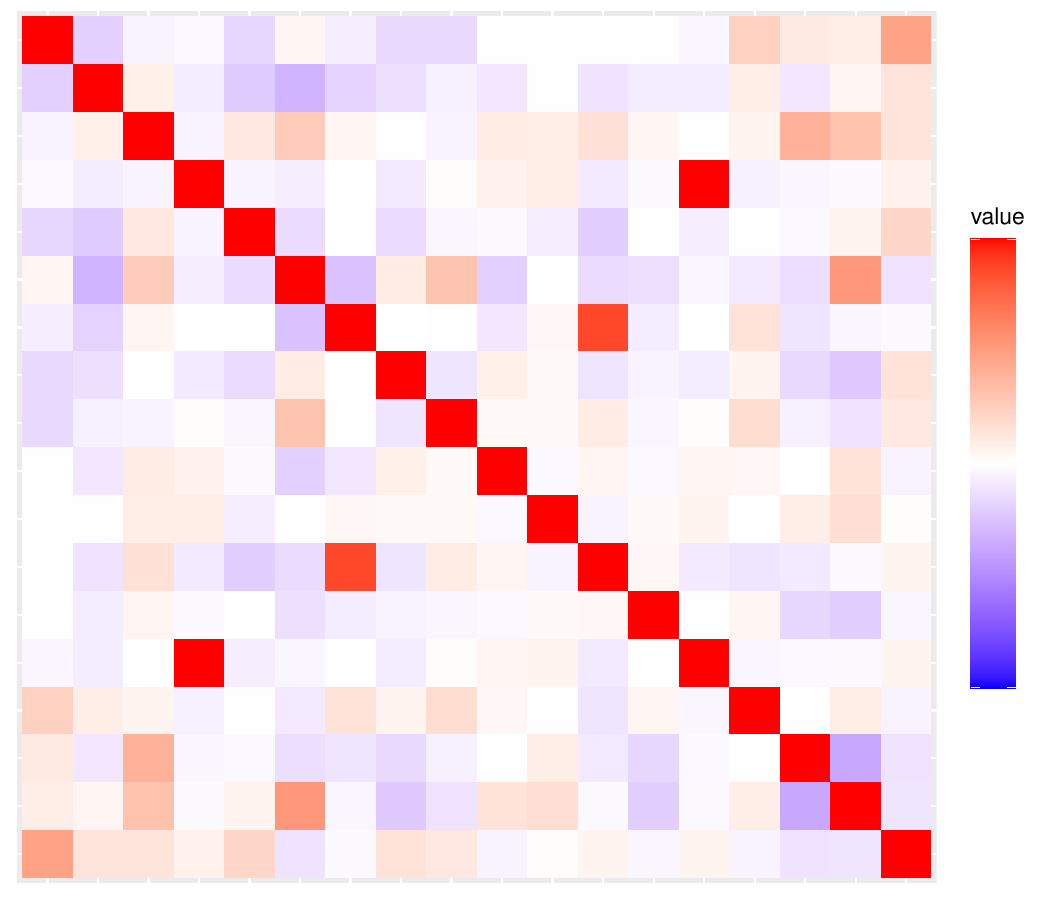}}\\
    \caption{Heatmaps of covariance matrix $\Sigma$ from Guassianized design optimization. The first and second row correspond to the single feature setup and the uniform covariate setup. In each row, from left to right, we show the initial block design, optimized block designs under the operator and nuclear norm, optimized i.i.d. designs under the operator and nuclear norm. Red denotes 1, blue denotes -1, and white denotes 0.}
    \label{fig:corr_3treatments}
\end{figure}

\section{Related Work}
Randomized experiments, such as i.i.d. Bernoulli designs and complete randomizations, are generally viewed as robust designs and are thus desirable in practice. Fisher first demonstrated that complete randomization ensures unbiased estimations and therefore facilitates inference and testing \citep{fisher1925statistical,fisher1926arrangement}. Subsequently, \cite{wu1981robustness} framed robustness in terms of the worst-case mean squared error (MSE) by showing that complete randomization is a minimax design, meaning it minimizes the worst-case MSE. In time-series experiments such as N-of-1 trials, complete randomization has also been shown to be robust against both estimand choices and model misspecifications \citep{liang2023randomization}. See also \cite{kallus2018optimal,Harshaw2019,basse2023minimax,nordin2022properties,bai2023randomize} for related discussions.

In settings where covariate information is available, which 
is the focus of this paper, it is 
reasonable to adapt the experimental design so as to achieve covariate balance across treatment arms. 
These designs are collectively known as {\em covariate-adaptive randomization}.
With few covariates, blocking is the canonical way to reduce unwanted variation and increase precision~\citep{fisher1926arrangement}. 
Matched-pair designs \citep{greevy2004optimal,imbens2015causal} are prime examples of blocking, where each block contains two units, and are optimal under certain conditions~\citep{bai2022optimality}.
However, blocking can be impractical with many covariates. 
This has motivated sampling-based techniques such as rerandomization \citep{morgan2012rerandomization}, 
which follow an accept-reject sampling procedure according to certain covariate balance criteria.
As shown in \cite{wang2023bestchoice,li2020jrssb}, rerandomization procedures reduce estimation variance by adjusting for the linear component in outcomes that covariates can explain, making them optimal given appropriate covariate balance criteria. However, choosing the right trade-off between covariate balance criteria and the computational complexity of sampling can be challenging, especially with high-dimensional covariates. Moreover, blocking and rerandomization mainly focus on binary treatment settings, and thus it remains unclear how to optimally balance covariates with multiple treatments.

More relevant to our work, \cite{Harshaw2019} introduced the Gram-Schmidt Walk (GSW) design that formalizes the trade-off between covariate balance and robustness in binary treatment settings. Specifically, considering $Z$ as a binary treatment assignment vector, \cite{Harshaw2019} proposed $\|\Cov(Z)\|_{\op}$ and $\|X^\top \Cov(Z) X\|_{\op}$ as measures of robustness and covariate balance, respectively. The GSW design navigates the robustness-balance trade-off by proposing a weighted combination of the aforementioned measures, and employs a random walk to sequentially generate treatment assignments that optimize the weighted combination. Their measure of covariate balance (i.e., $\|X^\top \Cov(Z) X\|_{\op}$) motivates us to study similar norm-based objectives; when $K=2$, our objective reduces to the GSW objective.

A recent line of work has addressed inference under 
the covariate-adaptive designs described above, which can be challenging due to the complex covariate-treatment dependencies.
See, for instance, \cite{bugni2018inference, bugni2019inference,ma2020statistical, bai2024inference} for inference under stratified designs; \cite{bai2022inference} for matched-pair designs; and \cite{li2018asymptotic,li2020rerandomization} for rerandomization. 
Our asymptotic results under Gaussianization follow this line of work with different proof techniques.
In addition to asymptotic inference, we also consider Fisherian-style randomization inference in Section \ref{sec:appendix_inference}. Although Fisherian randomization was originally developed to test the sharp null of no treatment effects \citep{fisher1935design}, it has recently been extended to detect heterogeneity \citep{ding2016variation} and interference \citep{basse2019randomization,huang2025monotonespillover}; these extensions can also be combined with flexible machine learning models for higher efficiency \citep{guo2025ml}. Our paper leverages randomization to compute design-based confidence intervals.

Our paper also contributes to the growing literature on continuous treatments. Using techniques from doubly robust methodology, \cite{colangelo2020double, kennedy2017non} studied the estimation of the average potential outcome function, while \cite{hsu2024testing} tested functional properties such as  monotonicity. Recently, \cite{callaway2024difference} analyzed difference-in-differences setups with a continuous treatment, and discussed the identification of response functions and their derivatives. See \cite{dong2023nonparametric, schindl2024incremental,de2022difference} for related studies. However, all these works consider i.i.d. data from observational studies, which is distinct from our experimental design setup. 

\section{A Gaussianization Framework}\label{sec:discrete}
In this section, we formally introduce the Gaussianization framework, which includes both norm-based covariate balance measures and their Gaussianized representations. Our formulation accommodates general experimental setups with the discrete support $\mathbb{D} = \{1, \dots, K\}$. We conclude this section by presenting Mehler’s formula \citep{mehler1866, liang2022mehler}, a key technical insight that motivates our design optimization.

\subsection{General Covariate Balance Measures}
We consider the potential outcome framework as in Section \ref{sec:example}, and focus on uniform designs such that $\P(D_i = k) = 1/K$ for any $i = 1, \dots, n$ and $k = 1, \dots, K$. Non-uniform designs, where $D_i$ follows non-uniform marginal treatment probabilities, can also fit within our Gaussianization framework by slightly adjusting the Gaussianization function $g$. The key requirement is that all treatment assignments share the same marginal distribution to enable effective design optimization.

We define our estimand as follows
\begin{equation*}
    \tau_w = \sum_{k=1}^K w_k \tau_k\;, \quad \tau_k = \frac{1}{n} \sum_{i=1}^n Y_i(k)\;,
\end{equation*}
where $w = (w_1, \dots, w_k)$ is a pre-specified vector. This can be a contrast vector, e.g., $w = (1,-1,0,\dots, 0)$, leading to the average treatment effect of treatment arm 1 over 2. It can also be a weight vector, e.g., $w_k = 1/K$ and $\sum w_k = 1$, which reduces to the estimand in Section \ref{sec:example} given $K = 3$. These estimands encompass a rich class of causal quantities, and thus they are of primary interest in empirical research.

To estimate $\tau_w$, we use the Horvitz-Thompson estimator as mentioned in Section \ref{sec:example}:
\begin{equation*}
    \widehat{\tau}_w = \sum_{k=1}^K w_k \widehat{\tau}_k\;,\quad \widehat{\tau}_k = \frac{K}{n} \sum_{i=1}^n \mathbb{I}\{D_i = k\} Y_i\;.
\end{equation*}
We focus on Horvitz-Thompson estimators, similar to previous works in the design optimality literature \citep{bai2022optimality, Harshaw2019,wang2023bestchoice}. More importantly, the Horvitz-Thompson estimator $\widehat{\tau}_w$ is the optimal linear unbiased sampling estimator of $\tau_w$ \citep{hege1967optimum}, and thus is desirable for design optimization. 
Alternatively, one could consider covariate-adjusted estimators \citep{fisher1935design, Chang2023designbased, list2024using}, but their performance is model-specific, potentially leading to biased estimations \citep{Freedman2008a}. More detailed discussions are provided in Section \ref{sec:cov_adjust}.

While $\widehat{\tau}_w$ is unbiased, its mean squared error (MSE) would depend on specific design structures through the covariance matrix of $D$. The following result has been proved in many works, e.g., \cite{Chang2023designbased}.
\begin{lemma}\label{lem:discrete_var}
Under uniform experimental designs, for $k = 1,\dots, K$, we have
\begin{equation*}
\E(\widehat{\tau}_k - \tau_k)^2 = \frac{K^2}{n^2} Y(k)^\top \Cov_k(D) Y(k)\;,
\end{equation*}
where $Y(k) = (Y_1(k), \dots, Y_n(k))^\top$, and $\Cov_k(D)$ is defined in Section \ref{sec:example}.
\end{lemma}
From Lemma \ref{lem:discrete_var}, the MSE of $k$-th treatment effect is a quadratic form in the covariance matrix of the treatment assignment vector, $\Cov_k(D)$, evaluated at the (unknown) potential outcome vector $Y(k)$. Then, for a general estimator $\widehat{\tau}_w$, we utilize the AM-QM inequality to obtain
\begin{align*}
\E (\widehat{\tau}_w - \tau_w)^2 &= \E \left(\sum_{k=1}^K w_k (\widehat{\tau}_k - \tau_k)\right)^2 \\
&\le K \sum_{k=1}^K w_k^2 \E (\widehat{\tau}_k - \tau_k)^2 = \frac{K^3}{n^2} \sum_{k=1}^3 w_k^2 Y(k)^\top \Cov_k(D) Y(k)\;.
\end{align*}

The MSE bound on $\widehat{\tau}_w$ leads to measures of covariate balance. Specifically, following a similar idea as in \cite{Harshaw2019}, let's assume for the moment that potential outcomes are perfectly linear in the covariates, i.e., $Y(k) = X\beta_k$, for some $\beta_k\in\R^{d}$. This reduces the MSE bound to 
\begin{equation*}
    \mathrm{MB} \coloneqq \frac{K^3}{n^2}\sum_{k=1}^K w_k^2 \beta_k^\top X^\top \Cov_k(D) X\beta_k\;.
\end{equation*}
In practice, even if we can somehow justify perfect linearity, the signals $\{\beta_k\}_{k=1}^K$ are in general unknown. \cite{Harshaw2019} formulate a worst-case MSE by assuming that the signal has a fixed norm with arbitrary directions. Following their idea, we consider a structural assumption that for $k = 1,\dots, K$, $\|\beta_k \| \le M$, leading to a measure of worst-case MSE:
\begin{align}
    \sup_{\|\beta_k\|\le M} \mathrm{MB} &\propto \sup_{\|\beta_k\|\le M} \sum_{k=1}^K w_k^2 \beta_k^\top X^\top \Cov_k(D) X\beta_k \propto \sum_{k=1}^K w_k^2 \sup_{\|\beta_k\|\le 1} \beta_k^\top X^\top \Cov_k(D) X\beta_k\nonumber\\
    &= \sum_{k=1}^K w_k^2 \|X^\top \Cov_k(D) X\|_{\mathrm{op}} \;.\label{eq:opt_op_general}
\end{align}
As an alternative, \cite{Isaki1982} and \cite{Chang2023designbased} have introduced the notion of ``anticipated variance" that measures an averaged MSE under a prior distribution on the potential outcomes. Following their idea, we consider that $\{\beta_k\}_{k=1}^K$ are random signals with mean zero and identity covariance. This leads to a measure of average-case MSE:
\begin{align}
	\E_{\beta_k} \mathrm{MB} &\propto \sum_{k=1}^K w_k^2 \E \beta_k^\top X^\top \Cov_k(D) X\beta_k \nonumber = \sum_{k=1}^K w_k^2 \tr( X^\top \Cov_k(D) X \E \beta_k \beta_k^\top)\nonumber\\
    &= \sum_{k=1}^K w_k^2 \tr( X^\top \Cov_k(D) X) = \sum_{k=1}^K w_k^2 \|X^\top \Cov_k(D) X\|_{\mathrm{nuc}}\;. \label{eq:opt_nuc_general}
\end{align}
The derivation above leads to the formal definition of covariate balance measures.
\begin{definition}
For a uniform design with $K$ treatments, we define the covariate balance measure in the nuclear and operator norm as
\begin{equation*}
    \sum_{k=1}^K w_k^2 \|X^\top \Cov_k(D) X\|_{\norm}\;, \quad \norm\in\{\nuc, \op\}\;.
\end{equation*}
\end{definition}
To summarize, by making various structural assumptions on the value of $\beta_k$, we derive covariate balance measures that only depend on $X$ and the design $D$. Importantly, this motivates the study of objective \eqref{eq:cov_metric} as the basis for optimal experimental design, as we show in the sequel. By setting $w_k = 1/K$ and $K = 3$, one recovers the measures exemplified in Procedure \ref{proc:GDO}. Additionally, when $K=2$, our measure $\sum_{k=1}^K\|X^\top \Cov_k(D) X\|_{\mathrm{op}}$ is equivalent to the covariate balance measure studied in \citep{Harshaw2019}.

\subsection{Gaussianized Representation}
We introduce a Gaussianized representation of the uniform design $D$ through a map $g:\R\to\{1,\dots, K\}$ as defined below:
\begin{equation*}
    g(t) = 
    \begin{cases}
        i \quad &\text{if}~ t\in \left(\Phi^{-1}\left(\frac{i-1}{K}\right), \Phi^{-1}\left(\frac{i}{K}\right)\right]\;,\quad i=1,\dots, K-1\\
        K \quad &\text{if}~ t \in (\Phi^{-1}\left(\frac{K-1}{K}\right), \infty)\\
    \end{cases}\;,
\end{equation*}
where $\Phi(\cdot)$ is the standard normal CDF. In other words, we discretize the Gaussian treatments $T$ according to the equidistance quantiles. This recovers the uniform design since $g(T_i)$ is uniformly distributed on $\{1, \dots, K\}$.

When a uniform design is from a Gaussianized representation, i.e., $D_i = g(T_i)$ for $T\sim\cN(0, \Sigma)$, the variance-covariance matrix of $D$ is completely captured by $\Sigma$. Surprisingly, one can link these two covariance matrices through analytical formulas.
\begin{proposition}\label{prop:f_formula_general}
Under Gaussianization $D_i = g(T_i)$ for $K$ treatment arms, we have $\Cov_k(D) = f_k(\Sigma)$, where $f_k:[-1, 1]\to\R$, $k = 1, \dots, K$ are elementwise functions defined by 
\begin{align*}
    f_k(\rho) = 
    \begin{cases}
        r_{1,1}(\rho) &\text{if}~k = 1\\
        r_{K-1,K-1}(\rho) &\text{if}~k = K\\
        r_{k-1,k-1}(\rho) + r_{k, k}(\rho) - 2  r_{k-1,k}(\rho) &\text{otherwise}
    \end{cases}\;.
\end{align*}
For $i,j=1, \dots, K-1$, we have
\begin{equation}\label{eq:corr_func}
    r_{i,j}(\rho) \coloneqq \Cov(\mathbb{I}\{X\le q_i\}, \mathbb{I}\{Y\le q_j\}) = \int_0^{\rho} \frac{1}{2\pi \sqrt{1-r^2}} \exp( -\frac{q_i^2 + q_j^2 -2r q_i q_j}{2(1-r^2)} ) \dd r\;,
\end{equation}
where $q_i = \Phi^{-1}(i/K)$, and $(X, Y)$ follows the bivariate normal distribution with variance one and correlation $\rho$.
\end{proposition}
Proposition \ref{prop:f_formula_general} provides a concrete procedure to compute the covariance matrix $\Cov_k(D)$. Importantly, it facilitates design optimization under Gaussianization, since one can formulate the covariate balance measures as objective functions on $\Sigma$: 
\begin{equation}\label{eq:design_opt}
    \sum_{k=1}^K w_k^2 \|X^\top \Cov_k(D) X\|_{\norm} = \sum_{k=1}^K w_k^2 \|X^\top f_k(\Sigma) X\|_{\norm}\;.
\end{equation}
In summary, we propose general covariate balance measures for uniform designs and derive their explicit Gaussianized representations. This Gaussianization enables feasible design optimization algorithms over the space of Gaussian covariance matrices, which will be the focus of Section \ref{sec:opt}.

Proposition \ref{prop:f_formula_general} warrants more technical clarifications. First, its main benefit comes from \eqref{eq:corr_func}, which provides analytical expressions of $\Cov_k(D)$. Alternatively, one may evaluate each covariance in \eqref{eq:corr_func} by Monte Carlo, but such simulation-based methods can be computationally challenging for large-scale randomized experiments.
Second, we illustrate below the shape of $f_k$ through the three-treatment example.
\begin{remark}[Evaluation of $f_k$ in the three-treatment example]\label{rmk:f}
Given $K = 3$, we evaluate $f(\rho) = \sum_{k=1}^3 f_k(\rho)$, which maps $\Sigma$ to $\sum_{k=1}^3 \Cov_k(D)$. This function represents the design optimization objective in Section \ref{sec:example}, since for $w_k = 1/3$, the covariate balance measure in the nuclear norm reduces to
\begin{equation*}
    \sum_{k=1}^3 w_k^2 \|X^\top \Cov_k(D) X\|_{\nuc} \propto \bigl\| X^\top \sum_{k=1}^3\Cov_k(D) X \bigr\|_{\nuc} = \left\| X^\top f(\Sigma) X \right\|_{\nuc}\;.
\end{equation*}
From the visualization of $f(\cdot)$ in Figure \ref{fig:f_function}, we observe that negative (positive resp.) correlations in $\Sigma$ induce negative (positive resp.) correlations in $\sum_k \Cov_k(D)$, with $f(0) = 0$ being a fixed point. More interestingly, $f(-1)$ and $f(0)$ induce similar correlations that are close to zero, implying that perfect negative correlation and zero correlation in $T$ lead to similar MSE performance. Lastly, we highlight that the derivative of $f$ goes to infinity at the endpoints $\pm 1$. This singular behavior of $f^\prime$ will guide us in developing optimization algorithms in Section \ref{sec:opt}.
\begin{figure}[t!]
    \vspace{-3mm}
    \centering
    \includegraphics[width=.9\linewidth]{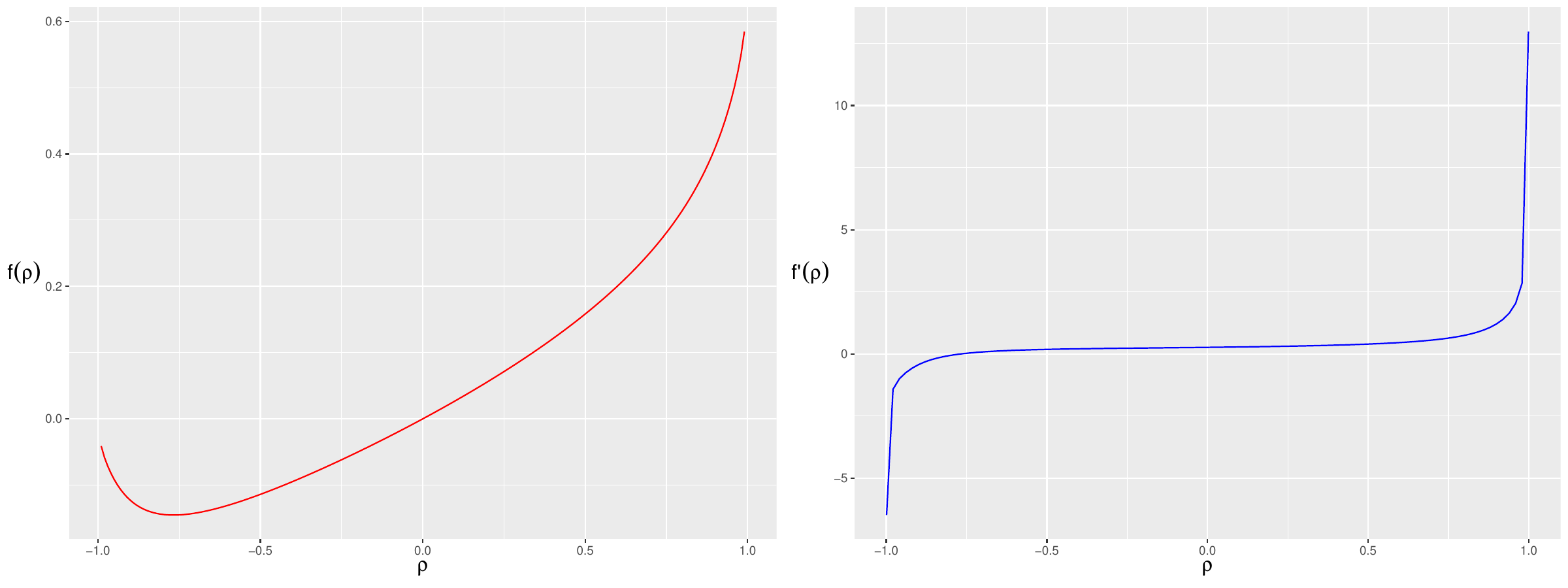}
    \caption{Function $f(\rho)$ and its derivative $f^\prime(\rho)$ on $(-1, 1)$.}
    \label{fig:f_function}
\end{figure}
\end{remark}

\subsection{Mehler's Formula and Proof Sketch of Proposition~\ref{prop:f_formula_general}}
We prove Proposition \ref{prop:f_formula_general} by leveraging a representation \citep{liang2022mehler} of bivariate normal distribution based on Mehler’s formula and Hermite polynomials. To begin with, we define Hermite polynomials in the probabilists' convention. 
\begin{definition}\label{def:hermite}
For non-negative integers $m\ge 0$, the $m$-th order Hermite polynomial is defined by 
\begin{equation*}
    \hermite_m(x) = \frac{(-1)^n}{\phi(x)} \frac{\dd^m}{\dd x^m}\phi(x)\;.
\end{equation*}
Here $\phi$ is the standard normal density function. Define the normalized Hermite polynomials as
\begin{equation*}
	h_m(x):=\frac{1}{\sqrt{m !}} \hermite_m(x)\;.
\end{equation*}
\end{definition}
Let $L^2_\phi$ be the class of square-integrable functions with respect to the standard normal distribution. Then, the set $\{h_m\}_{m=0}^\infty$ forms an orthonormal basis of $L^2_\phi$ as one can verify that
\begin{equation*}
	\E_{Z \sim \cN(0,1)}\left[h_m(Z) h_{m^{\prime}}(Z)\right]=\mathbb{I}\{m=m^{\prime}\}\;.
\end{equation*}
We can then define the Hermite coefficients as follows.
\begin{definition}
For any function $g\in L_\phi^2$, the $m$-th Hermite coefficient is defined by
\begin{equation*}
    \alpha_m[g]:=\underset{Z \sim \cN(0, 1)}{\mathbb{E}}\left[g(Z) h_m(Z)\right]\;.
\end{equation*}
\end{definition}

Let $p_\rho(x, y)$ be the density function of the bivariate normal distribution with variance one and correlation $\rho$. Mehler's formula \citep{mehler1866} connects $p_\rho(x, y)$ to Hermite polynomials as shown below:
\begin{equation*}
    p_\rho(x, y)=\sum_{m=0}^{\infty} \rho^m h_m(x) h_m(y) \phi(x) \phi(y)\;.
\end{equation*}
That is, the density $p_\rho(x, y)$ can be decomposed into a sequence of products of Hermite polynomials and standard normal densities. Based on this result, we establish a representation for the covariance of functions defined over bivariate normal distributions.
\begin{lemma}\label{lem:mehler}
For $g, h\in L_\phi^2$, if $(X, Y)\in\R^2$ follow a bivariate normal distribution with variance one and correlation $\rho$, we have
\begin{equation*}
    \underset{(X,Y)}{\Cov}[g(X), h(Y)] = \sum_{m=1}^{\infty} \alpha_m[g] \alpha_m[h] \rho^m\;.
\end{equation*}
\end{lemma}
Based on Mehler's formula and Lemma \ref{lem:mehler}, we show a sketch proof for Proposition \ref{prop:f_formula_general}. A complete proof can be found in Section \ref{sec:proof_mehler}.
\begin{proof}[Proof Sketch]
Here we focus on the proof of Equation \eqref{eq:corr_func}, which is the key step in proving the result. Let $g(x) = \mathbb{I}\{x \le q_i\}$ and $h(x) = \mathbb{I}\{x\le q_j\}$. We leverage the derivative representation of Hermite polynomials (Definition \ref{def:hermite}) to obtain
\begin{align*}
    \alpha_m[g] &= -\frac{1}{\sqrt{m!}} \phi(q_i)\hermite_{m-1}(q_i)\;,\quad \alpha_m[h] = -\frac{1}{\sqrt{m!}} \phi(q_j)\hermite_{m-1}(q_j)\;.
\end{align*}
Based on Lemma \ref{lem:mehler}, this implies 
\begin{align*}
    r_{ij}(\rho) = \sum_{m=1}^\infty \frac{1}{m!} \hermite_{m-1}(q_i) \hermite_{m-1}(q_i) \phi(q_i) \phi(q_j) \rho^m\;.
\end{align*}
Notice that $r_{ij}(0)= 0$ and 
\begin{align*}
    r_{ij}^\prime(\rho) &= \sum_{m=1}^\infty \frac{1}{(m-1)!} \hermite_{m-1}(q_i) \hermite_{m-1}(q_i) \phi(q_i) \phi(q_j) \rho^{m-1} = p_\rho(q_i, q_j)\;.
\end{align*}
We obtain
\begin{equation*}
    r_{ij}(\rho) = \int_0^{\rho} p_r(q_i, q_j) \dd r =  \int_0^{\rho} \frac{1}{2\pi \sqrt{1-r^2}} \exp( -\frac{q_i^2 + q_j^2 -2r q_i q_j}{2(1-r^2)} ) \dd r\;.
\end{equation*}
\end{proof}
In summary, Proposition \ref{prop:f_formula_general} can be proved by applying Mehler's formula to the covariance expression in \eqref{eq:corr_func}. This trick will be used again in design optimization for the continuous setting (Section \ref{sec:continuous}). Notably, this technical tool is designed for bivariate normal distributions, which further motivates the act of Gaussianization of treatments.

\section{Gaussianized Design Optimization}\label{sec:opt}
In this section, we will focus on solving the following optimization problems: 
\begin{equation}\label{eq:opt_general}
    \min_{\Sigma \in \cE} \|X^\top f(\Sigma) X\|_{\mathrm{norm}}\eqqcolon l_{\mathrm{norm}}(\Sigma)\;, \mathrm{norm} \in \{\mathrm{nuc}, \mathrm{op}\}\;,
\end{equation}
where $f$ is a given elementwise function defined on $[-1, 1]$. Based on the linearity of the nuclear norm, the objective in \eqref{eq:opt_general} in $\|\cdot\|_{\nuc}$ is equivalent to \eqref{eq:design_opt} by setting $f(\rho) = \sum_{k} w_k^2 f_k(\rho)$. Under the operator norm, the design optimization problem \eqref{eq:design_opt} is a weighted sum of objectives in the form of \eqref{eq:opt_general}, and one can slightly modify the algorithm below to solve \eqref{eq:design_opt}. Moreover, the general optimization problem \eqref{eq:opt_general} encompasses other covariate balance objectives in Section \ref{sec:continuous}.

Formally, we propose Algorithm \ref{alg:pgd} to solve \eqref{eq:opt_general} above. This algorithm applies projected gradient descent (PGD-Gauss) on a factorized representation of $\Sigma$, similar to the Burer-Monteiro approach in semidefinite programming \citep{burer2003nonlinear}.
\begin{algorithm}
	\DontPrintSemicolon
	\KwData{$X\in\R^{n\times d}$, an evaluation function $f$, and an initial design $\Sigma^1$. Number of iterations $T$.}
	\KwResult{Optimized covariance matrix $\Sigma^*$.}
	\Begin{
	Parametrize $\Sigma^1 = V^1 (V^1)^\top$, where $V^1\in\R^{n\times k}$, $\|v_i^1\| = 1$, and $k$ equals to the rank of $\Sigma^1$. Here $v_i^1$ is the $i$-th row of $V^1$. \;
	\For{$t = 1, \dots, T$}{
        Compute $\nabla l_{\mathrm{norm}}(\Sigma^t)$. \;
        $V^{t+1} = \big[ I_n - \eta_t \nabla l_{\mathrm{norm}}(\Sigma^t) \big]  V^{t}$ for a proper step size $\eta_t$.\;
        $v_i^{t+1}\gets v_i^{t+1}/\|v_i^{t+1}\|$. \;
        $\Sigma^{t+1} \gets V^{t+1} (V^{t+1})^\top$.\;
	}
    $\Sigma^* \gets \Sigma^T$.\;
	}
	\caption{Projected Gradient Descent for Gaussianized Design Optimization (PGD-Gauss)\label{alg:pgd}}
\end{algorithm}

The function $f$ in design optimization may have an infinite derivative at $\pm 1$ (Remark \ref{rmk:f}). Conceptually, this type of $f$ will set $\pm 1$ to be a barrier. Therefore, in Algorithm \ref{alg:pgd}, we fix the diagonal values of $\Sigma^t$ and only update on the off-diagonal entries. That is, we consider
\begin{align*}
   \nabla l_{\mathrm{nuc}}(\Sigma^t) &= (XX^\top - \mathrm{diag}(XX^\top)) \circ f'(V^t{V^t}^\top)\;,\\
   \nabla l_{\mathrm{op}}(\Sigma^t) &= (Xu_1 u_1^\top X^\top - \mathrm{diag}(Xu_1 u_1^\top X^\top)) \circ f'(V^t{V^t}^\top)\;,
\end{align*}
where $\circ$ is the Hadamard product, $u_1 \in \mathbb{R}^d$ is the leading eigenvector of $X^\top f(\Sigma^t) X$. For diagonal elements in the gradient, we adopt the convention $0 \times f^\prime(\pm 1) = 0$. Notably, $f^\prime$ can be obtained by directly differentiating the analytic functions $f_k$ defined in Proposition \ref{prop:f_formula_general}. That is, Proposition \ref{prop:f_formula_general} enables the direct computation of the gradient in Gaussianized design optimization.

Since the objective function is non-convex in $\Sigma$ in general, the PGD-Gauss only obtains a local optimizer near the initial covariance matrix $\Sigma_1$. As explained in Section \ref{sec:example}, Gaussianized design optimization is not tailored to identify the global solution that perfectly balances the covariates, but rather to serve as a tool for achieving local improvements based on a given input design.

By default, we initialize the design optimization by setting $\Sigma^1 = I_n$, which results in i.i.d. treatments. We view this as the baseline Gaussian design, as it does not take any covariate information, and i.i.d. designs have a robust performance against unknown outcome-generating models \citep{Harshaw2019}. Therefore, the number of steps we run PGD-Gauss is an explicit tradeoff between robustness and covariate balance. In simulations of Section \ref{sec:simu}, the i.i.d. initialization leads to satisfactory performance compared to state-of-the-art designs.

\section{Gaussian Design with Continuous Treatments}\label{sec:continuous}
In this section, we extend Gaussianization to settings with continuous treatments. Specifically, we introduce a new experimental design, called Gaussian design, to give continuous treatments based on multivariate Gaussian distribution. 
\begin{definition}[Gaussian Design]
A Gaussian design allocates treatment $T_i$ to unit $i$, where $T = (T_1, \dots, T_n)\sim\cN(0, \Sigma)$ for some $\Sigma\in\cE$.
\end{definition}
When the actual treatment is restricted to a bounded interval $[a,b]$, one may compute a rescaled treatment assignment $(a+b)/2 + T_i(b-a)/(2z_{0.999})$, where $z_\alpha$ denotes the $\alpha$-quantile of the standard normal distribution. This ensures that the rescaled treatment falls in $[a,b]$ with high probability ($>0.998$). Gaussian designs directly allocate continuous treatments as above, and it is thus mechanically different from the Gaussianization perspective, where we focus on discrete treatments but model them using latent Gaussian variables. Compared to Gaussianization, Gaussian designs capture average structural properties of potential outcome functions as we discuss below.

\subsection{Causal Estimands}\label{sec:cont_estimand}
We denote $Y_i(t)$ to be the response function for unit $i$ and $t\in\R$, which generalizes the potential outcomes to continuous treatments. With an abuse of notation, we use $Y_i = Y_i(T_i)$ to denote the observed outcome for unit $i$. Given continuous treatments and response functions, we work with a class of causal effects of the form
\begin{equation}\label{eq:estimand_cont}
    \tau_w^c = \frac{1}{n} \sum_{i=1}^n \int_{\R} Y_i(t) w(t) \phi(t) \dd t\;,
\end{equation}
where $w(\cdot)$ is a pre-specified weight function on different treatment values. We use the superscript $c$ in $\tau_w^c$ to distinguish it from $\tau_w$ under the discrete setting.

Similar to the discrete setup, we focus on Horvitz-Thompson-type estimators
\begin{equation*}
    \widehat{\tau}^c_w = \frac{1}{n} \sum_{i=1}^n Y_i(T_i) w(T_i) = \frac{1}{n} \sum Y_i W_i\;, \quad W_i\coloneqq w(T_i)\;.
\end{equation*}
Clearly, $\widehat{\tau}^c_w$ is an unbiased estimator of $\tau_w^c$ under Gaussian design. In the following, we provide several leading examples of the weight function $w(\cdot)$ in \eqref{eq:estimand_cont} to get meaningful causal estimands.

\begin{example}[Average Treatment Effects on a Given Interval]
Suppose we want to learn about the average treatment effect on a treatment interval $[r, l]$~\citep{fryges2008exports}. We may set $w(t) = \frac{\mathbb{I}\{t\in[r, l]\}}{(l-r)\phi(t)}$, which leads to
\begin{equation*}
\tau_w^c = \frac{1}{n}\sum_{i=1}^n \frac{1}{l-r}\int_r^l Y_i(x) \dd x\;, \quad \widehat{\tau}_w^c = \frac{1}{n} \sum_{i=1}^n Y_i \frac{\mathbb{I}\{T_i\in[r, l]\}}{(l-r)\phi(T_i)}\;.
\end{equation*}
\end{example}

\begin{example}[First derivative]\label{ex:first_derivative}
Suppose $Y_i(t)$ is differentiable with $\E Y_i^2(T_i) < 0$ and  $\E |Y_i^\prime(T_i)| < 0$. To learn the first derivative of response functions, we consider $w(t) = t$ and obtain
\begin{equation*}
    \tau_w^c =  \frac{1}{n} \sum_{i=1}^n \int_{\R} Y_i(t) t \phi(t) \dd t \stackrel{\text{(i)}}{=} \frac{1}{n} \sum_{i=1}^n \int_{\R} Y_i^{'}(t) \phi(t) \dd t\;, \quad \widehat{\tau}_w^c = \frac{1}{n} \sum_{i=1}^n Y_i T_i\;,
\end{equation*}
where (i) follows from Stein's Lemma. Notably, if we replace the base Gaussian density $\phi(t)$ with $\psi(t) = \frac{1}{2}\delta_{-1} + \frac{1}{2} \delta_1$, the causal estimand reduces to $\tau_w^c = \frac{1}{2n} \sum_{i=1}^n (Y_i(1) - Y_i(-1))$, which resembles the average treatment effect in binary treatment setups.
\end{example}

\begin{example}[Second derivative]\label{ex:second_derivative}
Suppose $Y_i(t)$ is twice differentiable with $\E Y_i^2(T_i) < 0$ and  $\E |Y_i^{\prime\prime}(T_i)| < 0$.
To learn the second derivative, we consider $w(t) = t^2 - 1$ and obtain
\begin{equation*}
    \tau_w^c =  \frac{1}{n} \sum_{i=1}^n \int_{\R} Y_i(t) (t^2-1) \phi(t) \dd t \stackrel{\text{(i)}}{=} \frac{1}{n} \sum_{i=1}^n \int_{\R} Y_i^{''}(t) \phi(t) \dd t\;, \quad \widehat{\tau}_w^c = \frac{1}{n} \sum_{i=1}^n Y_i (T_i^2 - 1)\;,
\end{equation*}
where (i) follows from an extension of Stein's Lemma \citep{mamis2022stein}. $\widehat{\tau}^c_w$ serves as an unbiased estimator for the average second derivative of response functions.
\end{example}


\subsection{Variance Formula and Measures of Covariate Balance}
To get traction on estimating the variance of the estimators, we decompose $Y_i(t)$ as follows:
\begin{equation}\label{eq:model1}
    Y_i(t) = a_i Y_0(t) + b_i\;, \quad a_i = X_i^\top \beta_1\;, \quad b_i = X_i^\top \beta_2\;.
\end{equation}
In this decomposition, $a_i$ and $b_i$ control the scale and location of the $i$-th response function, and they are perfectly linear in covariates. $Y_0(t)$ is a baseline response function, which is assumed to be known by the researcher. This assumption is justified as researchers often have prior knowledge of the shape of response functions, such as sigmoid dose-response curves in clinical trials \citep{meddings1989analysis}, and exponential utility functions in economics \citep{arrow1971theory}.

Under \eqref{eq:model1}, we analyze the variance of $\widehat{\tau}_w^c$. For two random vectors $X, Y\in\R^{d}$, we use the notation $\Cov(X, Y) := \E[(X - \E X) (Y - \E Y)^\top]$ and $\Cov(X) := \Cov(X, X)$. Then, one can show that under Equation \eqref{eq:model1}, it holds that
\begin{align}
    \Var(\widehat{\tau}_w^c) &= \frac{1}{n^2} \left( \beta_1^\top X^\top  \mathrm{Cov}(Y_0 \circ W) X \beta_1 + \beta_2^\top X^\top \mathrm{Cov}(W) X \beta_2 + 2 \beta_1^\top X^\top \mathrm{Cov}(Y_0 \circ W, W) X \beta_2 \right)\nonumber\\
    &\le \frac{2}{n^2} \left( \beta_1^\top X^\top  \mathrm{Cov}(Y_0 \circ W) X \beta_1 + \beta_2^\top X^\top \mathrm{Cov}(W) X \beta_2 \right)\;.\label{eq:cont_var2}
\end{align}
With a slight abuse of notation, we define $Y_0 = (Y_0(T_1), \dots, Y_0(T_n))^\top$, $W = (W_1, \dots, W_n)^\top$, $\circ$ is the Hadamard (elementwise) product, and the second line follows from the AM-GM inequality. 
From Equation \eqref{eq:cont_var2}, the estimation performance is characterized by quadratic forms similar to the discrete setting (Lemma \ref{lem:discrete_var}). 
In addition, the variance in inequality~\eqref{eq:cont_var2} depends on the coefficients $\beta_1, \beta_2$, which are unknown in general. 

To make progress, we adopt a similar approach as in Section \ref{sec:discrete}. We first assume that $\beta_1, \beta_2$ are random signals with mean zero and identity covariance matrix, which leads to a measure of average-case MSE:
\begin{align*}
\underset{\beta_1, \beta_2}{\E} \Var(\widehat{\tau}_w^c) &\le \frac{2}{n^2}\tr(X^\top (\Cov(Y_0 \circ W) + \Cov(W)) X) \\
&\propto \|X^\top \Cov(Y_0 \circ W) X\|_{\mathrm{nuc}} + \|X^\top \Cov(W) X\|_{\mathrm{nuc}}\;.
\end{align*}
Alternatively, by assuming $\|\beta_1\| \le M$, $|\beta_2\|\le M$, we obtain an upper bound on the worst-case MSE: 
\begin{align*}
    \sup_{\|\beta_1\|\le M, \|\beta_2\|\le M} \Var(\widehat{\tau}_w^c) &\le \sup_{\|\beta_1\|\le M, \|\beta_2\|\le M} \frac{2}{n^2} \left( \beta_1^\top X^\top  \mathrm{Cov}(Y_0 \circ W) X \beta_1 + \beta_2^\top X^\top \mathrm{Cov}(W) X \beta_2 \right)\\
    &\propto\|X^\top \Cov(Y_0 \circ W) X\|_{\mathrm{op}} + \| X^\top \Cov(W)) X\|_{\mathrm{op}}\;.
\end{align*}
These analytical steps lead to the formal definition of covariate balance measures under Gaussian design in the continuous setting.
\begin{definition}
For Gaussian designs with a baseline response function $Y_0$ and 
a weight function $w$, define the average and worst-case covariate balance measures as
\begin{equation}\label{eq:cov_measure_cont}
    \|X^\top \mathrm{Cov}(Y_0 \circ W) X\|_{\mathrm{norm}} + \|X^\top \mathrm{Cov}(W) X\|_{\mathrm{norm}}\;, \mathrm{norm} \in \{\mathrm{nuc}, \mathrm{op}\}\;.
\end{equation}
\end{definition}

\subsection{Gaussianized Representation}
Using Mehler's formula and Hermite coefficients in Section \ref{sec:discrete}, we derive the following result, which is a direct application of Lemma \ref{lem:mehler}.
\begin{proposition}\label{prop:f_formula}
Suppose that $Y_0 w:t\mapsto Y_0(t)w(t)\in L^2_\phi$ and $w \in L^2_\phi$. Then we have
\begin{equation*}
\Cov(Y_0\circ W) = f_{Y_0,w}(\Sigma)\;,\quad \Cov(W) = f_w(\Sigma)\;.
\end{equation*}
Here, $f_{Y_0,w}$ and $f_{w}$ are elementwise functions defined by 
\begin{align*}
f_{Y_0, w}(\rho) = \sum_{m=1}^\infty \alpha_m[Y_0 w]^2 \rho^m\;,\quad 
f_w(\rho) = \sum_{m=1}^\infty \alpha_m[w]^2 \rho^m\;, \quad \rho\in[-1,1]\;,
\end{align*}
where $\alpha_m[g]$ is the $m$-th Hermite coefficient of the function $g$. 
\end{proposition}
Proposition \ref{prop:f_formula} demonstrates that the covariance matrices in covariate balance measures can be explicitly written as functions of $\Sigma$. This result facilitates optimization over $\Sigma$, similar to the role of Proposition \ref{prop:f_formula_general} in the uniform design setup. Combining the results above, we formulate covariate balance measures
\begin{equation*}
    \|X^\top f_{Y_0, w}(\Sigma) X\|_{\mathrm{norm}} + \|X^\top f_{w}(\Sigma) X\|_{\mathrm{norm}}\;, \mathrm{norm} \in \{\mathrm{nuc}, \mathrm{op}\}\;.
\end{equation*}
Consequently, one may directly apply the algorithm proposed in Section \ref{sec:opt} to Gaussian design and optimize the covariate balance. We evaluate this design concretely in Section \ref{sec:real}.

\section{Asymptotics and Inference}\label{sec:asymp}
In this section, we study asymptotic properties and inference under the Gaussianization $T\sim\cN(0, \Sigma)$, where $\Sigma$ is a solution obtained from the PGD-Gauss algorithm in Section \ref{sec:opt}, within the design-based framework. In design-based inference \citep{imbens2015causal}, we view the potential outcomes as fixed and the only randomness comes from the treatment assignment, i.e., the Gaussian treatment $T$.
Here, we prove asymptotic normality under Gaussianization, and provide concrete procedures to compute confidence intervals.
The key takeaway is that Gaussianization under the PGD-Gauss solution results in smaller variance compared to i.i.d. Gaussianization, and thus improves estimation efficiency. Notably, our asymptotic theory allows high-dimensional covariates, where $d$ can grow with, or even be larger than $n$.

For presentation purposes, we focus on the uniform design setup in Section \ref{sec:discrete} where the treatments are modeled by $D_i = g(T_i)$. Inferential procedures for continuous treatments are discussed in Section \ref{sec:appendix_inference}.

\subsection{Asymptotic Normality}
Here, we focus on the PGD-Gauss under the nuclear norm objective $\|X^\top f_k(\Sigma) X\|_{\mathrm{nuc}}$. Recall that $f_k$ defined in Proposition \ref{prop:f_formula_general} is the covariance mapping with respect to treatment $k$, and thus this objective serves as a covariate balance measure for $k$-th average treatment effect. Similar normality results can be shown under the operator norm, but we will focus on the nuclear norm for simplicity.

We study the asymptotic properties of the average treatment effect for arm $k$: 
\begin{equation*}
\widehat{\tau}_k = \frac{K}{n} \sum_{i=1}^n Y_i \mathbb{I}\{D_i = k\}\;,\quad D_i = g(T_i)\;.
\end{equation*}
We focus on implementing one step of PGD-Gauss with step size $\eta$ using the default initialization $\Sigma^1 = I_n$, and denote the obtained solution by $\Sigma_{\eta}$. 
We impose the following assumption on the step size $\eta$ and covariates $X$.
\begin{assumption}\label{asmp:stepsize}
The covariates $X\in\R^{n\times d}$ satisfy $\|X_i\| = 1$, i.e., each row of $X$ has unit norm. The step size in PGD-Gauss satisfies $\eta\|XX^\top - I_n\|_{\op} = o(1)$.
\end{assumption}
Assumption \ref{asmp:stepsize} requires that $\Sigma_{\eta}$ is from a local perturbation of $I_n$ by controlling the step size, which is the key condition to establish asymptotic normality. To better understand the step size condition, we may consider $X_{i} \stackrel{iid}{\sim} \cN(0, \frac{1}{d}I_d)$, so that $\|X_i\| \approx 1$ in expectation. Then, the random matrix theory suggests that $\|XX^\top\|_{\op} = O(n/d)$ with high probability \citep{tropp2015introduction}. If, for intuition, we assume that $n>d$, the step size condition boils down to $\eta = o\left(\frac{d}{n}\right)$.

To characterize the asymptotic distribution under the one-step PGD-Gauss, we define a sequence of ancillary potential outcomes. Specifically, using the $f_k$ in Proposition \ref{prop:f_formula_general}, we define
\begin{equation*}
    \tilde{Y}(k) = f_k(I_n)^{-1/2} f_k(\Sigma_\eta)^{1/2} Y(k)\;.
\end{equation*}
\begin{theorem}\label{thm:normality} 
Suppose Assumption \ref{asmp:stepsize} holds. Consider Gaussianization $T\sim\cN(0, \Sigma_{\eta})$, where $\Sigma_{\eta}$ is the obtained solution from the one-step PGD-Gauss. If, as $n$ goes to infinity, 
\begin{enumerate}
    \item $\max_{i=1, \dots, n} \tilde{Y}_i^2(k) / n \to 0$,
    \item $n\Var(\widehat{\tau}_k) = \frac{K-1}{n} \sum_{i=1}^n \tilde{Y}_i^2(k)$ has a positive limit,
    \item $\|Y(k)\|^2 \le nM$ for some constant $M>0$,
\end{enumerate}
it holds that
$$
\sqrt{n}\left(\widehat{\tau}_k-\tau_k\right) \stackrel{d}{\to} \cN(0, \lim_{n\to\infty} n\Var(\widehat{\tau}_k)) = \cN\left(0, \lim_{n\to\infty}\frac{K-1}{n} \|\tilde{Y}(k)\|^2\right)\;.
$$
\end{theorem}
Theorem \ref{thm:normality} establishes the asymptotic normality of $\widehat{\tau}_k$ under Gaussianization. The proof of Theorem \ref{thm:normality} relies on the asymptotic equivalence between $\widehat{\tau}_k$ under $\Sigma_{\eta}$ and an ancillary estimator under i.i.d. Gaussianization. Due to the asymptotic equivalence, it suffices to prove the asymptotic normality for the ancillary estimator using Lindeberg's central limit theorem. A full proof and a generalization to multi-step PGD-Gauss can be found in Section \ref{sec:proof_normality}.

Notably, the variance term in Theorem \ref{thm:normality} indicates the benefit of running PGD-Gauss for covariate balance. To see this, we may define
\begin{equation*}
    V(\Sigma) \coloneqq \frac{K-1}{n}\|\tilde{Y}(k)\|^2\;.
\end{equation*}
From the proof of Theorem \ref{thm:normality}, $V(\Sigma)$ not only captures the variance limit, but also exactly matches the finite-sample variance of $\sqrt{n}(\widehat{\tau}_k - \tau_k)$, i.e., the MSE of $\widehat{\tau}_k$ after rescaling. The following proposition shows that $V(\Sigma_{\eta})$ is strictly smaller than $V(I_n)$ on average. Denote by $\|A\|_F$ the Frobenius norm of a matrix $A$. Write $a_n = \Omega(b_n)$ if there exists a constant $c$ such that $a_n \ge c b_n$ for $n$ large enough.
\begin{proposition}\label{prop:asymp_var}
Suppose $Y(k) = X\beta_k^\top$, where $\beta_k$ is a random signal with zero mean and identity covariance matrix. In addition, suppose Assumption \ref{asmp:stepsize} holds and 
$$
\eta = o(1)\;,\quad n\eta^3 \|XX^\top - I_n\|_{\op}^4 = o(\|XX^\top - I_n\|_F^2)\;.
$$
If $f_k^\prime(0) \neq 0$, we have
\begin{equation*}
    \E_{\beta_k} V(I_n) - \E_{\beta_k} V(\Sigma_{\eta}) = \Omega\left(\frac{\eta}{n}\|XX^\top - I_n\|_{F}^2\right) > 0\;.
\end{equation*}
\end{proposition}
Proposition \ref{prop:asymp_var} suggests that for $n$ large enough, there is a nonzero improvement in $V(\Sigma_{\eta})$ in the average sense above. Therefore, $\widehat{\tau}_k$ has a smaller variance under $\Sigma_{\eta}$ compared to the initial design $I_n$, which reveals the benefit of covariate balance. However, we clarify that the improvement in Proposition \ref{prop:asymp_var} is with respect to the non-asymptotic variance $V(\Sigma_{\eta})$, which does not directly translate into an improvement in the limiting variance of the asymptotic distribution. Theoretical conditions under which the one-step PGD-Gauss reduces the limiting variance remain an open and complex problem, which we consider as future work.

\subsection{Inference}
To make inference under Gaussianized designs, we need to estimate the variance of $\widehat{\tau}_k$. Under the one-step PGD-Gauss with a covariance matrix $\Sigma_{\eta}$, by Theorem \ref{thm:normality} and Proposition \ref{prop:f_formula_general}, we write the asymptotic variance of $\widehat{\tau}_k$ as
\begin{align*}
    V(\Sigma_{\eta}) &= \frac{K-1}{n} \|\tilde{Y}(k)\|^2 = \frac{K^2}{n} Y(k)^\top f_k(\Sigma_{\eta}) Y(k) = \frac{K^2}{n} \sum_{i,j=1}^n Y_i(k) Y_j(k) f_k(\Sigma_{\eta, ij})\;.
\end{align*}
We use a Horvitz-Thompson estimator to estimate the variance as below: 
\begin{equation}\label{eq:V_HT}
    \widehat{V}_{\eta} = \frac{K^2}{n} \sum_{i,j}^n Y_i Y_j f_k(\Sigma_{\eta, ij}) \frac{\mathbb{I}\{g(T_i) = k, g(T_j) = k\}}{\P(g(T_i) = k, g(T_j) = k)} \;.
\end{equation}
Note that the joint treatment probabilities $\P(g(T_i)=k, g(T_j)=k)$ are determined by the design, and can be computed using $f_k$ based on Proposition \ref{prop:f_formula_general}, i.e, $\P(g(T_i) = k, g(T_j) = k) = f_k(\Sigma_{\eta, ij}) + 1/K^2$.
The following result shows that $\widehat{V}_{\eta}$ is a consistent variance estimator.
\begin{theorem}\label{thm:var_HT}
Suppose that $\max_{i} |Y_i(k)| = O(1)$ and Assumption \ref{asmp:stepsize} holds with $\eta$ satisfying $n^2 \eta^2 \|XX^\top - I_n\|_{\op}^2 = o(1)$. Then, $\widehat{V}_\eta$ is a well-defined variance estimator with 
\begin{equation*}
    \E \widehat{V}_\eta = V(\Sigma_\eta)\;, \quad \Var(\widehat{V}_\eta) = o(1)\;. 
\end{equation*}
\end{theorem}
Theorem \ref{thm:var_HT} enables inference under the Gaussianized design $\Sigma_{\eta}$, as one can combine Theorems \ref{thm:normality} and \ref{thm:var_HT} to derive the design-based confidence interval
\begin{equation}\label{eq:CI}
	[\widehat{\tau}_k - z_{\alpha/2} \sqrt{{\widehat{V}_\eta}/{{n}}}\;, \;\widehat{\tau}_k + z_{\alpha/2} \sqrt{{\widehat{V}_\eta}/{{n}}}]\;.
\end{equation}
Here, we set $z_{\alpha/2} = \Phi^{-1}(1-\alpha/2)$ to obtain an asymptotic $(1-\alpha)$ confidence interval. Compared to Theorem \ref{thm:normality}, Theorem \ref{thm:var_HT} requires a stronger condition $n^2 \eta^2 \|XX^\top - I_n\|_{\op}^2 = o(1)$ on the stepsize $\eta$, as we need to bound higher moments of treatment assignments in the variance estimator.

In this section, we have focused on $\widehat{\tau}_k$ under the one-step PGD-Gauss. However, it is also desirable to construct confidence intervals for $\widehat{\tau}_w$ under a general Gaussian covariance matrix $\Sigma$, which may be obtained by running PGD-Gauss until convergence. To this end, we discuss two general approaches in Section \ref{sec:appendix_inference}: first, we apply existing variance bounding techniques to analyze the variance of $\widehat{\tau}_w$ and compute conservative confidence intervals, similar to \eqref{eq:CI}; second, we propose an alternative procedure to construct randomization-based confidence intervals using the design distribution and an imputation model learned from data.

\section{Simulations}\label{sec:simu}
Here we conduct comprehensive experiments on different designs under a factorial setup. Additional numerical results can be found in Section \ref{sec:appendix_simu}, including a real data example in the continuous treatment setting and simulation details on the 3-treatment setup in Section \ref{sec:example}.

We set $n = 100$, $d = 5$, and $X_i \stackrel{iid}{\sim} \cN(0, I_d)$. Consider a factorial design with two treatments $A_i\in\{0, 1\}$, $B_i\in\{0, 1\}$ with potential outcomes:
$Y_i(A_i, B_i) = X_i^\top \beta_1 + A_i (X_i^\top \beta_2) + B_i (0.2+X_i^\top \beta_3) + 0.5 A_i B_i + \eps_i$, where $\beta_1 = (-1, -1, -2/3, -6/5, 0)$, $\beta_2 = (0,0,-8/5,8/5,8/5)$, $\beta_3 = (2,2,2,0,0)^\top$, $\eps_i\sim\cN(0, 0.1^2)$, and we fix $\eps_i$ for different potential outcomes. To translate the factorial design to a standard uniform design, we encode the treatments by $D_i = 1+2A_i + B_i \in\{1, 2, 3, 4\}$. Then, we apply the Gaussianization techniques in Section \ref{sec:discrete} to model the treatments by $D_i = g(T_i)$ for the map $g$ in Section \ref{sec:discrete}, enabling Gaussianized design optimization.

In the factorial design under the potential outcome framework, one is usually interested in estimating main effects and interaction effects \citep{dasgupta2015causal}: 
\begin{align*}
	\tau_1 &\coloneqq \frac{1}{2n}\sum_{i=1}^n (-Y_i(0, 0) - Y_i(0, 1) + Y_i(1, 0) + Y_i(1,1)) = 0.25 + \frac{1}{n} \sum_{i=1}^n X_i^\top \beta_2\;, \\
	\tau_2 &\coloneqq \frac{1}{2n}\sum_{i=1}^n (-Y_i(0,0) + Y_i(0,1) - Y_i(1,0) + Y_i(1,1)) = 0.45 + \frac{1}{n} \sum_{i=1}^n X_i^\top \beta_3\;, \\
	\tau_{12} &\coloneqq \frac{1}{2n}\sum_{i=1}^n (Y_i(0,0) - Y_i(0,1) - Y_i(1,0) + Y_i(1,1)) = 0.25\;.
\end{align*}
We will estimate these quantities based on Horvitz-Thompson estimators.

We evaluate the MSE of Horvitz-Thompson estimators under different designs. We implement baseline Gaussianization (BG) with $\Sigma = I_n$, and the optimized Gaussianization (OG) with $\Sigma^*$. The optimized covariance matrix $\Sigma^*$ is obtained from PGD-Gauss for solving the nuclear-norm objective with i.i.d. initialization and 200 iterations. For comparison purposes, we implement  complete randomization (CR) \citep{dasgupta2015causal}, recursive matching (RM) \citep{bai2024inference} and rerandomization (RR) \citep{li2020rerandomization}. RM and RR can be considered as state-of-the-art designs for covariate balance in the factorial setup.

The MSEs are presented in boxplots in Figure \ref{fig:mse}, where we evaluate the MSEs based on 1,000 simulations and generate different covariates and potential outcomes for 100 times. We observe across all three estimation problems, OG achieves the smallest MSE among five designs. In Figure \ref{fig:cov_vs_mse}, we provide a scatter plot of the MSEs for $\tau_1$ over the covariate balance objective in the nuclear norm, evaluated under all different designs. We observe that 1) a smaller covariate balance measure indicates smaller MSE on average, and 2) OG achieves the smallest covariate balance measure across all designs, trailed by RM and RR. In Section \ref{sec:appendix_simu}, we provide further simulation details on design-based confidence intervals.

\begin{figure}[t!]
    \vspace{-7mm}
    \centering
    \subfloat[$\tau_1$]{\includegraphics[width=.3\linewidth]{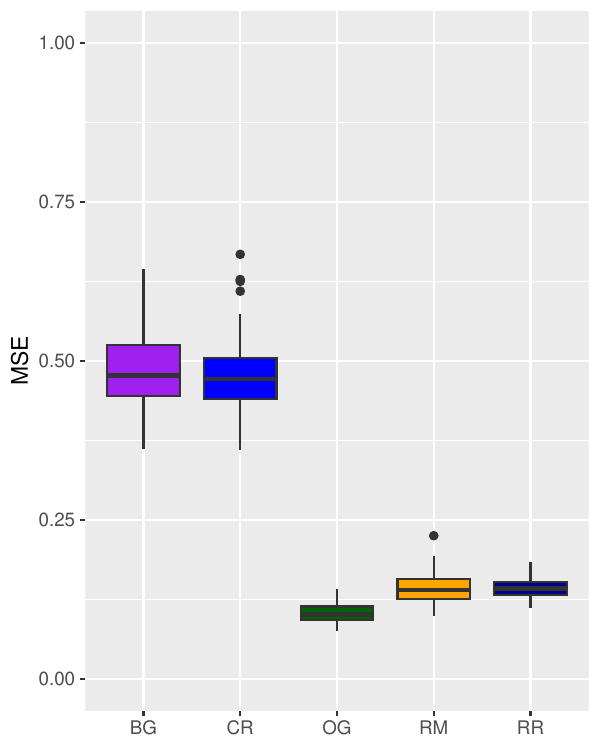}}
    \subfloat[$\tau_2$]{\includegraphics[width=.3\linewidth]{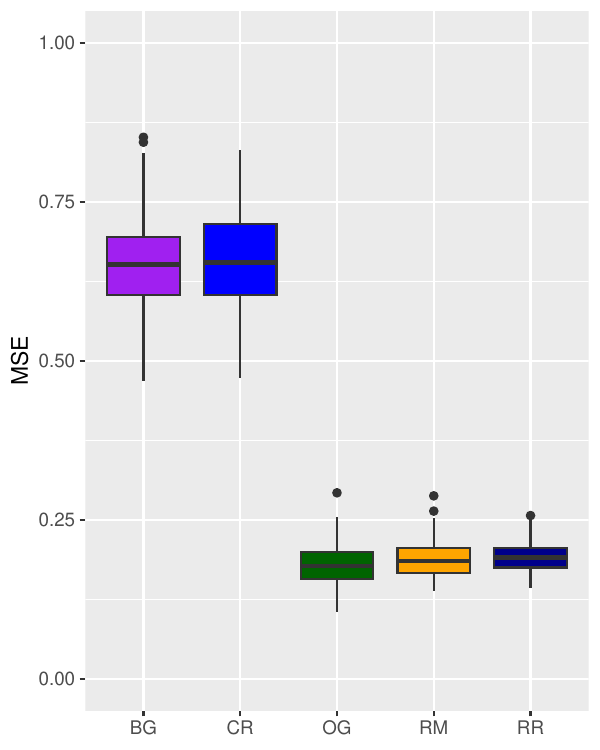}}
    \subfloat[$\tau_{12}$]{\includegraphics[width=.3\linewidth]{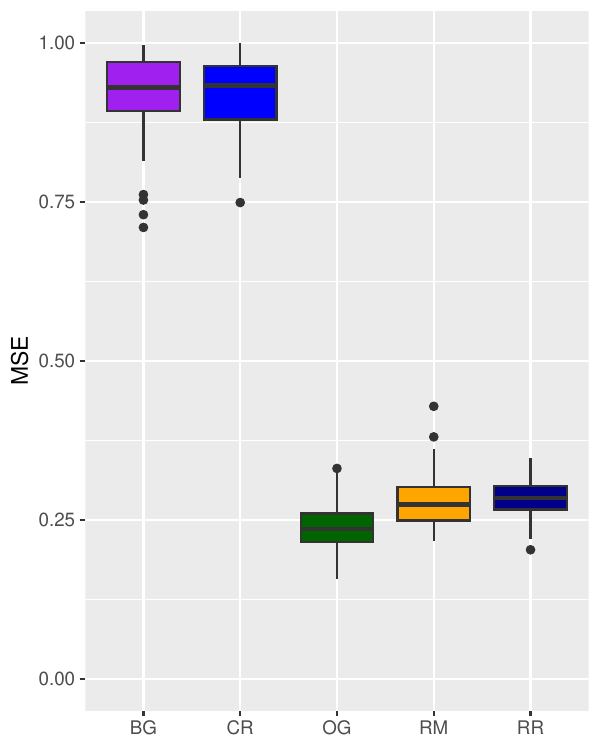}}
    \caption{MSEs for estimating $\tau_1, \tau_2, \tau_{12}$ under different designs.}
    \label{fig:mse}
\end{figure}

\begin{figure}[t!]
    \vspace{-5mm}
    \centering
    \includegraphics[width=.35\linewidth]{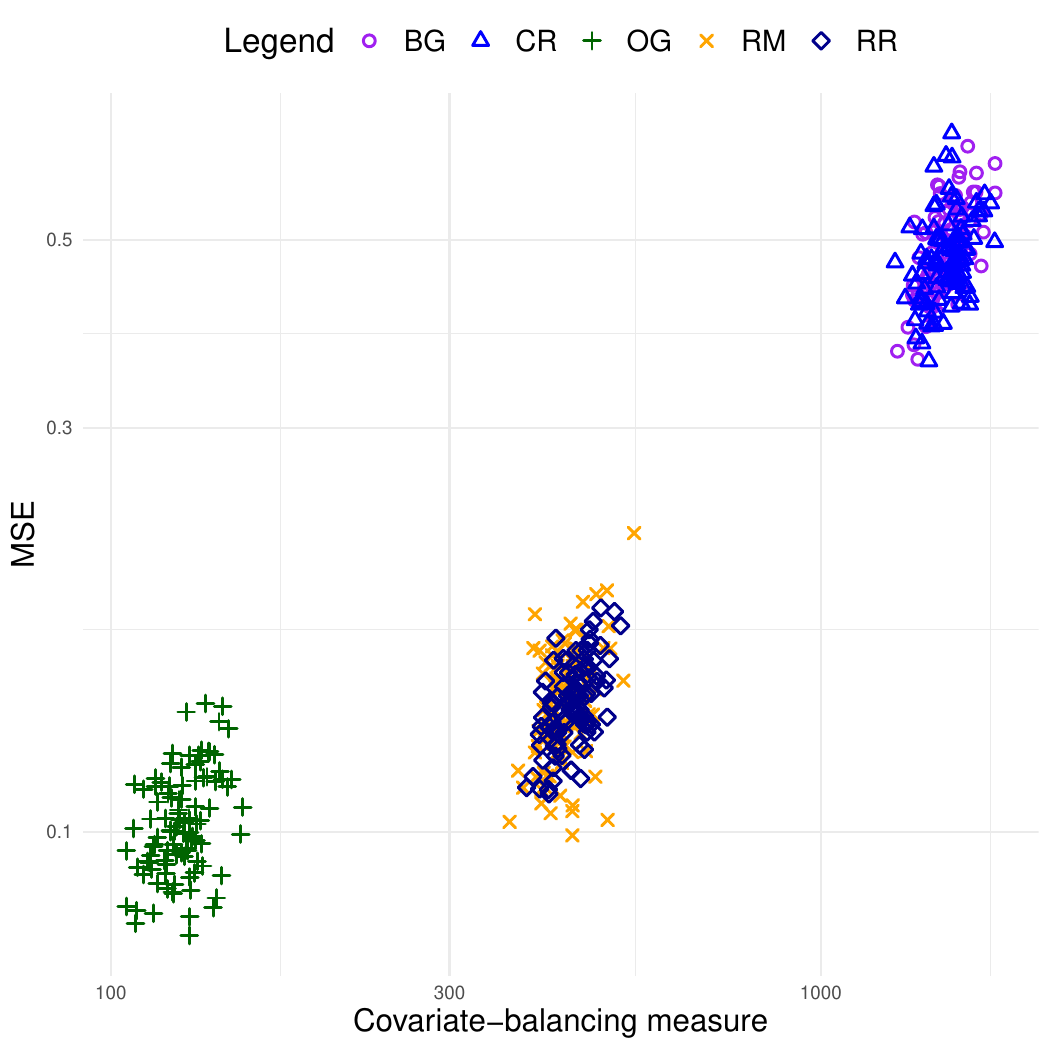}
    \caption{Scatter plot of MSEs over the covariate balance objective (nuclear norm) evaluated under different designs. }
    \label{fig:cov_vs_mse}
\end{figure}

\section{Conclusion}
In our paper, we develop a Gaussianization framework to optimize experimental designs for covariate balance. This approach accommodates general covariates and multiple treatment arms, offering a key advantage over existing methods. Moreover, Gaussianization seamlessly extends to continuous treatments via the Gaussian design, which may be of independent interest in practical applications. 
As an extension, it would be interesting to consider more complex settings, such as those involving interference. Second, developing a general asymptotic theory for Gaussianized designs that extends beyond local perturbations remains an open problem. We consider these areas promising topics for future work.

\section*{Acknowledgments}
TL is supported by the NSF Career Award (DMS-2042473) and by the Wallman Society of Fellows at the University of Chicago. PT acknowledges support from  NSF award SES-2419009.

\bibliography{ref}
\bibliographystyle{natbib}

\newpage
\appendix 

\section{Supplementary Inferential Results}\label{sec:appendix_inference}
In this section, we discuss the extension of Section \ref{sec:asymp} to construct confidence intervals for $\widehat{\tau}_w$ with a general Gaussian covariance matrix. Recall that $\widehat{\tau}_w$ in Section \ref{sec:discrete} is defined as 
\begin{equation*}
    \widehat{\tau}_w = \sum_{k=1}^K w_k \widehat{\tau}_k\;,\quad \widehat{\tau}_k = \frac{K}{n} \sum_{i=1}^n \mathbb{I}\{D_i = k\} Y_i\;.
\end{equation*}

\paragraph{Aronow-Samii Variance Bound.}
By simple algebra, we have
\begin{align*}
    n\Var(\widehat{\tau}_w) &= \sum_{k, l=1}^K w_k w_l \Cov(\widehat{\tau}_k, \widehat{\tau}_l) \\
    &= \frac{K^2}{n} \sum_{k,l=1}^K w_k w_l \sum_{i,j=1}^n Y_i(k) Y_j(l) \Cov(\mathbb{I}\{g(T_i) = k\}, \mathbb{I}\{g(T_j) = l\})\;.
\end{align*}
Different from the $\widehat{\tau}_k$ case, the variance of $\widehat{\tau}_w$ is not directly estimable, as it involves products of different potential outcomes that can not be jointly observed. Concretely, consider the decomposition below:
\begin{align*}
    n\Var(\widehat{\tau}_w) &= \frac{K^2}{n} \sum_{k=1}^K w_k^2 \sum_{i=1}^n Y_i(k)^2 \Var(\mathbb{I}\{g(T_i) = k\}) \\
    &+ \frac{K^2}{n} \sum_{k,l=1}^K w_k w_l \sum_{i\neq j} Y_i(k) Y_j(l) \Cov(\mathbb{I}\{g(T_i) = k\}, \mathbb{I}\{g(T_j) = l\}) \\
    &+ \underbrace{\frac{K^2}{n} \sum_{k\neq l} w_k w_l \sum_{i=1}^n Y_i(k) Y_i(l) \Cov(\mathbb{I}\{g(T_i) = k\}, \mathbb{I}\{g(T_i) = l\})}_\text{(I)}\;.
\end{align*}
The last term (I) is not estimable, since $\P(g(T_i)=k, g(T_i) = l) = 0$ for any $k\neq l$. 

To address this fundamental problem, we follow the design-based inference literature and resort to estimating a variance bound. Specifically, since $\Cov(\mathbb{I}\{g(T_i) = k\}, \mathbb{I}\{g(T_i) = l\}) = -1/K^2$, we have
\begin{align*}
    \text{(I)} = - \frac{1}{n} \sum_{k\neq l} w_k w_l \sum_{i=1}^n Y_i(k) Y_i(l)\;.
\end{align*}
Then, we apply the Anorow-Samii variance bound \citep{aronow2017varbound} to obtain 
\begin{equation*}
    \text{(I)} \le \frac{1}{n} \sum_{k\neq l} |w_k| |w_l| \sum_{i=1}^n \frac{Y_i(k)^2 + Y_i(l)^2}{2}\;.
\end{equation*}
Therefore, we have
\begin{align*}
    n \Var(\widehat{\tau}_w) &\le \text{VB}\;,\\
    \text{VB} &\coloneqq \frac{K^2}{n} \sum_{k=1}^K w_k^2 \sum_{i=1}^n Y_i(k)^2 \Var(\mathbb{I}\{g(T_i) = k\}) \\
    &+ \frac{K^2}{n} \sum_{k,l=1}^K w_k w_l \sum_{i\neq j} Y_i(k) Y_j(l) \Cov(\mathbb{I}\{g(T_i) = k\}, \mathbb{I}\{g(T_j) = l\}) \\
    &+ \frac{1}{2 n} \sum_{k\neq l} |w_k| |w_l| \sum_{i=1}^n (Y_i(k)^2 + Y_i(l)^2)\;.
\end{align*}
Notice that all quantities in the variance bound VB become estimable, and we propose the following conservative variance estimator.\footnote{For simplicity, we assume that $\P(g(T_i)=k,g(T_j)=l)> 0$ for any $i\neq j$. }
\begin{align*}
    \widehat{\mathrm{VB}} &= \frac{K^2}{n} \sum_{k=1}^K w_k^2 \sum_{i=1}^n Y_i^2 \Var(\mathbb{I}\{g(T_i) = k\})\frac{\mathbb{I}\{g(T_i)=k\}}{\P(g(T_i)=k)} \\
    &+ \frac{K^2}{n} \sum_{k,l=1}^K w_k w_l \sum_{i\neq j} Y_i Y_j \Cov(\mathbb{I}\{g(T_i) = k\}, \mathbb{I}\{g(T_j) = l\}) \frac{\mathbb{I}\{g(T_i)=k,g(T_j)=l\}}{\P(g(T_i)=k,g(T_j)=l)} \\
    &+ \frac{1}{2 n} \sum_{k\neq l} |w_k| |w_l| \sum_{i=1}^n Y_i^2 \left( \frac{\mathbb{I}\{g(T_i)=k\}}{\P(g(T_i)=k)} + \frac{\mathbb{I}\{g(T_i)=l\}}{\P(g(T_i)=l)}\right)\;.
\end{align*}
Based on $\widehat{\mathrm{VB}}$, we compute conservative confidence intervals by
\begin{equation*}
	\left[\widehat{\tau}_w - z_{\alpha/2} \sqrt{{\widehat{\mathrm{VB}}}/{{n}}}\;, \;\widehat{\tau}_w + z_{\alpha/2} \sqrt{{\widehat{\mathrm{VB}}}/{{n}}}\right]\;.
\end{equation*}

\paragraph{Randomization-based Confidence Interval.}
In practice, the variance bound above can present a large numerical gap to the true variance, leading to over-conservative confidence intervals. To mitigate the conservativeness, we propose a randomization-based confidence interval as below. This can be viewed as a variant of parametric bootstrap. Recall that $T_i$ and $Y_i$ denote the observed treatment and outcome for unit $i$, respectively.
\begin{testprocedure}[ (Randomization-Based Confidence Interval for $\widehat{\tau}_w$)]\label{proc}
    \begin{enumerate}
	\item For $k = 1, \dots, K$, fit a model $\widehat{m}_k(X_i)$ by regressing $Y_i$ over $X_i$ for all units with treatment $k$.
	\item Generate $\{T^b\}_{b=1}^B\stackrel{iid}{\sim}\cN(0, \Sigma)$. For each randomization $T^b$, impute outcomes by
    \begin{equation*}
    Y_i^b = 
        \begin{cases}
        Y_i &\text{if}~g(T_i^b) = g(T_i)\\
        \widehat{m}_k(X_i)\;, k = g(T_i^b) &\text{if}~g(T_i^b) \neq g(T_i)
        \end{cases}\;.
    \end{equation*}
    Compute the randomization-based estimate $\widehat{\tau}_w^b$ based on $T_i^b$ and $Y_i^b$, $i=1, \dots, n$.
	\item Construct the randomization-based confidence interval $\left[\widehat{c}(\alpha/2)\;,\;\widehat{c}(1 - \alpha/2)\right]$, where $\widehat{c}(\alpha)$ is the $\alpha$-sample quantile for $\{\widehat{\tau}_w^b\}_{b=1}^B$.
\end{enumerate}
\end{testprocedure}
Procedure \ref{proc} conducts simulation-based inference by first learning a regression model to impute all potential outcomes under treatment $k$, and then generating new treatments and outcomes to simulate the distribution of the estimator $\widehat{\tau}_k$. The validity of Procedure \ref{proc} hinges on step 1, i.e., how well the fitted model captures the true outcome functions, which will be numerically validated in Section \ref{sec:appendix_simu}.

Conceptually, Procedure \ref{proc} follows similar ideas as \cite{imbens2018bootstrap}, which introduce a causal bootstrap to construct confidence intervals for the average treatment effect. However, \cite{imbens2018bootstrap} focused exclusively on a binary treatment setting under complete randomization, whereas Procedure \ref{proc} accommodates multiple treatment arms and general Gaussianized designs.

Notably, Procedure \ref{proc} can be extended to the continuous setting, which will be used to construct confidence intervals in Section \ref{sec:real}.
\begin{testprocedure}[ (Randomization-Based Confidence Interval for $\widehat{\tau}_w^c$)]\label{proc2}
    \begin{enumerate}
	\item Fit a model $\widehat{m}(X_i, T_i)$ by regressing $Y_i$ over $X_i$ and $T_i$, $i=1, \dots, n$.
	\item Generate $\{T^b\}_{b=1}^B\stackrel{iid}{\sim}\cN(0, \Sigma)$. For each randomization $T^b$, impute outcomes by
    \begin{equation*}
    Y_i^b = 
        \begin{cases}
        Y_i &\text{if}~T_i^b = T_i\\
        \widehat{m}(X_i, T_i^b) &\text{if}~T_i^b \neq T_i
        \end{cases}\;.
    \end{equation*}
    Compute the randomization-based estimate $\widehat{\tau}_w^{c,b}$ based on $T_i^b$ and $Y_i^b$, $i=1, \dots, n$.
	\item Construct the randomization-based confidence interval
    \begin{equation*}
	\left[\widehat{c}(\alpha/2),~\widehat{c}(1 - \alpha/2)\right]\;.
    \end{equation*}
    Here, $\widehat{c}(\alpha)$ is the $\alpha$-sample quantile for $\{\widehat{\tau}_k^{c,b}\}_{b=1}^B$.
\end{enumerate}
\end{testprocedure}

\section{Additional Simulation Results}\label{sec:appendix_simu}
\subsection{Simulation Details of the 3-treatment Experiment}
For the example in Section~\ref{sec:example}, we consider the following setup. 
\begin{itemize}
    \item $n = 18, d = 5$.
    \item For covariates, we consider (a) single feature: $X_{i1}\sim\cN(2, 3^2)$ and $X_{ij}\sim\cN(0,0.1^2)$ for $j=2, \dots, d$. $\beta_{1k}\sim 2 + 2\exp(1)$ and $\beta_{jk}\sim 2\exp(1)$; (b) uniform covariates: $X_{ij}\sim\cN(0, 3.6^2)$ and $\beta_{jk}\sim 2\exp(1)$. Here, $\beta_{jk}$ denotes the $j$-th entry of the vector $\beta_k$ for $j\in\{1,\dots,d\}$ and $k\in\{1,2,3\}$, and $\exp(1)$ is the exponential distribution with the rate parameter equal to one.
    \item Generate potential outcomes based on $Y(k) = X \beta_k$. 
\end{itemize}
In the block initialization, we first construct the size-3 blocks by sorting the first coordinate of $X$. Then, for each block matrix, we set diagonals to be 1 and off-diagonal entries to be -0.5. We run the PGD-Gauss in Section \ref{sec:opt} for 200 iterations.

\subsection{Simulation Details under the Factorial Setup}
In CR, one assigns same number of units to different treatments uniformly at random, which serves as a baseline that does not leverage covariate information. In RR, we repeatedly generate treatment assignments from CR according to the covariate balance criteria on Mahalanobis distance  with the asymptotic acceptance probability $p_a = 0.01$ as defined in \cite[Section 4]{li2020rerandomization}. In RM, we recursively match the experimental units for different treatment factors following \citep{bai2024inference}.

Here, we follow the simulation setup in Section \ref{sec:simu} and compare the computable confidence intervals under different designs. The confidence intervals for BG and OG can be constructed based on Procedure \ref{proc} in Section \ref{sec:appendix_inference}, where we fit a linear model $\widehat{m}_k$ of outcomes over covariates for each treatment arm. For RR and CR, we adopt the variance estimators proposed in \citep{li2020rerandomization,dasgupta2015causal}, respectively, and construct confidence intervals based on asymptotics. Here, we exclude the recursive matching design (RM) because, although it is a powerful design, the inferential results in \cite{bai2022inference} are derived in a superpopulation framework which is distinct from our design-based framework.

We present in Figure \ref{fig:confidence} the boxplots of the width of confidence intervals, along with the coverage rates. For simplicity, we focus on $\tau_1$, as the results for $\tau_2$ and $\tau_{12}$ are similar. Note that all methods achieve a correct coverage of $95\%$, while some of them are conservative. In term of the width, we observe that OG returns shortest confidence intervals, which reveals practical benefits of our design.
\begin{figure}[htbp!]
    \vspace{-7mm}
    \centering
    \subfloat[Width]{\includegraphics[width=.45\linewidth]{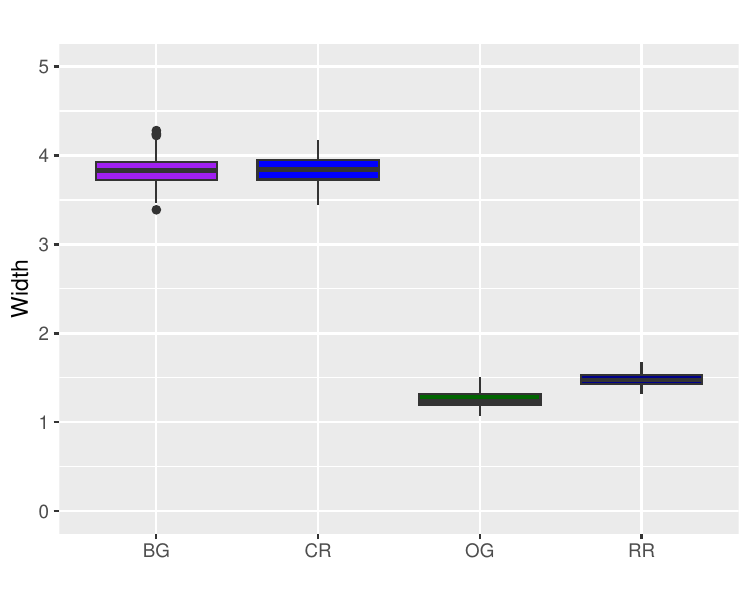}}
    \subfloat[Coverage]{\includegraphics[width=.45\linewidth]{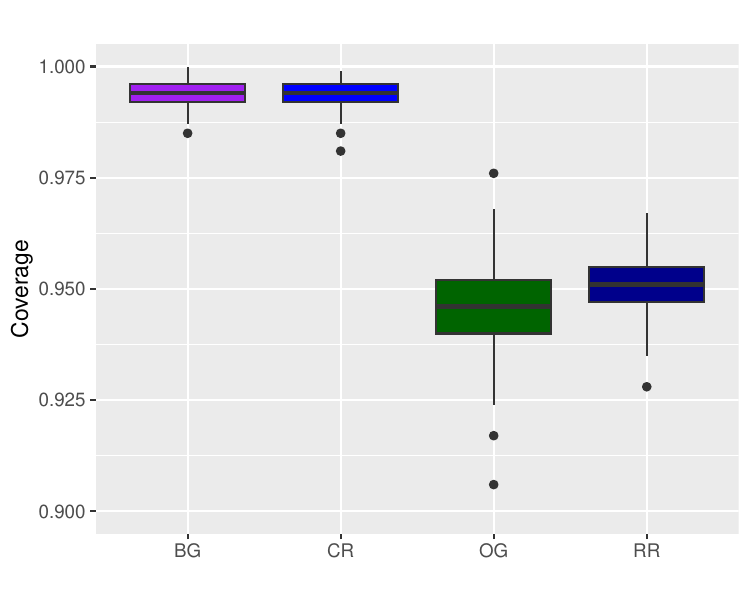}}
    \caption{Width of confidence intervals and their coverage rates for $\tau_1$ under different designs.}
    \label{fig:confidence}
\end{figure}

\subsection{Bed Nets Study on Continuous Treatments}\label{sec:real}
In \cite{Dupas2014}, the authors conducted a field experiment in Kenya, where households in different regions were encouraged to purchase insecticide-treated bed nets designed to prevent malaria. \cite{Dupas2014} treated households by sending vouchers with different discounted prices for the bed nets, effectively inducing a continuous price variable. The original outcome was a binary variable indicating whether a household purchased the bed nets using the voucher, and \cite{Dupas2014,gerber2012field} analyzed the effect of voucher on purchase rates of bed nets. Here, we implement Gaussian designs to assign continuous price treatments, and evaluate their performance compared with the original design in their study.

The original experiment in \cite{Dupas2014} was a 2-stage randomization (referred to as 2S), which fixes the price variable at discrete levels: 
\begin{itemize}
\item Stage 1 (region-level): Assign treatment levels (discounted prices) for six different regions in Kenya. These values are fixed once assigned throughout the experiment.
\item Stage 2 (household-level): Randomly assign treatments for households in each region, with the treatment levels determined in Stage 1.
\end{itemize}
The design and results of the bed nets study are presented in Table \ref{tab:sales}, which reports the proportion of households which purchased a bet net given a region and a discounted price. For instance, in region 2 with price 40, there were 61 households who received the voucher and 75.4\% of them eventually redeemed the voucher and purchased bed nets. Clearly, the rate at which bed nets were purchased declines steadily as the price increases: 75.4\% of households offered a price of 40 shillings purchased a net, compared to only 17.0\% of those offered a price of 200 shillings.

\renewcommand{\arraystretch}{1.1}
\begin{table}[t]
    \centering
    \begin{tabular}{p{9.5em}p{4em}p{4em}p{4em}p{4em}p{4em}p{4em}}
    \hline
	\hline
	Price \\(in Kenyan Shillings) & Region 1 & Region 2&  Region 3&  Region 4&  Region 5&  Region 6 \\
    \hline
	0 & &&&& 96.9(64) & 98.1 (53) \\
	\hline
	40 && 75.4(61) & & \\
	\hline
	50 &&& 72.4 (58) & 40.0 (35) & \\
	\hline
	60 &&&&& 73.0(37) & \\
	\hline 
	70 & 55.2(29) &&&&& \\
	\hline
	80 & & 57.1(70) &&&& \\
	\hline
	90 &&& 55.0(60) &&& \\
	\hline 
	100 & 34.0(47) &&& 28.6(49) & & 61.1(18) \\
	\hline 
	110 &&&&& 32.4(37) & \\
	\hline 
	120 & & 28.1(64) &&&& \\
	\hline 
	130 & 24.5(49) &&&&& \\
	\hline 
	140 &&&&& 37.9(29) & \\
	\hline 
	150 &&& 31.0(58)& 35.6(45)& & 22.2 (18) \\
	\hline
	190 & 17.9(28) &&&&& \\
	\hline 
	200 & & 17.0(59) && 10.3(29) &&\\
	\hline 
	210 & && 18.8 (48) &&&\\
	\hline
	250 & 6.7(30) &&& 7.7 (26) &&\\
	\hline
    \end{tabular}
    \bigskip
    \caption{Rates at which anti-malaria bed nets are purchased, by sales price (after subtracting the value of a randomly assigned voucher). The total number of households per group is in parentheses, and the exchange rate at the time of this study was 65 shillings = \$1.00.
    }
    \label{tab:sales}
\end{table}

\subsubsection{Estimation of Linear Effect}\label{sec:real_linear}
In our numerical study, we define each experimental unit as a cluster of households corresponding to a data point in Table \ref{tab:sales}, with outcome defined as the proportion of households who purchased bed nets. This results in an experiment on 26 cluster-level units. Each unit $i$ has a dummy covariate vector $U_i\in\R^{6}$ indicating the region of the unit, and a cluster-level covariate vector $X_i\in\R^{3}$, including the proportion of male heads sampled to receive the voucher, the proportion of households that have ever shopped at the shop, and the average age of the female heads in households. The three covariates are selected due to their statistical significance in an OLS regression of outcome over all collected covariates.


To assess the performance of different designs, we need to impute the outcome value at any counterfactual price level. To this end, we use the following imputation model:
\begin{equation*}
Y_i(t) = X_i^\top {\alpha}_1 + U_i^\top \alpha_2 + U_i^\top {\beta} t + \eps_{i}\;,
\end{equation*}
The coefficients ${\alpha}_1, \alpha_2, {\beta}$ are OLS estimates for this linear model based on the observed data, and $\eps_i\stackrel{iid}{\sim}\cN(0, \sigma^2)$ where $\sigma^2$ is the OLS estimate of the error variance. Our goal is to estimate the average linear treatment effect under the imputation model
\begin{equation*}
	\tau_L^c \coloneqq \frac{1}{n} \sum_{i=1}^n U_i^\top {\beta}\;.
\end{equation*}
Under the original design, one can unbiasedly estimate $\tau_L^c$ by
\begin{equation*}
	\widehat{\tau}_{2S} = \frac{1}{n} \sum_{i=1}^n Y_i^{2S} \sum_{j=1}^6 U_{ij} \frac{D_i^{2S}-\mu_j}{\sigma_j^2}
\end{equation*}
$D_i^{2S}$ in $\widehat{\tau}_{2S}$ denotes the treatment in the 2-stage (2S) design, i.e., randomly selected from the discrete set of price levels for each region. Accordingly, $\mu_j$ and $\sigma_j^2$ are the mean and variance of $D_i^{2S}$ in region $j$.

We discuss implementation details about Gaussian design toward estimating $\tau_L^c$. First, since the price treatment takes values in $[0, 250]$, we implement $T_i\sim \cN(\mu, \sigma^2)$ with $\mu = 125$ and $\sigma = 41.67 = 250/6$, ensuring that that $T_i$ falls in $[0, 250]$ with high probability. Then, noticing that the estimand $\tau_L^c$ is same as the average of $Y_i^\prime(t)$, we follow Example \ref{ex:first_derivative} to obtain an unbiased estimator
\begin{equation*}
	\widehat{\tau}_{L}^c = \frac{1}{n} \sum_{i=1}^n Y_i w_L(T_i)\;, \quad w_L(t) =  \frac{t-\mu}{\sigma^2}\;.
\end{equation*}
Lastly, to perform Gaussianized design optimization in Section~\ref{sec:continuous}, we specify a linear baseline response function
\begin{equation}\label{eq:baseline}
	Y_0(t) = - \frac{t}{250} + 1\;.
\end{equation}
The specified response function captures the true imputation model in the sense that
\begin{equation*}
    Y_i(t) = a_i Y_0(t) + b_i \;, \quad a_i = -250 U_i^\top \beta~\text{and}~b_i = X_i^\top \alpha_1 +  U_i^\top \alpha_2 + \eps_i - 250 U_i^\top \beta\;,
\end{equation*}
which validates the modeling assumption in Equation \eqref{eq:model1}. We initialize PGD-Gauss from i.i.d. Gaussian design with covariates $\{(X_i, U_i)\}_{i=1}^{26}$, and obtain the optimized Gaussian design after 200 iterations. We focus on the baseline response function \eqref{eq:baseline}, i.i.d. initialization, and the nuclear norm objective in design optimization throughout the bed nets study.

Table \ref{tab:sales_lin} presents the MSE and inference properties for different designs. We implement the baseline i.i.d. Gaussian design (BG) and the original 2S design for comparison. To conduct inference, we use the randomization-based confidence intervals (Section~\ref{sec:appendix_inference}).\footnote{We apply Procedure \ref{proc2} by fitting a linear model $\widehat{m}(x,t)$ of outcomes over treatments and all covariates, i.e., $Y_i\sim X_i + U_i + T_i$. Note that the original inference procedure in \cite{Dupas2014} is no longer applicable under our setup, as we are considering a different causal estimand.}
We observe that OG achieves the smallest MSE as well as the shortest confidence interval. We observe a numerical gap between the actual coverage rates and the expected 95\% coverage for BG, which is due to the small sample size ($n = 26$).
\begin{table}[htbp!]
    \centering
    \begin{tabular}{lp{5em}p{8em}p{7em}}
	\hline
    design & $n\times$ MSE & average CI width $\times\sqrt{n}$ & coverage (\%) \\
	\hline 
	\multicolumn{3}{l}{$\tau_L^c = -3.75 \times 10^{-3}$} & \\
	\hline
	BG &  $1.2\times 10^{-4}$ & $4.00\times 10^{-2}$ & 90.7\\
	OG & $0.5\times 10^{-4}$ & $2.59\times 10^{-2}$ & 97.9\\
	2S & $0.8\times 10^{-4}$ & $3.26\times 10^{-2}$ & 100.0\\
	\hline
    \end{tabular}
    \bigskip
    \caption{MSE properties and inference for linear effects based on 1,000 simualtions. }
    \label{tab:sales_lin}
\end{table}

We visualize the optimized Gaussian covariance matrix in Figure \ref{fig:corr_map}. The covariance matrix ---initialized from the identity matrix--- automatically learns the block structure for units from regions 1-6 under PGD-Gauss. In addition, within each block, it reveals an approximate equicorrelation structure, which resembles the covariance matrix of complete randomization. In short, the OG design in this setup performs a continuized block randomization.
\begin{figure}[t!]
    \centering
    \includegraphics[width=.45\linewidth]{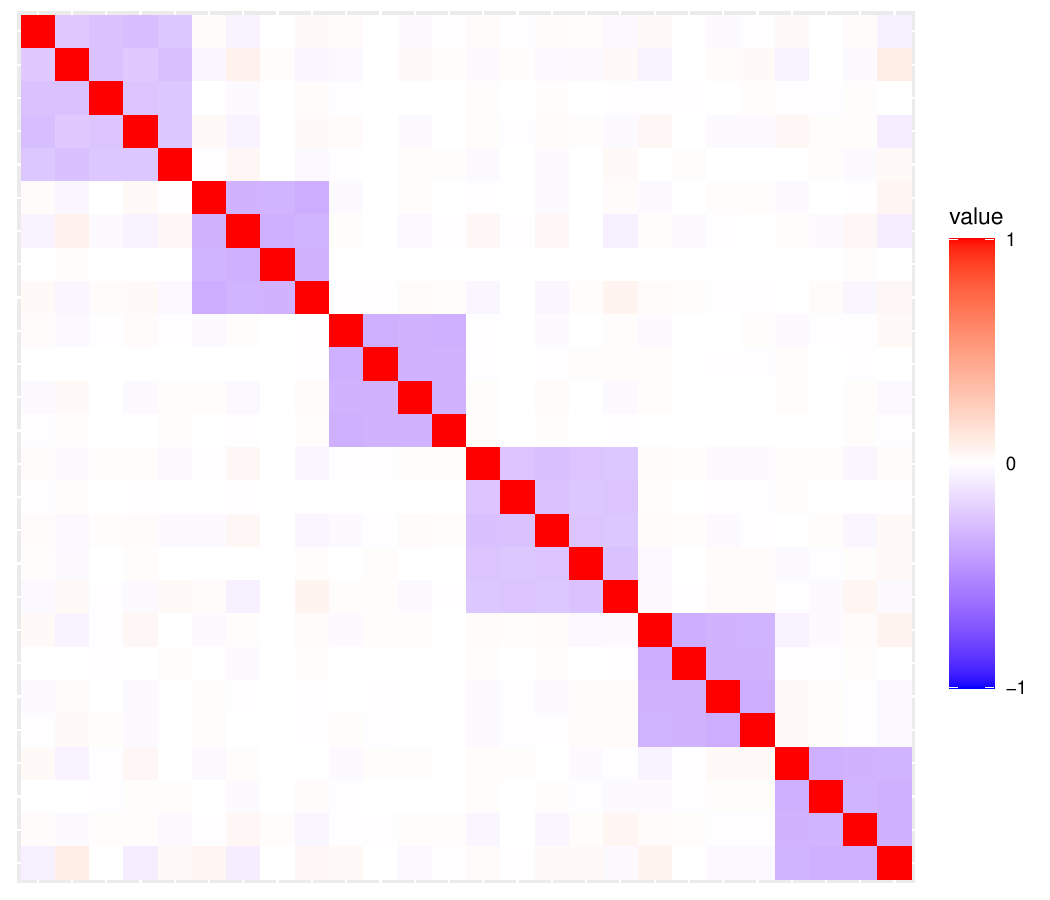}
    \caption{Heatmap of the optimized Gaussian covariance in OG.}
    \label{fig:corr_map}
\end{figure}

\subsubsection{Testing Monotonicity and Convexity}
\paragraph{Monotonicity.}
Testing the monotonicity, say, non-decreasingness, can be formulated as the following hypothesis on underlying response functions:
\begin{equation*}
	H_0^M: Y^{\prime}_i(t)\ge 0\;, ~\text{for any $i=1,\dots,n$ and $t\in\R$}.
\end{equation*}
Directly testing for $H_0^M$ is impossible, since we have only one observation for each response function. We consider a weaker null hypothesis of $H_0^M$
\begin{equation*}
	H_{0,g}^M: \frac{1}{n} \sum_{i=1}^n \underset{Z\sim\cN(\mu, \sigma^2)}{\E} Y^{\prime}_i(Z)\ge 0\;.
\end{equation*}
This weak null hypothesis is motivated by Gaussian design, and it indicates that the derivative averaged over units and treatments is non-negative. The design-induced hypothesis $H_{0,g}^M$ allows us to check monotonicity through Gaussian design.

Similar to Section \ref{sec:real_linear}, we consider an imputation model with a nonlinear component
\begin{equation*}
Y_i(t) = X_i^\top {\alpha}_1 + U_i^\top \alpha_2 + b U_i^\top {\beta} t^3 + \eps_i\;,
\end{equation*}
The coefficients ${\alpha}_1, \alpha_2, {\beta}$ are OLS estimates for this linear model based on the observed data and $b=1$, and $\eps_i\stackrel{iid}{\sim}\cN(0, \sigma^2)$ where $\sigma^2$ is the OLS estimate of the error variance. We report that each element of ${\beta}$ is negative, indicating that the null hypothesis $H_0^M$ is false. To evaluate the power under different degree of monotonicity, we inspect $b = 0, 0.5, 1, 1.5, 2$, where a larger $b$ indicates more significant decreasingness in the data.

Same as Section \ref{sec:real_linear}, under Gaussian design, one can use $\widehat{\tau}^c_L$ to unbiasedly estimate 
\begin{equation*}
    \tau_M^c \coloneqq \frac{1}{n} \sum_{i=1}^n \underset{Z\sim\cN(\mu, \sigma^2)}{\E} Y^{\prime}_i(Z)\;,
\end{equation*}
which is guaranteed by Example \ref{ex:first_derivative}. Hence, we implement BG and OG to test for $H_{0, g}^M$ by checking whether the computed confidence interval for $\tau_M^c$ is below zero. Confidence intervals are computed in the same way as in Section \ref{sec:real_linear}.\footnote{In Procedure \ref{proc2}, we fit a linear model $\widehat{m}(x, t)$ based on $Y_i \sim X_i + U_i + U_iT_i^3$.} For comparison purposes, we employ a parametric approach that first fits an OLS regression on 
\begin{equation*}
Y_i \sim X_i + U_i + T_i\;,
\end{equation*}
and then applies a $t$-test for $T_i$ as a surrogate method to check monotonicity. We evaluate the parametric linear model approach (LM) under all three designs BG, OG, and 2S.

From Figure \ref{fig:monotonicity}(a), OG is more powerful for testing $H_{0,g}^{M}$ compared to BG, justifying the benefits of covariate balance. Under the LM approach, the original 2S design provides the highest power. However, we note that LM approaches are not directly comparable with BG and OG approach, as they target the null hypothesis that whether the OLS coefficient is negative, which is different from $H_{0,g}^M$.
\begin{figure}[t!]
    \vspace{-7mm}
    \centering
    \subfloat[Monotonicity]{\includegraphics[width=.45\linewidth]{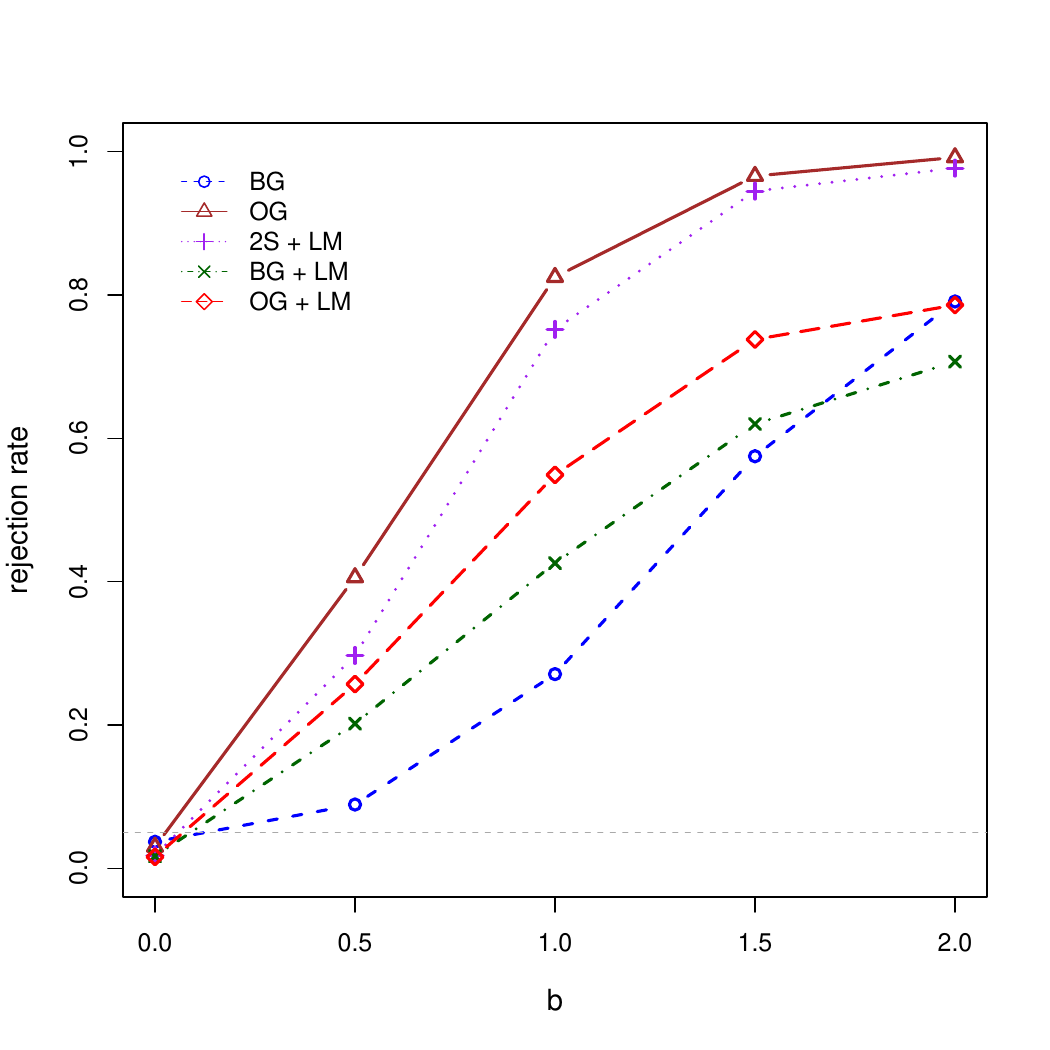}}
    \subfloat[Convexity]{\includegraphics[width=.45\linewidth]{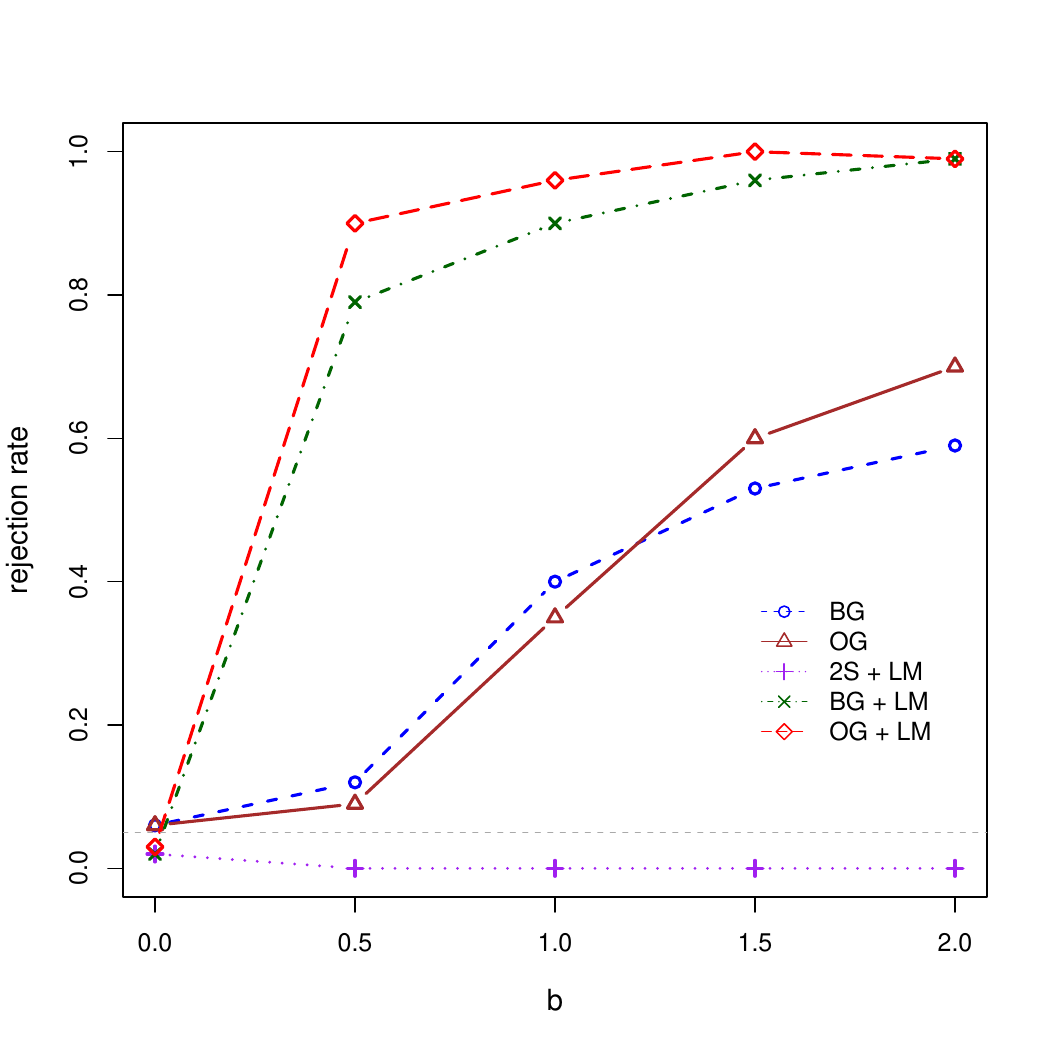}}
    \caption{Rejection rates for testing monotonicity and convexity over different $b$. Rejection means that the confidence interval for the parameter of interest is strictly below zero.}
    \label{fig:monotonicity}
\end{figure}

\paragraph{Convexity.}
Similar to the monotonicity case, we test $H_0^C : Y_i^{''}(t) >0$ for any $i$ through a weaker null hypothesis
\begin{equation*}
	H_{0,g}^C: \frac{1}{n} \sum_{i=1}^n \underset{Z\sim\cN(\mu, \sigma^2)} {\E} Y^{''}_i(Z) \ge 0\;.
\end{equation*}
We consider an imputation model 
\begin{equation*}
	Y_i(t) = X_i^\top {\alpha}_1 + U_i^\top\alpha_2 + b U_i^\top {\beta} t^2 + \eps_i \;.
\end{equation*}
The coefficients ${\alpha}_1, \alpha_2, {\beta}$ are OLS estimates for this linear model based on the observed data and $b=1$, and $\eps_i\stackrel{iid}{\sim}\cN(0, \sigma^2)$ where $\sigma^2$ is the OLS estimate of the error variance. Since each element of $\beta$ is negative, the imputation model implies that the null hypothesis $H_0^C$ is false, i.e., the response functions are concave. Again, we inspect $b = 0, 0.5, 1, 1.5, 2$, where a larger $b$ indicates more significant concavity in the data.

Convexity reflects the second-order information of the response functions, which typically requires larger sample sizes to gain any meaningful conclusions. Hence, to make nontrivial power comparisons, we simulate a new set of covariates of size $n = 500$, by sampling uniformly from the original covariates of 26 samples. The covariates are fixed once generated. This ends up with a new experimental setup with 500 units.

Under Gaussian designs, we compute the following estimator
\begin{equation*}
	\widehat{\tau}^c_{C} = \frac{1}{n}\sum_{i=1}^n Y_i w_C(T_i)\;,\quad w_C(t) = \frac{((t-\mu)^2/\sigma^2-1)}{\sigma^2}\;,
\end{equation*}
which is an unbiased estimator of $\tau_C^c \coloneqq \frac{1}{n}\sum \E Y_i^{''}(Z)$ based on Example \ref{ex:second_derivative}.\footnote{The subscript C denotes the convexity, whereas the superscript c denotes the continuous setting.} Hence, we implement BG and OG to test for $H_{0, g}^C$ by checking whether the computed confidence interval for $\tau_C^c$ is below zero.\footnote{In Procedure \ref{proc2}, we fit a linear model $\widehat{m}(x, t)$ based on $Y_i \sim X_i + U_i + U_iT_i^2$.} For comparison purposes, we implement a parametric approach that fits a linear regression model
\begin{equation*}
    Y_i \sim X_i + U_i + T_i + T_i^2
\end{equation*}
and applies the $t$-test for the coefficient of $T_i^2$ to check for convexity. We evaluate the parametric linear model approach (LM) under BG, OG, and 2S.

From Figure \ref{fig:monotonicity}(b), OG and BG have similar performance, and OG achieves higher power only for $b = 1.5, 2$. This is because the optimized covariance matrix for OG is numerically similar to that for BG, the identity matrix, as we will explain below. Among all methods, OG combined with LM (OG + LM) yields highest power. Note that the LM approach under the original design fails to reject convexity, as the 2S design focuses on discrete treatment values, making it difficult to probe the concave structure.

\paragraph{Estimands and Optimized Gaussian Designs.}
We conclude our numerical study by showing how different estimands lead to different structures in the optimized covariance matrix of OG. In Figure \ref{fig:corr}(a)-(b), we visualize the function $f$ in the covariate balance objective $\|X^\top f(\Sigma) X\|_{\nuc}$ for monotonicity and convexity. That is, based on Section \ref{sec:continuous}, we compute
\begin{equation*}
    f(\rho) \coloneqq f_{Y_0, w}(\rho) + f_{w}(\rho)\;, \quad \rho \in [-1, 1]\;,
\end{equation*}
where $w$ corresponds to $w_L$ and $w_C$ defined before, and $Y_0$ is the linear baseline response function \eqref{eq:baseline}. Observe that they are approximately linear and quadratic functions. In the second row, we visualize the scatter plot for the off-diagonal entries $\Sigma_{ij}$ in the optimized covariance and a pairwise covariate-similarity $X_i^\top X_j / \|X_i\|\|X_j\|$ for all $i\neq j$. In (c), the optimized design balances the covariates by assigning negative correlations to pairs of units with higher similarities. In (d), the optimized covariance  assigns a constant correlation (a negligible negative value) to all pairs of units, and hence the optimized design performs similarly as the i.i.d. Gaussian design, as seen in Figure~\ref{fig:monotonicity}(b). It is because $f^\prime(0)$ is almost zero in Figure \ref{fig:corr}(b), and thus the PGD-Gauss algorithm stops at the identity matrix, which is already a local optimizer.
\begin{figure}[t!]
    \vspace{-7mm}
    \centering
    \subfloat[Function $f$ for monotonicity]{\includegraphics[width=.35\linewidth]{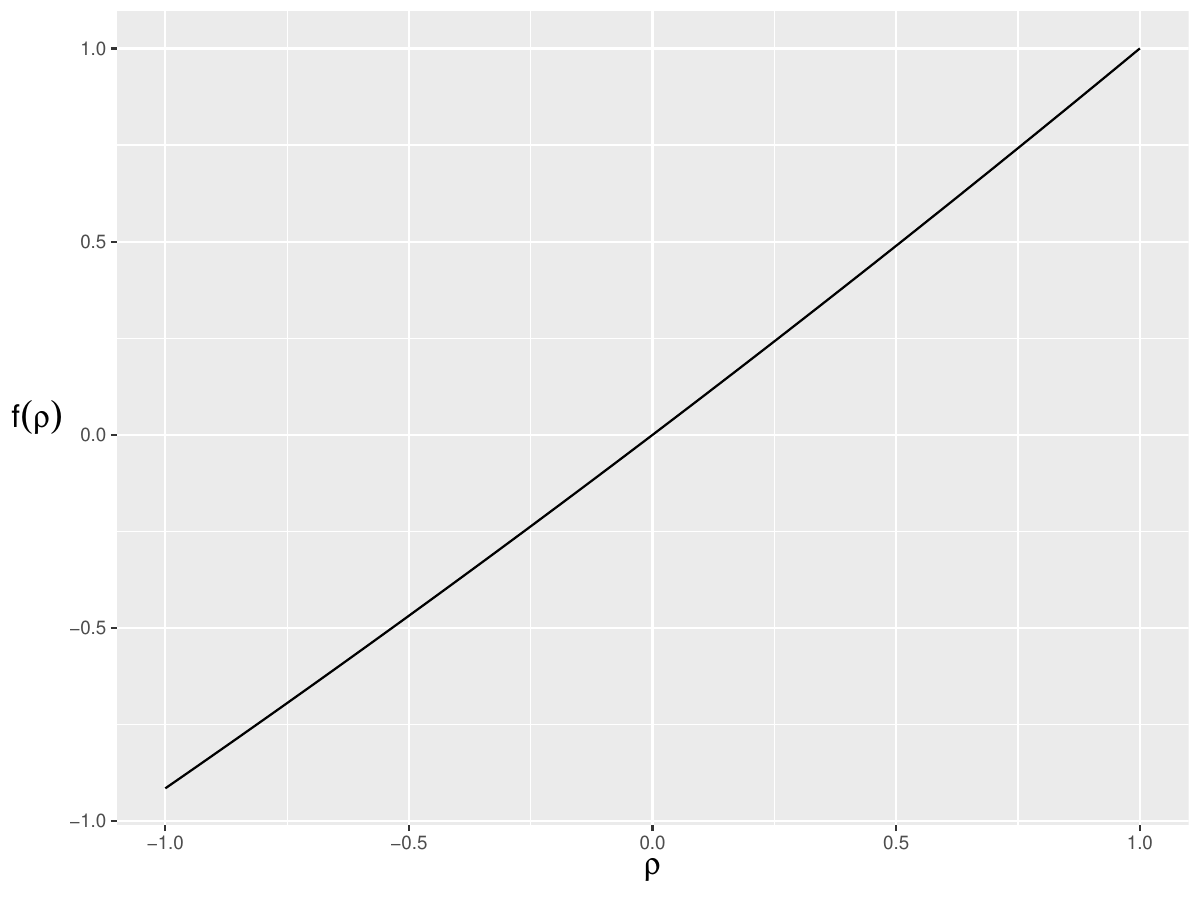}}
	\subfloat[Function $f$ for convexity]{\includegraphics[width=.35\linewidth]{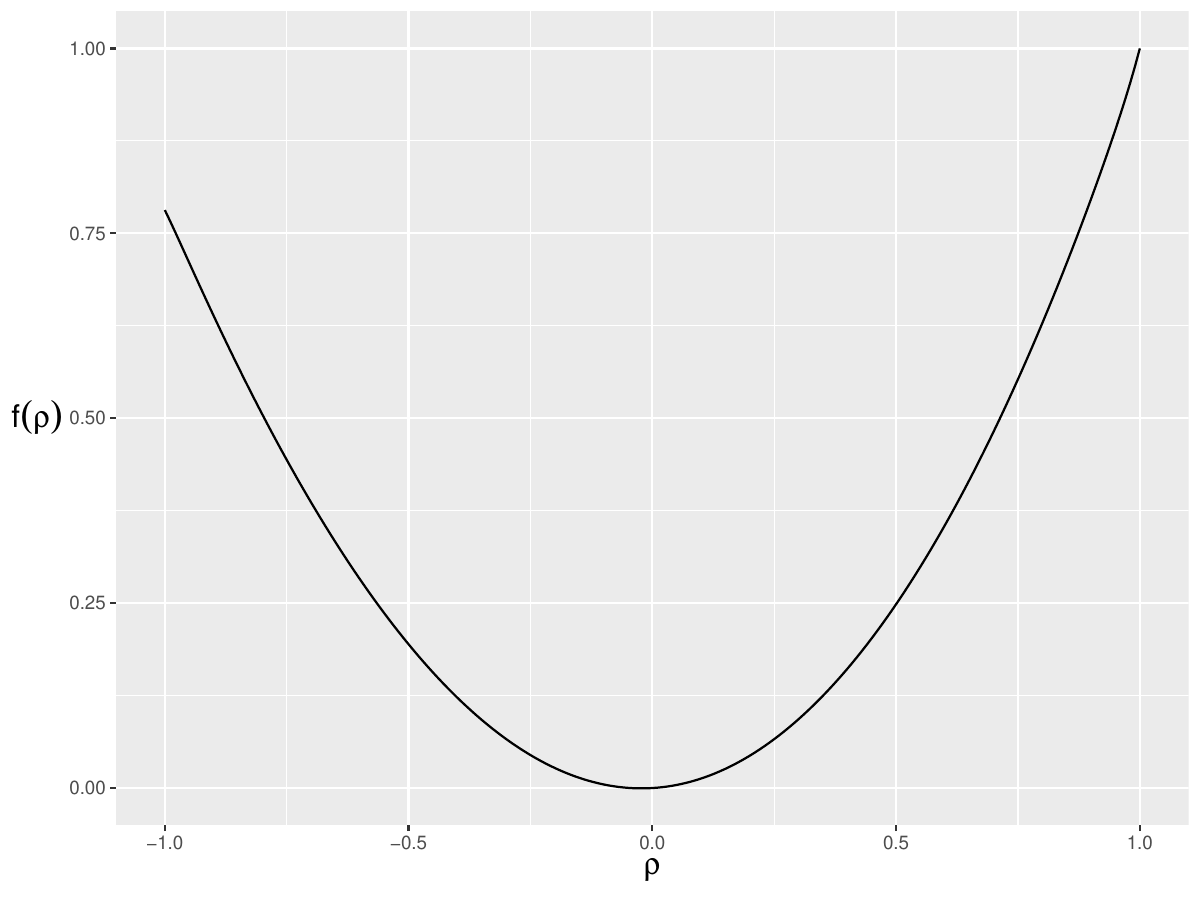}}\\
	\subfloat[Scatter plot for monotonicity]{\includegraphics[width=.35\linewidth]{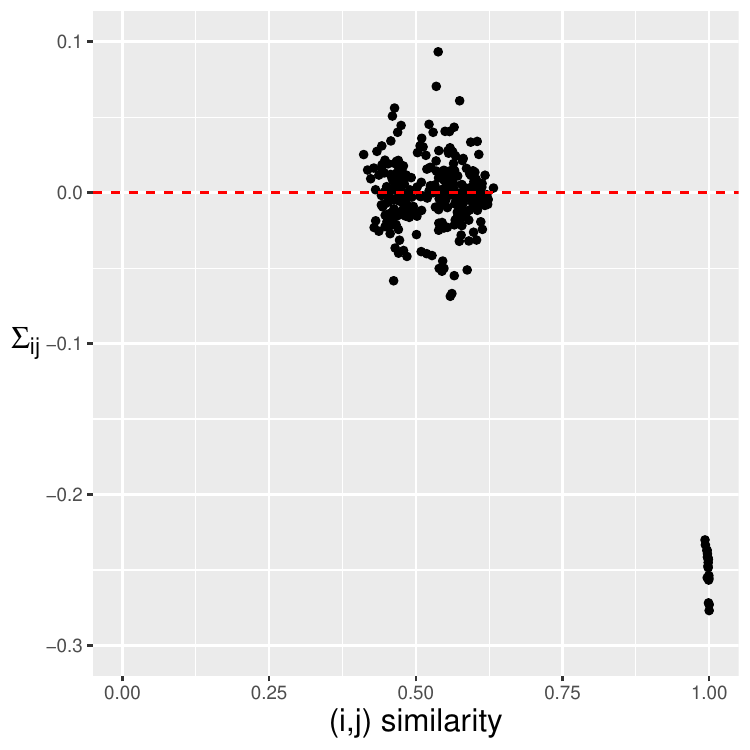}}
	\subfloat[Scatter plot for convexity]{\includegraphics[width=.35\linewidth]{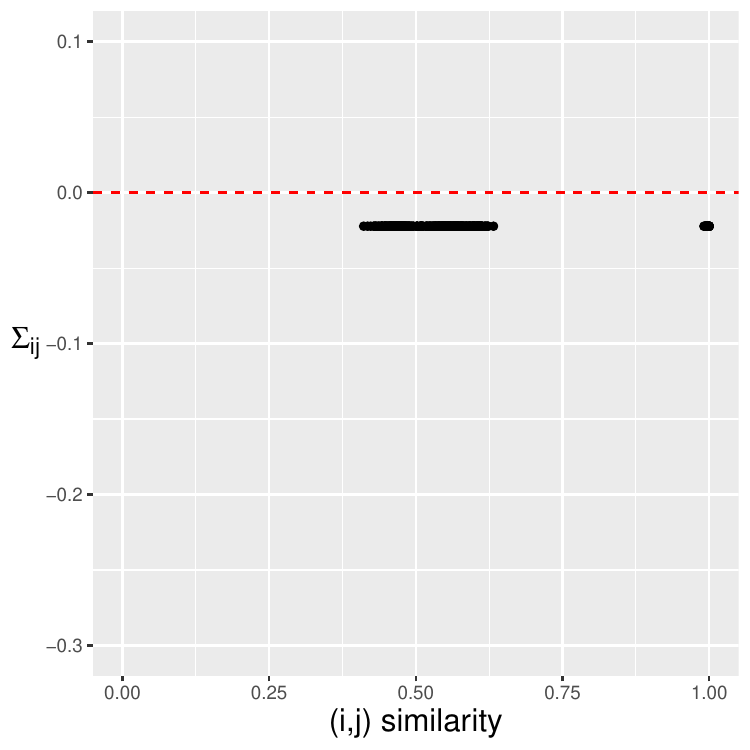}}
    \caption{Function $f$ and the correlation structure across different designs. The red dashed line indicates zero correlation, which corresponds to the i.i.d. Gaussian design. }
    \label{fig:corr}
\end{figure}

\section{Main Proofs}
Here we provide the core proofs related to Mehler's formula, asymptotic normality, and variance estimation. We also discuss inferential procedures under the continuous setting by the end of this section.
\subsection{Mehler's Formula and Related Proofs}\label{sec:proof_mehler}
Here we prove results that are based on Mehler's formula, namely, Lemma \ref{lem:mehler} and Proposition \ref{prop:f_formula_general}. Proposition \ref{prop:f_formula} is a direct result of Lemma \ref{lem:mehler} and its proof is omitted. 
\begin{proof}[Proofs of Lemma \ref{lem:mehler}]
From Mehler's formula, for any $|\rho|\le 1$, it holds that
\begin{equation*}
    p_\rho(x, y)=\sum_{m=0}^{\infty} \rho^m h_m(x) h_m(y) \phi(x) \phi(y)\;.
\end{equation*}
Therefore,
\begin{align*}
\E_{X, Y}g(X)h(Y) &= \int g(x)h(y) p_{\rho}(x,y)\dd x \dd y\\
&= \sum_{m=0}^\infty \rho^m \int g(x)h(y)h_m(x) h_m(y)\phi(x) \phi(y)\dd x \dd y\\
&= \sum_{m=0}^\infty \rho^m \int g(x)h_m(x)\phi(x) \dd x \int h(y) h_m(y) \phi(y)\dd y\\
&= \sum_{m=0}^\infty \alpha_m[g] \alpha_m[h]\rho^m\;.
\end{align*}
At the same time, by $h_0(x) = 1$ we notice
\begin{equation*}
\E g(X) = \E g(X) h_0(X) = \alpha_0[g]\;,\quad \E h(X) = \E h(X) h_0(X) = \alpha_0[h]\;.
\end{equation*}
We have
\begin{equation*}
\Cov_{X, Y}(g(X), h(Y)) = \sum_{m=0}^\infty \alpha_m[g] \alpha_m[h]\rho^m - \alpha_0[g]\alpha_0[h] = \sum_{m=1}^\infty \alpha_m[g] \alpha_m[h]\rho^m\;.
\end{equation*}
\end{proof}

\begin{proof}[Proof of Proposition \ref{prop:f_formula_general}]
By definition, for any $i\neq j$, the $(i,j)$-th entry of $\Cov_k(D)$ is 
\begin{equation*}
    \Cov(\mathbb{I}\{D_i = k\}, \mathbb{I}\{D_j = k\})\;.
\end{equation*}
Without loss of generality, we focus on $k = 2, \dots, K-1$. The extreme cases $k = 1, K$ can be proved using a similar argument. Under the Gaussianization $D_i = g(T_i)$, we have
\begin{align*}
\Cov(\mathbb{I}\{D_i = k\}, \mathbb{I}\{D_j = k\}) &= \Cov(\mathbb{I}\{T_i \in (q_{k-1}, q_k]\}, \mathbb{I}\{T_j \in (q_{k-1}, q_k]\})\\
&= \Cov(\mathbb{I}\{T_i \le q_k\} - \mathbb{I}\{T_i\le q_{k-1}\}, \mathbb{I}\{T_j \le q_k\} - \mathbb{I}\{T_j\le q_{k-1}\})\\
&= \Cov(\mathbb{I}\{T_i \le q_k\}, \mathbb{I}\{T_j \le q_k\}) + \Cov(\mathbb{I}\{T_i\le q_{k-1}\}, \mathbb{I}\{T_j\le q_{k-1}\}) \\
&-2 \Cov(\mathbb{I}\{T_i\le q_{k-1}\}, \mathbb{I}\{T_j \le q_k\})\\
&= r_{k,k}(\Sigma_{ij}) + r_{k-1. k-1}(\Sigma_{ij}) - 2 r_{k, k-1}(\Sigma_{ij})\;,
\end{align*}
where the last line follows by the definition of $r_{k,l}$ in Proposition \ref{prop:f_formula_general}. 

Then it suffices to prove \eqref{eq:corr_func}, i.e., 
\begin{equation*}
\Cov(\mathbb{I}\{X\le q_i\}, \mathbb{I}\{Y\le q_j\}) = \int_0^{\rho} \frac{1}{2\pi \sqrt{1-r^2}} \exp( -\frac{q_i^2 + q_j^2 -2r q_i q_j}{2(1-r^2)} ) \dd r\;.
\end{equation*}
Let $g(x) = \mathbb{I}\{x \le q_i\}$ and $h(x) = \mathbb{I}\{x\le q_j\}$. According to Lemma \ref{lem:mehler}, for any $|\rho| \le 1$, it holds that
\begin{equation*}
r_{ij}(\rho) = \sum_{m=1}^{\infty} \alpha_m[g] \alpha_m[h] \rho^m\;.
\end{equation*}
For $\alpha_m[g]$, we derive that
\begin{align*}
    \alpha_m[g] &= \int_{-\infty}^{q_i} h_m(x) \phi(x) \dd x \\
    &= \frac{1}{\sqrt{m!}} \int_{-\infty}^{q_i} \hermite_m(x) \phi(x) \dd x\\
    &\stackrel{\text{(i)}}{=}\frac{1}{\sqrt{m!}} \int_{-\infty}^{q_i} (-1)^m \frac{\dd^m}{\dd x^m}\phi(x) \dd x\\
    & = \frac{-1}{\sqrt{m!}} \phi(x)\hermite_{m-1}(x)\bigg|_{-\infty}^{q_i}\\
    &\stackrel{\text{(ii)}}{=} -\frac{1}{\sqrt{m!}} \phi(q_i)\hermite_{m-1}(q_i)\;.
\end{align*}
In the derivation above, (i) follows from the Definition \ref{def:hermite}, and (ii) follows from $\lim_{x\to-\infty}\phi(x) \hermite_m(x) = 0$. Hence, we have
\begin{align*}
    \alpha_m[g] &= -\frac{1}{\sqrt{m!}} \phi(q_i)\hermite_{m-1}(q_i)\;,\\
    \alpha_m[h] &= -\frac{1}{\sqrt{m!}} \phi(q_j)\hermite_{m-1}(q_j)\;,
\end{align*}
where the proof of $\alpha_m[h]$ is identical. Based on Lemma \ref{lem:mehler}, this implies 
\begin{align*}
    r_{ij}(\rho) = \sum_{m=1}^\infty \frac{1}{m!} \hermite_{m-1}(q_i) \hermite_{m-1}(q_i) \phi(q_i) \phi(q_j) \rho^m\;.
\end{align*}
Notice that 
\begin{align*}
    r_{ij}^\prime(\rho) &= \sum_{m=1}^\infty \frac{1}{(m-1)!} \hermite_{m-1}(q_i) \hermite_{m-1}(q_i) \phi(q_i) \phi(q_j) \rho^{m-1} = p_\rho(q_i, q_j)\;,\\
    r_{ij}(0) &= 0\;,
\end{align*}
where the first line follows from Mehler's formula. Then, by Newton–Leibniz theorem we obtain
\begin{equation*}
    r_{ij}(\rho) = \int_0^{\rho} p_r(q_i, q_j) \dd r =  \int_0^{\rho} \frac{1}{2\pi \sqrt{1-r^2}} \exp( -\frac{q_i^2 + q_j^2 -2r q_i q_j}{2(1-r^2)} ) \dd r\;.
\end{equation*}
\end{proof}

\subsection{Asymptotics and Inference}\label{sec:proof_normality}
In this section, we study the asymptotic properties of 
\begin{equation*}
\widehat{\tau}_k = \frac{K}{n} \sum_{i=1}^n Y_i \mathbb{I}\{D_i = k\}\;,\quad D_i = g(T_i)\;.
\end{equation*}
We prove its asymptotic normality in Theorem~\ref{thm:normality} and discuss the extensions. We defer the proof of Proposition \ref{prop:asymp_var} to Section \ref{sec:proof_supp}, as its proof follows a similar idea as those in supporting lemmas.

Our asymptotic analysis relies on H\'{a}jek's lemma \citep{lehmann1975nonparametrics}, which establishes the asymptotic equivalence between two sequences of random variables. We state below for completeness.
\begin{lemma}[H\'{a}jek's Lemma]\label{lem:hajek}
	If $(T_n - \E T_n)/\sqrt{\Var(T_n)}$ has a limit distribution $\cL$ and if 
	\begin{equation}\label{eq:hajek_condition}
		\frac{\E (T_n - S_n)^2}{\Var(T_n)} \to 0\;,
	\end{equation}
then $\Var(T_n)/\Var(S_n)\to 1$ and $(S_n - \E S_n)/\sqrt{\Var(S_n)}$ has the limit distribution $\cL$.
\end{lemma}
In words, $S_n$ and $T_n$ share the same asymptotic distribution if the second moment of their difference is asymptotically smaller than $\Var(T_n)$.

\subsubsection{Proof of Theorem \ref{thm:normality}}\label{sec:proof_thm1}
The crux of the proof is to establish an asymptotic equivalence between $\widehat{\tau}_k$ (hereafter denoted as $\widehat{\tau}^{opt}$) under $T\sim\cN(0,\Sigma_{\eta})$ and an ancillary estimator $\widehat{\tau}^{iid}$ under $T\sim\cN(0,I_n)$. To this end, we proceed with the following steps. 
\begin{enumerate}
    \item Construct a H\'{a}jek coupling $(\widehat{\tau}^{iid}, \widehat{\tau}^{opt})$.
    \item Establish the aforementioned asymptotic equivalence of $(\widehat{\tau}^{iid}, \widehat{\tau}^{opt})$ using H\'{a}jek's lemma.
    \item Prove the asymptotic normality for $\widehat{\tau}^{iid}$. 
\end{enumerate}

\noindent\underline{Step 1. Construct H\'{a}jek's coupling.}
To construct $(\widehat{\tau}^{iid}, \widehat{\tau}^{opt})$, we first define
\begin{align*}
T^{iid} \sim \cN (0, I_n)\;, \quad T^{opt }= \Sigma_{\eta}^{1/2} T^{iid}\;.
\end{align*}
One can easily check that $T^{opt}\sim\cN(0, \Sigma_{\eta})$ and $\Cov(T^{iid}, T^{opt}) = \Sigma_{\eta}^{1/2}$. Then, define
\begin{align*}
& \widehat{\tau}^{iid}=\frac{K}{n}\sum_{i=1}^n \mathbb{I}\{g(T_i^{iid})=k\} \tilde{Y}_i(k)+\frac{1}{n} \sum_{i=1}^n \left(Y_i(k)-\widetilde{Y}_i(k)\right)\;, \\
& \widehat{\tau}^{{opt}}=\frac{K}{n} \sum_{i=1}^n \mathbb{I}\{g(T_i^{opt})=k\} {Y}_i(k)\;.
\end{align*}
$\widehat{\tau}^{iid}$ matches the distribution of $\widehat{\tau}^{opt}$, since 
\begin{equation*}
    \E \widehat{\tau}^{iid} = \E \widehat{\tau}^{opt} = \tau_k\;.
\end{equation*}
More importantly, their variances also match, since
\begin{align*}
    \Var(\widehat{\tau}^{iid}) &= \Var\left(\frac{1}{n}\sum_{i=1}^n K \mathbb{I}\{g(T_i^{iid})=k\} \tilde{Y}_i(k)\right) \\
    &\stackrel{\text{(i)}}{=} \frac{K^2}{n^2} \tilde{Y}(k)^\top \Cov(D_k^{iid}) \tilde{Y}(k)\\
    &\stackrel{\text{(ii)}}{=} \frac{K^2}{n^2} \tilde{Y}(k)^\top f_k(I_n) \tilde{Y}(k)\\
    &\stackrel{\text{(iii)}}{=} \frac{K^2}{n^2} Y(k)^\top f_k(\Sigma_\eta)^{1/2} f_k(I_n)^{-1/2} f_k(I_n) f_k(I_n)^{-1/2} f_k(\Sigma_\eta) Y(k)\\
    &= \frac{K^2}{n^2} Y(k)^\top f_k(\Sigma_\eta) Y(k) = \Var(\widehat{\tau}^{opt})\;.
\end{align*}
In (i), $D_k^{iid}$ denotes the treatment vector $(\mathbb{I}\{g(T_1^{iid})=k\}, \dots, \mathbb{I}\{g(T_n^{iid})=k\})$; (ii) follows from Mehler's formula and Proposition \ref{prop:f_formula_general}; (iii) follows from the definition of $\tilde{Y}(k)$.

\noindent\underline{Step 2. Establish asymptotic equivalence.}
Based on H\'{a}jek's Lemma (Lemma \ref{lem:hajek}), we need to verify \eqref{eq:hajek_condition} for $(\widehat{\tau}^{iid}, \widehat{\tau}^{opt})$, that is
\begin{equation*}
    \frac{\E (\widehat{\tau}^{iid} - \widehat{\tau}^{opt})^2}{\Var(\widehat{\tau}^{iid})} \to 0\;.
\end{equation*}
Observe that
\begin{equation*}
    \Var(\widehat{\tau}^{iid}) = \frac{K-1}{n^2} \sum_{i=1}^n \tilde{Y}_i(k)^2\;.
\end{equation*}
Under the assumptions in Theorem \ref{thm:normality}, $n\Var(\widehat{\tau}^{iid})$ converges to a positive limit, and thus $\Var(\widehat{\tau}^{iid})\asymp 1/n$. Then, it suffices to verify 
\begin{equation}\label{eq:hajek_n}
    n \E (\widehat{\tau}^{iid} - \widehat{\tau}^{opt})^2 \to 0\;.
\end{equation}
Notice that 
\begin{align*}
\E (\widehat{\tau}^{iid} - \widehat{\tau}^{opt})^2 &\stackrel{\text{(i)}}{=} \Var (\widehat{\tau}^{iid} - \widehat{\tau}^{opt})\\
&= \frac{K^2}{n^2} \Var \left(\sum_{i=1}^n (\mathbb{I}\{g(T_i^{iid}) = k\} \tilde{Y}_i(k) - \mathbb{I}\{g(T_i^{opt}) = k\} Y_i(k))\right)\\
&= \frac{K^2}{n^2} \left( \tilde{Y}(k)^\top \Cov(D_k^{iid}) \tilde{Y}(k) - 2 \tilde{Y}(k)^\top \Cov(D_k^{iid}, D_k^{opt}) Y(k) + Y(k)^\top \Cov(D_k^{opt}) Y(k) \right)\\
&\stackrel{\text{(ii)}}{=} \frac{K^2}{n^2} \Bigl( \tilde{Y}(k)^\top f_k(I_n) \tilde{Y}(k) - 2\tilde{Y}(k)^\top f_k(\Sigma_{\eta}^{1/2}) Y(k) + Y(k)^\top f_k(\Sigma_{\eta}) Y(k) \Bigr)\;,
\end{align*}
where (i) follows from $\E\widehat{\tau}^{iid}=\E \widehat{\tau}^{opt}$, and (ii) follows from Proposition \ref{prop:f_formula_general}. 

By Proposition \ref{prop:f_formula_general}, it is easy to verify that $f_k(0) = 0$ and $f_k(1) = (K-1)/K^2$. From now on, without loss of generality, we may rescale $f_k$ such that $f_k(0) = 0$ and $f_k(1) = 1$. This does not affect the order of the quantity above, since $K$ is a fixed constant. After rescaling, we have $f_k(I_n) = I_n$, which simplifies the derivation below. By definition of $\tilde{Y}(k)$, we have
\begin{align}
\E (\widehat{\tau}^{iid} - \widehat{\tau}^{opt})^2 
&= \frac{K^2}{n^2} {Y}(k)^\top \left( 2 f_k(\Sigma_{\eta}) - 2 f_k(\Sigma_\eta)^{1/2} f_k(\Sigma_{\eta}^{1/2}) \right) {Y}(k)^\top\;,\nonumber\\
&\le \frac{2 K^2}{n^2} \|f_k(\Sigma_{\eta}) - f_k(\Sigma_\eta)^{1/2} f_k(\Sigma_{\eta}^{1/2})\|_{\op} \|Y(k)\|^2\nonumber\\
 &\stackrel{\text{(i)}}{\le} \frac{2 M K^2}{n} \|f_k(\Sigma_{\eta})\|_{\op}^{1/2} \|  f_k(\Sigma_\eta)^{1/2} - f_k(\Sigma_{\eta}^{1/2})\|_{\op} \;,\label{eq:norm_bound_fk}
\end{align}
where (i) follows from $\|Y(k)\|^2\le n M$. 

To analyze the operator norm above, we first give a decomposition of $\Sigma_{\eta}$. 
\begin{lemma}\label{lem:N_op}
Suppose Assumption \ref{asmp:stepsize} holds. In the one-step PGD-Gauss, the obtained solution $\Sigma_{\eta}$ satisfies a decomposition
\begin{equation*}
    \Sigma_{\eta} = I_n + \eta N\;,
\end{equation*}
where $N$ is a symmetric matrix with zero diagonal values, and 
$$
\|N\|_{\op} = O(\|XX^\top - I_n\|_{\op} + \eta \|XX^\top - I_n\|_{\op}^2)\;.
$$
\end{lemma}

Next we introduce the following result based on Taylor expansions.
\begin{lemma}\label{lem:matrix_expansion}
For $\Sigma\in\cE$, define $\Delta = \Sigma - I_n$, the residual matrix with zero diagonal values. Suppose $\|\Delta\|_{\op} = o(1)$.
We have
\begin{align}
f_k(\Sigma) &= I_n + f_k^\prime(0) \Delta + R_1\;,\quad \|R_1\|_{\op} = O(\|\Delta\|_{\op}^2)\;,\quad \|f_k(\Sigma)\|_{\op} = O(1)\nonumber\\
f_k(\Sigma)^{1/2} &= I_n + \frac{1}{2}f_k^\prime(0) \Delta + R_2\;,\quad \|R_2\|_{\op} = O(\|\Delta\|_{\op}^2)\;,\nonumber\\
f_k(\Sigma^{1/2}) &= I_n + \frac{1}{2}f_k^\prime(0) \Delta + R_3\;, \quad \|R_3\|_{\op} = o(1) + O(\|\Delta\|_{\op}^2)\;.\nonumber
\end{align}
Moreover, the operator norm of $R_1, R_2, R_3$ are all of order $o(1)$ since $\|\Delta\|_{\op} = o(1)$.
\end{lemma}

Now, we utilize Lemma \ref{lem:N_op} and Lemma \ref{lem:matrix_expansion} to verify the H\'{a}jek condition \eqref{eq:hajek_n}. Based on Lemma \ref{lem:N_op}, we have
\begin{align*}
\Sigma_{\eta} &= I_n + \Delta\;,\quad \Delta = \eta N\;, \\
\|N\|_{\op} &= O(\|XX^\top-I_n\|_{\op} + \eta \|XX^\top-I_n\|_{\op}^2)\;,\\
\Rightarrow \|\Delta\|_{\op} &= O(\eta \|XX^\top-I_n\|_{\op} + \eta^2 \|XX^\top - I_n\|_{\op}^2)\;,
\end{align*}
where $N$ is a symmetric matrix with zero diagonal values. Under Assumption \ref{asmp:stepsize} that $\eta \|XX^\top - I_n\|_{\op} = o(1)$, one can verify that
\begin{equation*}
\|\Delta\|_{\op} = o(1)\;.
\end{equation*}
Thus the condition required in Lemma \ref{lem:matrix_expansion} is satisfied. We apply Lemma \ref{lem:matrix_expansion} to obtain
\begin{equation}\label{eq:matrixnorm}
\begin{aligned}
\|f_k(\Sigma_\eta)\|_{\op} &= O(1)\;,\\
f_k(\Sigma_{\eta})^{1/2} &= I_n + \frac{1}{2}f_k^\prime(0) \eta N + R_2\;,\quad \|R_2\|_{\op} = o(1)\;,\\
f_k(\Sigma_{\eta}^{1/2}) &= I_n + \frac{1}{2}f_k^\prime(0) \eta N + R_3\;, \quad \|R_3\|_{\op} = o(1)\;.
\end{aligned} 
\end{equation}


By applying Equations \eqref{eq:matrixnorm} to \eqref{eq:norm_bound_fk}, we obtain
\begin{equation*}
    \E (\widehat{\tau}^{iid} - \widehat{\tau}^{opt})^2 = O\left(\frac{1}{n} \|f_k(\Sigma_{\eta})\|_{\op}^{1/2} \|  f_k(\Sigma_\eta)^{1/2} - f_k(\Sigma_{\eta}^{1/2})\|_{\op}\right) = o\left(\frac{1}{n}\right) \;.
\end{equation*}
Therefore, condition \eqref{eq:hajek_n} holds and one can apply H\'{a}jek's Lemma (Lemma \ref{lem:hajek}) to obtain that
\begin{equation*}
    \frac{\widehat{\tau}^{opt} - \tau_k}{\sqrt{\Var(\widehat{\tau}^{opt})}}
\end{equation*}
has the same asymptotic distribution as 
\begin{equation*}
    \frac{\widehat{\tau}^{iid} - \tau_k}{\sqrt{\Var(\widehat{\tau}^{iid})}}\;.
\end{equation*}

\noindent\underline{Step 3. Asymptotic normality for $\widehat{\tau}^{iid}$.}
We define
\begin{equation*}
    X_{ni} = (K\mathbb{I}\{g(T_i^{iid}) = k\} - 1) \tilde{Y}_i(k)\;,\quad S_n = \sum_{i=1}^n X_{ni}\;.
\end{equation*}
It is then easy to verify that 
\begin{align*}
    \widehat{\tau}^{iid} - \tau_k &= \frac{1}{n} S_n\;,\quad 
    \Var(\widehat{\tau}^{iid}) = \Var(\frac{1}{n}S_n)\;.
\end{align*}
Therefore, it suffices to derive the asymptotic distribution for $S_n$. The Lindeberg condition requires that for any $\eps>0$, 
\begin{equation*}
    \frac{1}{\Var(S_n)} \sum_{i=1}^n \E X_{ni}^2 \mathbb{I}\{X_{ni}^2\ge \eps \Var(S_n)\} \to 0\;.
\end{equation*}
Note that
\begin{equation*}
    \Var(S_n) = \sum_{i=1}^n \Var(X_{ni}) = (K-1) \sum_{i=1}^n \tilde{Y}_i^2(k)\;.
\end{equation*}
Under Condition 2 in Theorem \ref{thm:normality}, $\Var(S_n)$ is of order $n$. Hence, 
\begin{align*}
    \frac{\max_i X_{ni}^2}{\Var(S_n)} &\le \frac{\max_i (K-1)\tilde{Y}_{i}^2(k)}{\Var(S_n)}
    \asymp \frac{\max_i \tilde{Y}_{i}^2(k)}{n} = o(1)\;.
\end{align*}
The last equality follows from Condition 1. This suggests that all the summands in the Lindeberg condition become zero for large $n$. Therefore, the Lindeberg condition is satisfied, and we have
\begin{equation*}
    \frac{S_n}{\sqrt{\Var(S_n)}} \stackrel{d}{\to} \cN(0, 1)\;.
\end{equation*}
Note that $\Var(S_n) = (K-1)\sum Y_i^2(k)$. We have
\begin{equation*}
    \sqrt{n}(\widehat{\tau}^{iid} - \tau_k) \stackrel{d}{\to} \cN(0, \lim_{n}\frac{K-1}{n}\sum_i \tilde{Y}_i^2(k))\;.
\end{equation*}
This completes the proof of Theorem \ref{thm:normality}.

\subsubsection{Generalization to Multi-Step PGD-Gauss Solutions}
Here, we discuss a generalization of Theorem \ref{thm:normality} to multi-step PGD-Gauss solutions. Specifically, the proof of Theorem \ref{thm:normality} indicates a more general result below.
\begin{corollary}\label{cor:normality}
Consider a Gaussianization $T\sim \cN(0, \Sigma)$. Suppose that $\Delta\coloneqq \Sigma - I_n$ satisfies $\|\Delta\|_{\op} = o(1)$, and that Conditions 1-3 in Theorem \ref{thm:normality} hold. Then, we have
\begin{equation*}
\sqrt{n}\left(\widehat{\tau}_k-\tau_k\right) \stackrel{d}{\to} \cN\left(0, \lim_{n\to\infty}\frac{K-1}{n} \|\tilde{Y}(k)\|^2\right)\;.
\end{equation*}
\end{corollary}
Corollary \ref{cor:normality} can be viewed as a result for general PGD-Gauss solutions. Regardless of how many steps taken in the PGD-Gauss, the asymptotic normality for $\widehat{\tau}_k$ holds as long as the solution $\Sigma$ does not deviate too much from the identity matrix.

\begin{proof}
Following Step 1 in the proof of Theorem \ref{thm:normality}, we construct the coupling in the same way. Then, based on the analysis in Step 2, it suffices to show that 
\begin{equation*}
    \|f_k(\Sigma)\|_{\op}^{1/2} \|f_k(\Sigma)^{1/2} - f_k(\Sigma^{1/2})\|_{\op} = o(1) \;, \quad \Sigma = I_n + \Delta\;.
\end{equation*}
Under our assumption in Corollary \ref{cor:normality}, we have $\|\Delta\|_{op} = o(1)$ and thus the condition in Lemma \ref{lem:matrix_expansion} is satisfied. We can then apply Lemma \ref{lem:matrix_expansion} to show the above equation. Lastly, we can apply Step 3 in Theorem \ref{thm:normality} and complete the proof. 
\end{proof}

\subsubsection{Inference}\label{sec:proof_inference}
We prove Theorem \ref{thm:var_HT} below. 
\begin{proof}
By definition of $\widehat{V}_\eta$ \eqref{eq:V_HT}, we write
\begin{align*}
    \widehat{V}_\eta = \underbrace{\frac{K^2}{n} \sum_{i=1}^n Y_i^2 \frac{\mathbb{I}\{g(T_i) = k\}}{\P(g(T_i)=k)}}_{A} + \underbrace{\frac{K^2}{n} \sum_{i\neq j} Y_i Y_j f_k(\Sigma_{\eta, ij}) \frac{\mathbb{I}\{g(T_i) = k, g(T_j) = k\}}{\P(g(T_i)=k, g(T_j)=k)}}_{B}\;.
\end{align*}
That is, we decompose the variance estimator into the diagonal part (A) and the off-diagonal part (B). 

First, we show that the probabilities in the denominator of (A), (B) are positive, and thus $\widehat{V}_{\eta}$ is well-defined. Under the uniform design, we have $\P(g(T_i)=k)=1/K > 0$ and hence the denominators in (A) are positive. For (B), we first apply Proposition \ref{prop:f_formula_general} to obtain 
\begin{align*}
    \P(g(T_i) &= k, g(T_j) = k) = f_k(\Sigma_{\eta, ij}) + 1/K^2\\
    \Rightarrow~~ B &= \frac{K^2}{n} \sum_{i\neq j} Y_i Y_j f_k(\Sigma_{\eta,ij}) \frac{\mathbb{I}\{g(T_i) = k, g(T_j) = k\}}{f_k(\Sigma_{\eta, ij}) + 1/K^2}\;.
\end{align*}
Under Assumption \ref{asmp:stepsize}, we apply Lemma \ref{lem:N_op} to obtain
\begin{equation*}
    \Sigma_{\eta} = I_n  + \eta N\;, \quad \|N\|_{\op} = O(\|XX^\top - I_n\|_{\op})\;. 
\end{equation*}
Since $\Sigma_{\eta}$ is symmetric, we apply Lemma \ref{lem:op} to obtain
\begin{equation}\label{eq:element_bound}
    \max_{i\neq j} |\Sigma_{\eta, ij}| = O(\eta \|N\|_{\op}) = O(\eta \|XX^\top - I_n\|_{\op}) = o(1)\;.
\end{equation}
Therefore, $\Sigma_{\eta, ij}$ is $o(1)$. For any $i\neq j$, since $f_k(0) = 0$ and $f_k$ is smooth around zero, we have $f_k(\Sigma_{\eta, ij}) = o(1)$, and hence $f_k(\Sigma_{\eta, ij}) + 1/K^2 > 0$ for large enough $n$. Therefore $\widehat{V}_{\eta}$ is well-defined. 

To analyze the mean of $\widehat{V}_{\eta}$, we use the original definition \eqref{eq:V_HT} to obtain
\begin{align*}
    \E \widehat{V}_\eta &= \frac{K^2}{n} \E \sum_{i,j} Y_i Y_j f_k(\Sigma_{\eta, ij}) \frac{\mathbb{I}\{g(T_i) = k, g(T_j) = k\}}{\P(g(T_i) = k, g(T_j) = k)}\\
    &= \frac{K^2}{n} \E \sum_{i,j} Y_i(k) Y_j(k) f_k(\Sigma_{\eta, ij}) \frac{\mathbb{I}\{g(T_i) = k, g(T_j) = k\}}{\P(g(T_i) = k, g(T_j) = k)}\\
    &= \frac{K^2}{n} \sum_{i,j} Y_i(k) Y_j(k) f_k(\Sigma_{\eta, ij}) \frac{\P(g(T_i) = k, g(T_j) = k)}{\P(g(T_i) = k, g(T_j) = k)} = V(\Sigma_{\eta})\;. 
\end{align*}

To analyze the variance, by the AM-GM inequality, we have
\begin{align*}
    \Var(\widehat{V}_\eta) \le 2 (\Var({A}) + \Var(B))\;.
\end{align*}
Next, we will show that $\Var(A)$ and $\Var(B)$ are $o(1)$, respectively. 

For (A), since $\P(g(T_i)=k) = 1/K$, we have
\begin{equation*}
    A = \frac{K^3}{n} \sum_{i=1}^n Y_i^2 \mathbb{I}\{g(T_i) = k\}\;.
\end{equation*}
By the variance formula of Horvitz-Thompson estimator (Lemma \ref{lem:discrete_var}), we have
\begin{equation*}
    \Var(A) = \frac{K^6}{n^2} {Y(k)^2}^\top \Cov_k(D) Y(k)^2\;, \quad \Cov_k(D) = f_k(\Sigma_\eta)\;.
\end{equation*}
where $Y(k)^2$ is the vector of squared potential outcomes. Based on Lemma \ref{lem:N_op}, Lemma \ref{lem:matrix_expansion}, and the analysis in Step 2 of Section \ref{sec:proof_thm1}, we have $\|f_k(\Sigma_{\eta})\|_{\op} = O(1)$. Therefore, 
\begin{equation*}
    \Var(A) \le \frac{K^6}{n^2}  n \max_i |Y_i(k)|^2 \|f_k(\Sigma_\eta)\|_{\op} = O\left(\frac{1}{n}\right)\;.
\end{equation*}
The last equality follows from the assumption that $\max_i |Y_i(k)| = O(1)$. 

For (B), by definition we have
\begin{align*}
    B^2 &= \frac{K^4}{n^2} \sum_{i\neq j} \sum_{p\neq q} Y_i Y_j Y_q Y_l f_k(\Sigma_{ij}) f_k(\Sigma_{pq}) \frac{\mathbb{I}\{g(T_i) = k, g(T_j) = k, g(T_p) = k, g(T_q) = k\}}{(f_k(\Sigma_{\eta,ij})+1/K^2)(f_k(\Sigma_{\eta,pq})+1/K^2)}\\
    &= \frac{K^4}{n^2} \sum_{i\neq j} \sum_{p\neq q} Y_i(k) Y_j(k) Y_q(k) Y_l(k) f_k(\Sigma_{ij}) f_k(\Sigma_{pq}) \frac{\mathbb{I}\{g(T_i) = k, g(T_j) = k, g(T_p) = k, g(T_q) = k\}}{(f_k(\Sigma_{\eta,ij})+1/K^2)(f_k(\Sigma_{\eta,pq})+1/K^2)}\;.
\end{align*}
Therefore, the second moment of $B$ satisfies
\begin{align*}
    \E B^2 &= \frac{1}{n^2} \sum_{i\neq j} \sum_{p\neq q} Y_i(k) Y_j(k) Y_q(k) Y_l(k) f_k(\Sigma_{ij}) f_k(\Sigma_{pq}) \frac{\P(g(T_i) = k, g(T_j) = k, g(T_p) = k, g(T_q) = k)}{(f_k(\Sigma_{\eta,ij})+1/K^2)(f_k(\Sigma_{\eta,pq})+1/K^2)}\;.
\end{align*}
Based on \eqref{eq:element_bound}, $f_k(0) = 0$, and the smoothness of $f_k$ around zero, for any $i\neq j$, we have 
\begin{equation*}
    |f_k(\Sigma_{\eta,ij})| = O(\eta \|XX^\top - I_n\|_{\op}) = o(1)\;.
\end{equation*}
Then, we use the fact that $\P(g(T_i) = k, g(T_j) = k, g(T_p) = k, g(T_q) = k)\le 1$ and $f_k(\Sigma_{\eta,ij})+1/K^2 = o(1) + 1/K^2$ to obtain 
\begin{align*}
 \E B^2 &=  O\left(\frac{K^4}{n^2} \sum_{i\neq j} \sum_{p\neq q} \max_i |Y_i(k)|^4 \eta^2 \|XX^\top - I_n\|_{\op}^2\right) = n^2 \eta^2 \|XX^\top - I_n\|_{\op}^2\;,
\end{align*}
where the last equality follows from that fact that $K$ is constant and $\max |Y_i(k)| = O(1)$. Under the assumption that $n^2 \eta^2 \|XX^\top - I_n\|_{\op}^2 = o(1)$, we have $\E B^2 = o(1)$ and hence $\Var(B) = o(1)$. 

To sum up, we have
\begin{equation*}
    \Var(\widehat{V}_\eta) \le 2 (\Var(A) + \Var(B)) = o(1)\;.
\end{equation*}
This completes the proof. 
\end{proof}

\section{Proofs of Supporting Lemmas}\label{sec:proof_supp}
In the proof of supporting lemmas, we utilize matrix norm inequalities to analyze the perturbation of a Gaussian covariance $\Sigma$ with respect to $I_n$. Specifically, for any matrix $A$, we have
\begin{equation}\label{eq:1-inf-norm}
    \|A\|_{\op} \le \sqrt{\|A\|_1 \|A\|_{\infty}}\;.
\end{equation}
where $\|A\|_{1}$ and $\|A\|_{\infty}$ denote the matrix 1-norm and infinity-norm, i.e., 
\begin{equation*}
    \|A\|_1 = \max_{j=1,\dots, n} \sum_{i=1}^n |A_{ij}|\;, \quad \|A\|_\infty = \max_{i=1,\dots, n} \sum_{j=1}^n |A_{ij}|\;.
\end{equation*}
Moreover, when $A$ is a symmetric matrix, we have $\|A\|_1 = \|A\|_{\infty}$, and 
\begin{equation}\label{eq:inf-norm}
    \|A\|_{\op} \le \|A\|_{\infty}\;.
\end{equation}
These inequalities will be invoked multiple times in our proof. Additionally, for $A\in\R^{n\times n}$, let $\mathrm{diag}(A)\in\R^{n\times n}$ be the diagonal matrix with the same diagonal values as $A$.

Our proof leverages the specific form of the one-step PGD-Gauss under the nuclear norm. That is, based on Algorithm \ref{alg:pgd}, we have
\begin{equation}\label{eq:pgd_sigma}
   \begin{aligned}
   \Sigma_{\eta} &= V_{\eta}  V_{\eta}^\top\;,\quad V_{\eta} = D^{-1} U_{\eta}\;,\\
    U_{\eta} &= (I_n - \eta \nabla l_{\mathrm{norm}}(I_n))V_0 = I_n - \eta \nabla l_{\mathrm{norm}}(I_n)\;,\\
    \nabla l_{\norm}(I_n) &= f_k^\prime(0) (XX^\top - I_n)\;,
\end{aligned} 
\end{equation}
where $D$ is a diagonal matrix with $i$-th diagonal equal to the norm of $u_i$, the $i$-th row of $U_{\eta}$.

Lastly, we prove the following matrix inequality. 
\begin{lemma}\label{lem:op}
For any symmetric matrix $A$, we have $\max_{ij} |A_{ij}| \le 2 \|A\|_{\op}$.
\end{lemma}
\begin{proof}
By the variational expression of the operator norm, we have
\begin{equation*}
    \|A\|_{\op} = \sup_{\|x\|= 1} |x^\top A x|\;.
\end{equation*}
Let $x = e_i$, the basis vector with the $i$-th entry equal to one, we obtain 
\begin{equation*}
    |x^\top A x| = |A_{ii}| \le \|A\|_{\op}\;. 
\end{equation*}
Then, for any $i\neq j$, by setting $x = (e_i + e_j)/\sqrt{2}$, we obtain
\begin{equation*}
    |x^\top A x| = \frac{1}{2} |A_{ii} + A_{jj} + 2 A_{ij}| \le \|A\|_{\op}\;.
\end{equation*}
By the triangular inequality, we have $|A_{ii} + A_{jj} + 2 A_{ij}| \ge \bigl||A_{ii}+A_{jj}| - 2|A_{ij}|\bigr|$, and hence
\begin{equation*}
    \bigl||A_{ii}+A_{jj}| - 2|A_{ij}|\bigr| \le 2 \|A\|_{\op}\;.
\end{equation*}
This implies
\begin{equation*}
    2 |A_{ij}| \le 2 \|A\|_{\op} + |A_{ii} + A_{jj}| \le 4 \|A\|_{\op}\;.
\end{equation*}
Therefore, we obtain $\max_{ij}|A_{ij}|\le 2\|A\|_{\op}$.
\end{proof}

\subsection{Proof of Lemma \ref{lem:N_op}}
Based on \eqref{eq:pgd_sigma}, we have 
\begin{equation*}
    D_{ii} = \sqrt{1 + \sum_{j\neq i} \eta^2 \nabla^2 l_{\mathrm{norm}, ij}(I_n)}\;.
\end{equation*}
Notice that $\sum_{j\neq i} \nabla^2 l_{\mathrm{norm}, ij}(I_n) $ is the $i$-th diagonal of matrix $(\nabla l_{\mathrm{norm}}(I_n))^2$, we have
\begin{equation*}
    |\sum_{j\neq i} \nabla^2 l_{\mathrm{norm}, ij}(I_n)| \le \|\nabla l_{\mathrm{norm}}(I_n)\|_{\op}^2 = (f^\prime_k(0))^2 \|XX^\top  - I_n\|_{\op}^2\;.
\end{equation*}
Thus, 
\begin{equation}\label{eq:Dii}
    D_{ii} \le \sqrt{ 1 + \eta^2 \|XX^\top - I_n\|_{\op}^2 (f^\prime_k(0))^2 } = \sqrt{1 + O(\eta^2 \|XX^\top - I_n\|_{\op}^2)}\;.
\end{equation}
Under Assumption \ref{asmp:stepsize}, $D_{ii} = \sqrt{1 + o(1)}$. Therefore, $\max |D_{ii} - 1| = o(1)$. This fact will be used in our proof.

Noticing that $\Sigma_{\eta} = V_{\eta} V_{\eta}^\top$, we have 
\begin{equation}\label{eq:Sigma_eta}
\begin{aligned}
    \Sigma_{\eta} &= D^{-1} (I_n - \eta \nabla l_{\mathrm{norm}}(I_n))^2 D^{-1}\\
    &= D^{-1} (I_n - \eta f^\prime_k(0) (XX^\top - I_n))^2 D^{-1}\\
    &= D^{-1} (I_n - 2 \eta f^\prime_k(0) (XX^\top - I_n) + \eta^2 (f^\prime_k(0) (XX^\top - I_n))^2) D^{-1}\\
    &= D^{-2} + \underbrace{D^{-1} (- 2 \eta f^\prime_k(0) (XX^\top - I_n) + \eta^2 (f^\prime_k(0) (XX^\top - I_n))^2 )D^{-1}}_\text{$\eqqcolon M$}\;,
\end{aligned}
\end{equation}
where we use $M$ to denote the component that contributes to off-diagonal elements. 

We write
\begin{align*}
    \Sigma_{\eta} &= I_n + \eta N\;,\quad N \coloneqq \frac{1}{\eta} (\Sigma_{\eta} - I_n)\;.
\end{align*}
Based on the $M$ defined in Equation \eqref{eq:Sigma_eta}, we derive
\begin{equation}\label{eq:N_eq}
    N = \frac{1}{\eta} ( D^{-2} + M - I_n) = \frac{1}{\eta} (M - \mathrm{diag}(M))\;.
\end{equation}
Hence, 
\begin{equation*}
    \eta \|N\|_{\op} \le \|M\|_{\op} + \|\mathrm{diag}(M)\|_{\op}\;.
\end{equation*}
For $\|M\|_{\op}$, we have
\begin{align*}
    \|M\|_{\op} &\le \|D^{-1}\|^2_{\op} \left(2\eta \|f^\prime_k(0) (XX^\top - I_n)\|_{\op} + \eta^2 \|f^\prime_k(0) (XX^\top - I_n)\|^2_{\op}\right)\\
    &= O\left( \eta \|XX^\top - I_n\|_{\op} + \eta^2 \|XX^\top - I_n\|^2_{\op}\right)\;.
\end{align*}
The last line follows from $\|D^{-1}\|_{\op} = 1 + o(1)$ since $\max|D_{ii} - 1| = o(1)$. 
Similarly, for $\|\mathrm{diag}(M)\|_{\op}$, we have
\begin{align*}
    \|\mathrm{diag}(M)\|_{\op} &= O\left(\eta \|\mathrm{diag}(f^\prime(0) (XX^\top - I_n))\|_{\op}\right) + O\left(\eta^2 \|\mathrm{diag}((f^\prime(0) (XX^\top - I_n))^2)\|_{\op}\right)\\
    &\stackrel{\text{(i)}}{=} O\left(\eta^2 \|\mathrm{diag}((XX^\top - I_n)^2)\|_{\op}\right) \stackrel{\text{(ii)}}{=} O\left(\eta^2 \|XX^\top - I_n)\|^2_{\op}\right)\;.
\end{align*}
In (i), we use the fact that $\mathrm{diag}(XX^\top - I_n) = 0$, since $\|X_i\| = 1$ under Assumption \ref{asmp:stepsize}. (ii) follows from the fact that $\|\mathrm{diag}(A)\|_{\op}\le \|A\|_{\op}$ for a positive semidefinite matrix $A$.

Based on our analysis for $\|M\|_{\op}$ and $\|\mathrm{diag}(M)\|_{\op}$ above, we have
\begin{equation*}
    \|N\|_{\op} = O\left(\|XX^\top - I_n\|_{\op} + \eta \|XX^\top - I_n\|^2_{\op}\right)\;.
\end{equation*}

\subsection{Proof of Lemma \ref{lem:matrix_expansion}}
Here we prove the results in Lemma \ref{lem:matrix_expansion} one by one. Recall that we assume $f_k(1) = 1$, such that $f_k(I_n) = I_n$, as explained in the main proof of Section \ref{sec:proof_normality}.

\noindent\underline{Analyze $f_k(\Sigma)$.}
Note that $f_k$ is smooth around zero. For any $x\in(-1, 1)$, we use Taylor expansion to obtain
\begin{equation}\label{eq:fk_taylor}
    f_k(x) = f_k(0) + f_k^\prime(0) x + \frac{1}{2} f_k^{\prime\prime}(\xi_x) x^2 = f_k^\prime(0) x + \frac{1}{2} f_k^{\prime\prime}(\xi_x) x^2\;,
\end{equation}
where $\xi_{x}$ is a constant satisfying $|\xi_x| \le |x|$. With the matrix input $\Sigma$, we apply the above Taylor expansion to off-diagonal entries to obtain  
\begin{equation*}
    f_k(\Sigma) = I_n + f_k^\prime(0) \Delta + R_1\;, \quad R_{1,{ij}} = \frac{1}{2} f_k^{\prime\prime}(\xi_{ij}) \Delta_{ij}^2\;,
\end{equation*}
where $\xi_{ij}$ satisfies $|\xi_{ij}| \le |\Delta_{ij}|$. 
Since $R_1$ is symmetric, by \eqref{eq:inf-norm}, we have $\|R_1\|_{\op} \le \|R_1\|_{\infty}$ and
\begin{align*}
    \|R_1\|_\infty = \max_{i=1,\dots, n} \sum_{j=1}^n |R_{f, ij}| = \frac{1}{2} \max_i \sum_j |f_k^{\prime\prime}(\xi_{ij}) \Delta_{ij}^2|\;.
\end{align*}
By Lemma \ref{lem:op}, $\max_{ij}|\Delta_{ij}|= O(\|\Delta\|_{\op}) = o(1)$. Since $|\xi_{ij}| \le |\Delta_{ij}|$, we have $\max_{ij}|\xi_{ij}| \le \max_{ij}|\Delta_{ij}|= o(1)$, and hence $\max_{ij} |f_k^{\prime\prime}(\xi_{ij})| = O(1)$. We have 
\begin{equation*}
\sum_i |f_k^{\prime\prime}(\xi_{ij}) \Delta_{ij}^2| = O(\sum_i \Delta_{ij}^2)\;.
\end{equation*}
Notice that $\sum_i \Delta_{ij}^2$ is the $j$-th diagonal of the squared matrix $\Delta^2$. Therefore, we have 
\begin{equation*}
    O(\sum_i \Delta_{ij}^2) = O(\|\Delta^2\|_{\op}) = O(\|\Delta\|_{\op}^2)\;.
\end{equation*}
To sum up, we derive $\|R_1\|_\infty = O(\|\Delta\|_{\op}^2)$, and thus $\|R_1\|_{\op} = O(\|\Delta\|_{\op}^2)$. In addition, under the condition $\|\Delta\|_{\op} = o(1)$, we have
\begin{equation*}
    \|f_k(\Sigma)\|_{\op} \le \|I_n\|_{\op} + |f^\prime_k(0)| \|\Delta\|_{\op} + \|R_1\|_{\op} = O(1 + \|\Delta\|_{\op} +  \|\Delta\|_{\op}^2) = O(1)\;.
\end{equation*}

\noindent\underline{Analyze $f_k(\Sigma)^{1/2}$.} 
We first introduce the Taylor expansion of the matrix square root. For any symmetric matrix with $\|M\|_{\op}< 1$, we define its eigenvalue decomposition as $M = U\Lambda U^\top$. Then, by definition of the matrix square root, 
\begin{equation*}
(I_n + M)^{1/2} = U (I_n + \Lambda)^{1/2} U^\top\;.
\end{equation*}
That is, the $i$-th eigenvalue $(1+\Lambda_i)$ is transformed to $s(1+\Lambda_i)$, where $s(x) = \sqrt{x}$ is the square-root function. Based on Taylor expansion of $s(\cdot)$, 
\begin{equation*}
    s(1+\Lambda_i) = 1 + \frac{1}{2}\Lambda_i + \frac{1}{2} s^{\prime\prime}(\xi_{i}) \Lambda_i^2\;,
\end{equation*}
where $\xi_i$ is some value satisfying $|\xi_i - 1|\le |\Lambda_i|<1$. Therefore, we can write
\begin{equation*}
    (I_n + M)^{1/2} = U \left(I_n + \frac{1}{2}\Lambda + \frac{1}{2} S \Lambda^2 \right) U^\top = I_n + \frac{1}{2}M + \frac{1}{2} U S \Lambda^2 U^\top\;,
\end{equation*}
where $S$ is a diagonal matrix with elements $ s^{\prime\prime}(\xi_i)$. 

For $f_k(\Sigma)$, we may set
\begin{equation*}
    \Delta_f = f_k(\Sigma) - I_n = f_k^\prime(0) \Delta + R_1\;.
\end{equation*}
Moreover, $\|\Delta_f\|_{\op} \le |f_k^\prime(0)| \|\Delta\|_{\op} + \|R_1\|_{\op} = O(\|\Delta\|_{\op}) + O(\|\Delta\|_{\op}^2) = o(1)$. Therefore, we can apply the matrix square root Taylor expansion to obtain 
\begin{align*}
    f_k(\Sigma)^{1/2} &= I_n + \frac{1}{2}\Delta_f + \frac{1}{2} U S \Lambda^2 U^\top\\
    &= I_n + \frac{1}{2}f^\prime_k(0)\Delta + \frac{1}{2}R_1 + R_f\;,\\
    R_f &\coloneqq \frac{1}{2} U S \Lambda^2 U^\top\;.
\end{align*}
Here, the matrices $U$, $S$, $\Lambda$ are defined with respect to $M = \Delta_f$. Then, it remains to bound the operator norm of $R_f$. Note that $\|\Delta_f\|_{\op} = o(1)$ and $s(\cdot)$ is smooth around one. We have $\max_i \{|s^{\prime\prime}(\xi_i)|\} = O(1)$. Therefore, the norm of $R_f$ can be bounded as
\begin{equation*}
    \|R_f\|_{\op} \le O(\|\Delta_f\|_{\op}^2) = O( \|\Delta\|_{\op}^2 + \|\Delta\|_{\op}^4) = O(\|\Delta\|_{\op}^2)\;.
\end{equation*}
By setting $R_2 = R_1/2 + R_f$, we obtain
\begin{equation*}
    f(\Sigma)^{1/2} = I_n + \frac{1}{2}f^\prime_k(0)\Delta + R_2\;, \quad \|R_2\|_{\op}= O(\|\Delta\|_{\op}^2)\;.
\end{equation*}

\noindent \underline{Analyze $f_k(\Sigma^{1/2})$.} 
First, we apply the matrix square root Taylor expansion to $\Sigma = I_n + \Delta$ to obtain
\begin{align*}
    \Sigma^{1/2} &= I_n + \frac{1}{2}\Delta + R_s\;,\\
    R_s &\coloneqq \frac{1}{2} U S \Lambda^2 U^\top\;,
\end{align*}
where $U, S, \Lambda$ are defined with respect to $M = \Delta$. Using the same logic as in the last step, we can show $\max_i \{|s^{\prime\prime}(\xi_i)|\} = O(1)$. Therefore, $R_s$ can be bounded as
\begin{equation}\label{eq:Rs1}
    \|R_s\|_{\op} = O(\|\Delta\|_{\op}^2) = o(1)\;.
\end{equation}
Moreover, Lemma \ref{lem:op} indicates that
\begin{equation}\label{eq:Rs2}
    \max_{ij} |R_{s,ij}| = O(\|R_s\|_{\op}) = o(1)\;.
\end{equation}

To analyze $f_k(\Sigma^{1/2})$, we observe that
\begin{align*}
    f_k(\Sigma^{1/2}) &= f_k(I_n + \frac{1}{2}\Delta + R_s)\\
    &= f_k(I_n + \mathrm{diag}(R_s)) + f_k(\frac{1}{2}\Delta + R_s - \mathrm{diag}(R_s))\;,
\end{align*}
In the second line, we decompose the matrix into a diagonal matrix $f_k(I_n + \mathrm{diag}(R_s))$ and an off-diagonal matrix $f_k(\frac{1}{2}\Delta + R_s - \mathrm{diag}(R_s))$, which follows from the fact that $f_k$ is an elementwise operation. Note that the second part has diagonal entries equal to zero, since $\Delta$ has zero diagonal values. 

For the diagonal part $f_k(I_n + \mathrm{diag}(R_s))$, we have
\begin{align*}
    \|f_k(I_n + \mathrm{diag}(R_s)) - I_n\|_{\op} = \max_i \{f_k(1+ R_{s,ii}) - 1\}\;.
\end{align*}
Since $\max_{ij} |R_{s,ij}| = o(1)$ (Equation \eqref{eq:Rs2}) and $f_k$ is left continuous at one, we have 
\begin{equation*}
    \max_i \{f_k(1+ R_{s,ii}) - 1\} = o(1)\;.
\end{equation*}
This implies
\begin{equation*}
    f_k(I_n + \mathrm{diag}(R_s)) = I_n + R_D\;,\quad \|R_D\|_{\op} = o(1)\;.
\end{equation*}

For the off-diagonal part, we define $H = \frac{1}{2}\Delta + R_s - \mathrm{diag}(R_s)$. We apply the Taylor expansion \eqref{eq:fk_taylor} for the off-diagonal entries of $f_k(H)$ to obtain
\begin{align*}
    f_k(H) &= f_k^\prime(0) H + \frac{1}{2} F \circ H\circ H\\
    &=\frac{1}{2} f_k^\prime(0) \Delta + f_k^\prime(0)(R_{s} - \mathrm{diag}(R_s)) + \frac{1}{2} F \circ H\circ H\;,\\
    F_{ij} &= f_k^{\prime\prime}(\xi_{ij})\;, \quad |\xi_{ij}|\le |H_{ij}|\;,\quad i\neq j\;,\\
    F_{ii} &= 0\;.
\end{align*}
Noticing that $F\circ H \circ H$, the Hadamard product of the matrices, is symmetric, we apply matrix norm inequality \eqref{eq:inf-norm} to obtain $\|F \circ H\circ H\|_{\op} \le \|F \circ H\circ H\|_\infty$. Moreover, 
\begin{align*}
    \|F \circ H\circ H\|_\infty &= \max_i \sum_{j=1}^n |f_k^{\prime\prime}(\xi_{ij})| H_{ij}^2\;.
\end{align*}
By the smoothness of $f_k$ around zero, we have $\max |f_k^{\prime\prime}(\xi_{ij})| = O(1)$. Therefore, for any $i$, 
\begin{align*}
    \sum_{j=1}^n |f_k^{\prime\prime}(\xi_{ij})| H_{ij}^2 = O(\sum_{j=1}^n H_{ij}^2)\;.
\end{align*}
$\sum_j H_{ij}^2$ is the $i$-th diagonal of $H^2$, and thus $\sum_j H_{ij}^2 \le \|H\|_{\op}^2$. Based on the derivations above, we obtain $\|F \circ H\circ H\|_{\op} = O(\|H\|_{\op}^2)$. Additionally, based on the definition of $H$ and Equations \eqref{eq:Rs1}, \eqref{eq:Rs2}, we have
\begin{align*}
     \|H\|_{\op}^2 &= O(\|\Delta\|_{\op}^2 + \|R_s\|_{\op}^2 + \|\mathrm{diag}(R_s)\|_{\op}^2) = O(\|\Delta\|_{\op}^2 + \|R_s\|_{\op}^2 + \max_{i} R_{s,ii}^2)\\
     &= O(\|\Delta\|_{\op}^2 + \|R_s\|_{\op}^2)=O(\|\Delta\|_{\op}^2 + \|\Delta\|_{\op}^4) = O(\|\Delta\|_{\op}^2)\;.
\end{align*}
To sum up, we show that
\begin{align*}
    f_k(H) &= \frac{1}{2} f_k^\prime(0) \Delta + R_H\;,\\
    R_H &\coloneqq f_k^\prime(0)(R_{s} - \mathrm{diag}(R_s)) + \frac{1}{2} F \circ H\circ H\;,\quad \|R_H\|_{\op} = O(\|\Delta\|_{\op}^2)\;.
\end{align*}
Combining our analysis for diagonal and off-diagonal parts, we obtain our final result
\begin{align*}
    f_k(\Sigma^{1/2}) &= I_n + \frac{1}{2} f_k^\prime(0) \Delta + R_3\;, \\
    R_3 &\coloneqq R_D + R_H\;, \\
    \|R_3\|_{\op} &= o(1) + O(\|\Delta\|_{\op}^2)\;.
\end{align*}

\subsection{Proof of Proposition \ref{prop:asymp_var}}
Before our proof, we introduce the following matrix norm inequality: For a diagonal matrix $A$ with positive diagonals and $P\succeq 0$, we have 
\begin{equation}\label{eq:diag_psd}
    \min_i |A_{ii}|\tr(P) \le \tr(PA) \le \max_i |A_{ii}| \tr(P)\;.
\end{equation}
We will use this multiple times throughout the proof.

Note that given $Y(k) = X\beta_k$, we can write 
\begin{equation*}
    V(\Sigma) = \frac{K-1}{n} Y(k)^\top f_k(\Sigma) Y(k) = \frac{K-1}{n}\beta_k^\top X^\top f_k(\Sigma) X\beta_k\;.
\end{equation*}
Under the assumption that $\beta_k$ has zero mean and identity covariance, we have
\begin{align*}
    \E_{\beta_k} V(\Sigma) &= \frac{K-1}{n} \E \tr(X^\top f_k(\Sigma) X\beta_k \beta_k^\top )\\
    &= \frac{K-1}{n}\tr(X^\top f_k(\Sigma) X\E (\beta_k \beta_k^\top)) = \frac{K-1}{n} \tr(f_k(\Sigma) X X^\top)\;.
\end{align*}
Note that this holds for any $\Sigma\in\cE$. 

We write $\Delta = \Sigma - I_n$, which is a symmetric matrix with zero diagonal values. Then we have
\begin{equation}\label{eq:var_diff}
   \begin{aligned}
    \frac{n}{K-1} (\E_{\beta_k} V(I_n) - \E_{\beta_k} V(\Sigma)) &= \tr((f_k(I_n) - f_k(I_n + \Delta)) XX^\top) \\
    &= \tr((I_n - f_k(I_n + \Delta)) (XX^\top - I_n))\;.
\end{aligned}
\end{equation}
The last equality follows from the fact that $f_k(I_n) - f_k(I_n+\Delta)$ has zero diagonals. Additionally, we have assumed $f_k(1) = 1$ so that $f_k(I_n) = I_n$, as explained in the proof of Theorem \ref{thm:normality}.

By the elementwise Taylor expansion on $f_k$ (introduced in the proof of Lemma \ref{lem:matrix_expansion}), we have
\begin{align*}
    f_k(I_n + \Delta) = I_n + f_k^\prime(0) \Delta + \frac{1}{2} R\;, \quad R_{ij} = f_k^{\prime\prime}(\xi_{ij})\Delta_{ij}^2\;,
\end{align*}
where $\xi_{ij}$ satisfies $|\xi_{ij}|\le |\Delta_{ij}|$. Then we apply the Taylor expansion to \eqref{eq:var_diff}
and obtain
\begin{align*}
    \frac{n}{K-1}(\E_{\beta_k} V(I_n) - \E_{\beta_k} V(\Sigma)) &= \tr((-f_k^\prime(0) \Delta - \frac{1}{2} R) (XX^\top - I_n))\\
    &= -f_k^\prime(0) \underbrace{\tr(\Delta (XX^\top - I_n))}_{\eqqcolon \Delta_f} - \frac{1}{2} \underbrace{\tr( R(XX^\top - I_n))}_{\eqqcolon \Delta_R}\;.
\end{align*}
Next, we specify $\Sigma$ to be the solution $\Sigma_{\eta}$ from the one-step PGD-Gauss, and perform matrix analysis to bound $\Delta_f$ and $\Delta_R$, respectively. Under the assumption $f_k^\prime(0) \neq 0$, we either have $f_k^\prime(0) > 0$ or $f_k^\prime(0) < 0$. From now on, we assume that $f_k^\prime(0) > 0$ without loss of generality.

\noindent\underline{Step 1. Analyze $\Delta_f$. }
By definition of PGD-Gauss, we have
\begin{align*}
    \Sigma_{\eta} &= D^{-1} (I_n - \eta f_k^\prime(0) (XX^\top  - I_n))^2 D^{-1}\\
    &= D^{-1} (I_n - 2\eta f_k^\prime(0) (XX^\top  - I_n) + \eta^2 (f_k^\prime(0))^2 (XX^\top  - I_n)^2 ) D^{-1}\\
    &= D^{-2} - 2\eta f_k^\prime(0) D^{-1}(XX^\top  - I_n) D^{-1} + \eta^2 (f_k^\prime(0))^2 D^{-1}(XX^\top  - I_n)^2 D^{-1}\;.
\end{align*}
This implies 
\begin{equation}\label{eq:pgd_delta}
    \Delta = \Sigma_{\eta} - I_n = (D^{-2} - I_n) - 2\eta f_k^\prime(0) D^{-1}(XX^\top  - I_n) D^{-1} + \eta^2 (f_k^\prime(0))^2 D^{-1}(XX^\top  - I_n)^2 D^{-1}\;.
\end{equation}
Therefore, the difference can be decomposed as
\begin{equation}\label{eq:Delta_f}
\begin{aligned}
    \Delta_f &= \underbrace{\tr((D^{-2} - I_n)(XX^\top - I_n))}_\text{(I)} - \underbrace{2\eta f_k^\prime(0) \tr(D^{-1}(XX^\top  - I_n) D^{-1}(XX^\top - I_n))}_\text{(II)}\\
    &+ \underbrace{\eta^2 (f_k^\prime(0))^2 \tr(D^{-1}(XX^\top  - I_n)^2 D^{-1}(XX^\top - I_n))}_\text{(III)}\;.
\end{aligned}
\end{equation}
It is easy to verify that (I) is zero, since $D^{-2} - I_n$ is a diagonal matrix and $\mathrm{diag}(XX^\top - I_n) = 0$ under Assumption \ref{asmp:stepsize}. For (II), we apply the inequality \eqref{eq:diag_psd} with $A = D^{-1}$ and $P = (XX^\top  - I_n) D^{-1}(XX^\top - I_n)$ to obtain 
\begin{align*}
    \text{(II)} &\ge \frac{2\eta f_k^\prime(0)}{\max_i D_{ii}}\tr((XX^\top  - I_n) D^{-1}(XX^\top - I_n))\\
    &= \frac{2\eta f_k^\prime(0)}{\max_i D_{ii}} \tr(D^{-1}(XX^\top - I_n)^2)\;.
\end{align*}
Again, we apply \eqref{eq:diag_psd} with $A = D^{-1}$ and $P = (XX^\top - I_n)^2$ to obtain
\begin{equation}\label{eq:bound_II}
    \text{(II)} \ge \frac{2\eta f_k^\prime(0)}{\max_i D_{ii}^2} \tr((XX^\top - I_n)^2) = \frac{2\eta f_k^\prime(0)}{\max_i D_{ii}^2} \|XX^\top - I_n\|_F^2\;.
\end{equation}

For (III), we define the eigenvalue decomposition $XX^\top = U \Lambda U^\top$ and write
\begin{align*}
\tr(D^{-1}(XX^\top  - I_n)^2 D^{-1}(XX^\top - I_n)) &= \tr(D^{-1}U(\Lambda  - I_n)^2 U^\top D^{-1} U(\Lambda - I_n) U^\top) \\
&= \tr(U^\top D^{-1}U(\Lambda  - I_n)^2 U^\top D^{-1} U(\Lambda - I_n))
\end{align*}
We apply inequality \eqref{eq:diag_psd} with $P = U^\top D^{-1}U(\Lambda  - I_n)^2 U^\top D^{-1} U$ and $A = \Lambda - I_n$ to obtain
\begin{align*}
\tr(D^{-1}(XX^\top  - I_n)^2 D^{-1}(XX^\top - I_n)) 
&\le \max_{i} |\Lambda_i - 1| \tr(U^\top D^{-1}U(\Lambda  - I_n)^2 U^\top D^{-1} U) \\
&= \max_{i} |\Lambda_i - 1| \tr(D^{-2}U(\Lambda  - I_n)^2 U^\top)\\
&= \max_{i} |\Lambda_i - 1| \tr(D^{-2}(XX^\top  - I_n)^2)\;.
\end{align*}
Again, we apply \eqref{eq:diag_psd} with $A = D^{-2}$ and $P = (XX^\top  - I_n)^2$ to obtain
\begin{align}
\tr(D^{-1}(XX^\top  - I_n)^2 D^{-1}(XX^\top - I_n)) & \le \frac{\max_{i} |\Lambda_i - 1|}{\min_i \{D_{ii}^2\}} \|XX^\top  - I_n\|_F^2\nonumber\\
&= \frac{\|XX^\top-I_n\|_{\op}}{\min_i \{D_{ii}^2\}} \|XX^\top  - I_n\|_F^2\;.\label{eq:bound_III}
\end{align}

Since we have analyzed (I), (II), (III), we apply our bounds \eqref{eq:bound_II}, \eqref{eq:bound_III} to \eqref{eq:Delta_f} and obtain
\begin{align*}
    \Delta_f &\le -\frac{2\eta f_k^\prime(0)}{\max_i D_{ii}^2}\|XX^\top - I_n\|_F^2 + \eta^2 (f_k^\prime(0))^2 \frac{\|XX^\top-I_n\|_{\op}}{\min_i \{D_{ii}^2\}} \|XX^\top  - I_n\|_F^2\\
    &= \left(- \frac{2\eta f_k^\prime(0)}{\max_i \{D_{ii}^2\}} + \eta^2 (f_k^\prime(0))^2 \frac{\|XX^\top-I_n\|_{\op}}{\min_i \{D_{ii}^2\}}\right) \|XX^\top  - I_n\|_F^2\;.
\end{align*}

\noindent\underline{Step 2. Analyze $\Delta_R$. }
By definition of $R$, we have $R_{ii} = 0$ for any $i$. Hence we have
\begin{align*}
\Delta_R = \sum_{i\neq j} R_{ij} G_{ij} = \sum_{i\neq j} f_k^{\prime\prime}(\xi_{ij}) \Delta_{ij}^2 G_{ij}\;.
\end{align*}
where $G_{ij} = X_i^\top X_j$. Under Assumption \ref{asmp:stepsize} and the analysis in the proof of Theorem \ref{thm:normality}, we have $\max |\Delta_{ij}| = o(1)$. In addition, since $f_k$ is smooth around zero, we have $\max|f_k^{\prime\prime}(\xi_{ij})| = O(1)$. Therefore, 
\begin{align*}
|\Delta_R| &= O\left(\left|\sum_{i\neq j} \Delta_{ij}^2 G_{ij}\right|\right)\\
&\stackrel{\text{(i)}}{=} O\left(\sum_{i,j=1}^n \Delta_{ij}^2\right) = O\left(\|\Delta\|_F^2\right)\;.
\end{align*}
where (i) follows from $|G_{ij}|\le 1$ and $\Delta_{ii} = 0$. 

Now we give a concrete bound on $\|\Delta\|_F$ for the one-step PGD-Gauss. From \eqref{eq:pgd_delta} and the triangle inequality, we have
\begin{align*}
    \|\Delta\|_F &= \| (D^{-2} - I_n) - 2\eta f_k^\prime(0) D^{-1}(XX^\top  - I_n) D^{-1} + \eta^2 (f_k^\prime(0))^2 D^{-1}(XX^\top  - I_n)^2 D^{-1} \|_F\\
    &\le \|D^{-2} - I_n\|_F + 2\eta |f_k^\prime(0)| \|D^{-1}(XX^\top  - I_n) D^{-1}\|_F + \eta^2 (f_k^\prime(0))^2 \|D^{-1}(XX^\top  - I_n)^2 D^{-1} \|_F\;.
\end{align*}
Given Equation \eqref{eq:Dii} and Assumption \ref{asmp:stepsize}, we have $D_{ii}^2 = 1 + O(\eta^2 \|XX^\top-I_n\|_{\op}^2) = 1+o(1)$, and hence
\begin{align*}
\|D^{-2} - I_n\|_F^2 &= \sum_{i=1}^n (1/D_{ii}^2 -1)^2 = \sum_{i=1}^n \left(\frac{1}{1+\eta^2 \|XX^\top-I_n\|_{\op}^2} -1\right)^2 \\
&\le \sum_{i=1}^n \left(\eta^2 \|XX^\top-I_n\|_{\op}^2\right)^2 = n \eta^4 \|XX^\top-I_n\|_{\op}^4\;,
\end{align*}
where the last inequality follows from $1 - 1/(1+x) \le x$, for $x$ close to zero. Next, since Equation \eqref{eq:Dii} implies $\max_i |D_{ii} - 1| = o(1)$, we have 
\begin{equation*}
\|D^{-1}(XX^\top  - I_n) D^{-1}\|_F = O(\|XX^\top  - I_n\|_F)\;.
\end{equation*}
Lastly, observe that
\begin{align*}
    \|D^{-1}(XX^\top  - I_n)^2 D^{-1} \|_F &= O(\|(XX^\top  - I_n)^2\|_F) \\
    &= O(\|(XX^\top  - I_n)\|_{\op}\|(XX^\top  - I_n)\|_F)\;,
\end{align*}
where the last equality follows from $\|AB\|_F \le \|A\|_\op \|B\|_F$ for compatible matrices $A$, $B$. Combining the results above, we obtain
\begin{align*}
    \|\Delta\|_F &=O( \sqrt{n} \eta^2 \|XX^\top-I_n\|_{\op}^2) + O(\eta\|XX^\top  - I_n\|_F) + O(\eta^2 \|XX^\top  - I_n\|_{\op}\|XX^\top  - I_n\|_F)\;,\\
    |\Delta_R| &= O(\|\Delta\|_F^2) = O(n \eta^4 \|XX^\top-I_n\|_{\op}^4) + O(\eta^2 \|XX^\top  - I_n\|_F^2) + O(\eta^4 \|XX^\top  - I_n\|_{\op}^2\|XX^\top  - I_n\|_F^2)\;.
\end{align*}

\noindent\underline{Step 3. Derive final results. }
Based on Step 1 and 2, we have
\begin{equation}\label{eq:V_lb}
   \begin{aligned}
    \frac{n}{K-1}(\E_{\beta_k} V(I_n) - \E_{\beta_k} V(\Sigma_{\eta})) &= -f_k^\prime(0) \Delta_f - \frac{1}{2} \Delta_R\\
    &\ge f_k^\prime(0) \left(\frac{2\eta f_k^\prime(0)}{\max_i \{D_{ii}^2\}} - \eta^2 (f_k^\prime(0))^2 \frac{\|XX^\top - I_n\|_{\op}}{\min_i \{D_{ii}^2\}}\right) \|XX^\top  -  I_n\|_F^2 - \frac{1}{2} \Delta_R\;.
\end{aligned} 
\end{equation}

Given Assumption \ref{asmp:stepsize} and \eqref{eq:Dii}, we have 
\begin{equation*}
    \max_i \{D_{ii}^2\} \asymp 1\;,\quad \min_i \{D_{ii}^2\} \asymp 1\;.
\end{equation*}
Thus, we have 
\begin{align*}
    &f_k^\prime(0) \left(\frac{2\eta f_k^\prime(0)}{\max_i \{D_{ii}^2\}} - \eta^2 (f_k^\prime(0))^2 \frac{\|XX^\top - I_n\|_{\op}}{\min_i \{D_{ii}^2\}}\right) \|XX^\top  -  I_n\|_F^2 \\
    &= \Omega\left( f_k^\prime(0) \left(2\eta f_k^\prime(0) - \eta^2 (f_k^\prime(0))^2 \|XX^\top-I_n\|_{\op}\right) \|XX^\top  - I_n\|_F^2 \right)\;.
\end{align*}
By Assumption \ref{asmp:stepsize}, it holds that $\eta \|XX^\top-I_n\|_{\op} = o(1)$. Therefore, the term $f_k^\prime(0)\eta\|XX^\top - I_n\|_{\op}$ above is of order $o(1)$. Based on this observation, we further derive
\begin{align*}
f_k^\prime(0) \left(2\eta f_k^\prime(0) - \eta^2 (f_k^\prime(0))^2 \|XX^\top-I_n\|_{\op}\right) \|XX^\top  - I_n\|_F^2 = \Omega \left(\eta (f_k^\prime(0))^2 \|XX^\top  - I_n\|_F^2 \right)\;.
\end{align*}

Next we show that $\Delta_R$ term is negligible compared to the first term in \eqref{eq:V_lb}. Based on the analysis in Step 2, we have
\begin{equation*}
    |\Delta_R| = O(\|\Delta\|_F^2) = O(n \eta^4 \|XX^\top-I_n\|_{\op}^4) + O(\eta^2 \|XX^\top  - I_n\|_F^2) + O(\eta^4 \|XX^\top  - I_n\|_{\op}^2\|XX^\top  - I_n\|_F^2)\;.
\end{equation*}
and
\begin{align*}
    O(\eta^4 \|XX^\top  - I_n\|_{\op}^2\|XX^\top  - I_n\|_F^2) &\stackrel{\text{(i)}}{=} o(\eta^2\|XX^\top  - I_n\|_F^2)\;, \\
    O(\eta^2 \|XX^\top  - I_n\|_F^2) &\stackrel{\text{(ii)}}{=} o(\eta \|XX^\top  - I_n\|_F^2)\;.
\end{align*}
In the derivation above, (i) follows from Assumption \ref{asmp:stepsize} and (ii) follows from $\eta = o(1)$. Under the additional assumption that $n\eta^3 \|XX^\top - I_n\|_{\op}^4 = o(\|XX^\top - I_n\|_F^2)$, we obtain
\begin{equation*}
    \|\Delta\|_F^2 = o(\eta \|XX^\top  - I_n\|_F^2)\;.
\end{equation*}
Therefore,
\begin{equation*}
    \frac{n}{K-1}(\E_{\beta_k} V(I_n) - \E_{\beta_k} V(\Sigma_{\eta})) = \Omega(\eta \|XX^\top  - I_n\|_F^2)\;.
\end{equation*}
This completes the proof.

\section{Covariate-adaptive Designs and Covariate Adjustments}\label{sec:cov_adjust}
Our work has focused on optimizing the MSE property of Horvitz-Thompson estimators. In practice, researchers in the analysis stage could utilize more advanced estimators with covariate adjustments. For instance, Lin's estimator \citep{Lin2013} is widely used under the binary treatment setting, which is defined as the coefficients of $D_i$ in the OLS regression
\begin{equation*}
    Y_i \sim X_i + D_i + D_i (X_i - \bar{X})\;,\quad \bar{X} = \frac{1}{n}\sum_{i=1}^n X_i\;.
\end{equation*}
\cite{li2017general} show the ``optimality'' of Lin’s estimator
within a class of regression-adjusted estimators. Here we clarify two practical questions related to covariate adjustments and Gaussianization: 
\begin{enumerate}
    \item Why do researchers consider covariate-adaptive designs? Alternatively, one could simply use covariate adjustments under complete randomization.
    \item Suppose we want to balance for covariates in the design stage. What if we formulate covariate balance measures with respect to covariate-adjusted estimators?
\end{enumerate}

First, in real-world randomized experiments, designers and analyzers may not communicate or share the same set of covariates. Therefore, balancing for covariates in the design stage is desirable, as it improves the estimation precision even with the simple Horvitz-Thompson estimator. Moreover, studies have shown that covariate-adaptive designs ---such as rerandomiztion--- never hurt the estimation \citep{li2020jrssb}. We anticipate similar results hold under Gaussianization. That is, we view covariate-adaptive designs and covariate adjustments as synergistic approaches that can be combined to further enhance the estimation.

Second, under Gaussianization, it is possible to analyze the MSE property of covariate-adjusted estimators as in \cite{Chang2023designbased} and formulate covariate balance measures tailored to these estimators. However, we argue that design optimization is more robust toward different outcome-generating models when using model-agnostic estimators, e.g., Horvitz-Thompson estimators. In other words, under a misspecified model, the MSE of covariate-adjusted estimator will also be misspecified, and the estimator itself might be biased. Hence, conducting design optimization based on biased adjustments could impair estimation precision.
\end{document}